\documentclass{article}

\usepackage{amsmath, amssymb, amsthm, amsfonts, mathtools, mathrsfs}

\usepackage{enumitem}

\usepackage{dsfont}

\usepackage[pdftex]{graphicx}

\usepackage[margin=3cm]{geometry}

\usepackage{cite}

\usepackage{fancyhdr}

\usepackage{setspace}

\usepackage{tikz}
\usetikzlibrary{calc,intersections}
\usetikzlibrary{shapes}

\usepackage{float}
\restylefloat{table}
\restylefloat{figure}

\usepackage{appendix}

\usepackage[colorinlistoftodos]{todonotes}

\usepackage[symbol]{footmisc}

\numberwithin{equation}{section}

\theoremstyle{plain}
\newtheorem{theorem}{Theorem}
\newtheorem{proposition}[theorem]{Proposition}
\newtheorem{lemma}[theorem]{Lemma}

\newtheorem{definition}[theorem]{Definition}

\newtheorem{remark}[theorem]{Remark}

\usepackage[utf8]{inputenc} 
\usepackage[T1]{fontenc}    
\usepackage{hyperref}       
\usepackage{url}            
\usepackage{booktabs}       
\usepackage{amsfonts}       
\usepackage{nicefrac}       
\usepackage{microtype}      

\author{%
  Johannes Forkel\footnote{\texttt{johannes.forkel@maths.ox.ac.uk}, Mathematical Institute, University of Oxford, Oxford, OX2 6GG, United Kingdom} and Jonathan P. Keating\footnote{\texttt{jon.keating@maths.ox.ac.uk}, Mathematical Institute, University of Oxford, Oxford, OX2 6GG, United Kingdom}  \\
}

\begin{document}

\title{The Classical Compact Groups and Gaussian Multiplicative Chaos}

\maketitle 

\begin{abstract} 
\noindent We consider powers of the absolute value of the characteristic polynomial of Haar distributed random orthogonal or symplectic matrices, as well as powers of the exponential of its argument, as a random measure on the unit circle. We also consider the case where these measures are restricted to the unit circle minus small neighborhoods around $\pm 1$. We show that for small enough powers and under suitable normalization, as the matrix size goes to infinity, these random measures converge in distribution to a Gaussian multiplicative chaos measure. Our result is analogous to one relating to unitary matrices previously established by Christian Webb in \cite{Webb}. We thus complete the connection between the classical compact groups and Gaussian multiplicative chaos. To prove this convergence when excluding small neighborhoods around $\pm 1$ we establish appropriate asymptotic formulae for Toeplitz and Toeplitz+Hankel determinants with merging singularities. Using a recent formula due to Claeys {\it et al.} \cite{Claeys et al}, we are able to prove convergence on the whole of the unit circle.
\end{abstract}

\tableofcontents

\section{Introduction}
In \cite{HughesKeatingO'Connell}, Hughes, Keating and O'Connell proved that the real and the imaginary part of the logarithm of the characteristic polynomial of a random unitary matrix convergence jointly to a pair of Gaussian fields on the unit circle. Using this result Webb established in \cite{Webb} a connection between random matrix theory and Gaussian multiplicative chaos (GMC), a theory developed first by Kahane in the context of turbulence in \cite{Kahane} (see \cite{RhodesVargas} for a review). Webb proved that powers of the exponential of the real and imaginary part of the logarithm of the characteristic polynomial of a random unitary matrix converge, when suitably normalized, to Gaussian multiplicative chaos measures on the unit circle. This was achieved using results on Toeplitz determinants with merging Fisher-Hartwig singularities due to Claeys and Krasovsky in \cite{ClaeysKrasovsky}. In \cite{Webb} Webb proved the result only in the so-called $L^2$-phase, that is those powers for which the second moment of the total mass of the limiting GMC measure exists. In \cite{NikulaSaksmanWebb} the result was extended to the whole $L^1$- or subcritical phase, i.e. was proven to also hold for the (larger) set of powers for which the limiting GMC measure is non-trivial.\\ 
Since then the connection between the two fields has been extended to other random matrix ensembles. In \cite{ChhaibiNajnudel}, Chhaibi and Najnudel proved convergence (in a different sense) of the characteristic polynomials of matrices drawn from the Circular Beta Ensemble to a GMC measure on the unit circle. In \cite{BerestyckiWebbWong} Berestycki, Webb and Wong proved that, after suitable normalization, powers of the absolute value of the characteristic polynomial of a matrix from the Gaussian Unitary Ensemble converge to Gaussian multiplicative chaos measures on the real line. They proved this result in the $L^2$-phase, however it is likely to also hold in the whole $L^1$-phase. The analogous result for powers of the exponential of the imaginary part of the logarithm of the characteristic polynomial of the Gaussian Unitary Ensemble was proven in \cite{ClaeysFahsLambertWebb}, in the whole $L^1$-phase.\\
The connection with GMC is closely related to recent developments concerning the extreme value statistics of the characteristic polynomials of random matrices and the associated theory of moments of moments \cite{ArguinBeliusBourgade, Arguin et al, AssiotisBaileyKeating, AssiotisKeating, BaileyKeating, ChhaibiMadauleNajnudel, FyodorovGnutzmannKeating, FyodorovHiaryKeating, FyodorovKeating, PaquetteZeitouni}. It also has interesting applications to spectral statistics; for example, it implies strong rigidity estimates for the eigenvalues \cite{ClaeysFahsLambertWebb}.\\

Our purpose here it to extend Webb's result to the other classical compact groups, i.e. to the orthogonal and symplectic groups.  Our starting point is a theorem due  to Assiotis and Keating concerning the convergence of the real and imaginary parts of the logarithm of the characteristic polynomials of random orthogonal or symplectic matrices to a pair of Gaussian fields on the unit circle \footnote{This result has not previously been published.  With the kind agreement of Dr. Assiotis, we set out the theorem and its proof in Appendix \ref{appendix:Gaussian fields}}. This is the analogous result to the one for random matrices in \cite{HughesKeatingO'Connell}. We then complete the connection between the classical compact groups and Gaussian multiplicative chaos, by showing that for the orthogonal and symplectic groups we get statements similar to the one Webb proved for the unitary group. Using the same approach as in \cite{Webb} we prove our results only in the $L^2$-phase, i.e. when the limiting GMC measure's total mass has a finite second moment. We believe that using the techniques in \cite{NikulaSaksmanWebb} one can extend our results to hold more generally in the whole $L^1$-phase, however extra care is needed since the covariance function of the underlying Gaussian field has singularities not just on the diagonal but also on the antidiagonal; in particular at $\pm 1$ the field has special behaviour. \\
In order to prove convergence to the GMC measure after restricting all involved measures to $(\epsilon, \pi - \epsilon) \cup (\pi + \epsilon, 2\pi - \epsilon)$, i.e. when excluding small neighborhoods around $\pm 1$, we computed the uniform asymptotics of Toeplitz and Toeplitz+Hankel determinants with two pairs of merging singularities which are all bounded away from $\pm 1$. Our results on these asymptotics are similar to those in \cite{ClaeysKrasovsky} and \cite{DeiftItsKrasovsky}, and the proof techniques we employ are strongly influenced by these two papers.

To prove convergence on the full unit circle we need also to know the uniform asymptotics of Toeplitz and Toeplitz+Hankel determinants with 3 or 5 singularities merging at $\pm 1$. Claeys, Glesner, Minakov and Yang have recently proved a formula for the uniform asymptotics of Toeplitz+Hankel determinants with arbitrarily many merging singularities, up to a multiplicative constant \cite{Claeys et al}.  Using their formula allows us to extend our analysis to around $\pm 1$, and so to cover the full unit circle, however for a slightly smaller set of powers than when neighborhoods around $\pm 1$ are excluded.

\section{Statement of Main Result and Strategy of Proof}
Denote by $O(n)$ the group of orthogonal $n\times n$ matrices, and by $Sp(2n)$ the group of $2n\times 2n$ symplectic matrices, i.e. unitary $2n\times 2n$ matrices that additionally satisfy
\begin{equation}
UJU^T= U^TJU = J,
\end{equation}
where 
\begin{equation}
J:= \left( \begin{array}{cc} 0 & I_n \\ -I_n & 0 \end{array} \right).
\end{equation}
The characteristic polynomial 
\begin{equation}
p_n(\theta) = \text{det}\left(I_n-e^{-i\theta}U_n\right) = \prod_{k=1}^n (1-e^{i(\theta_k-\theta)})
\end{equation}
of $U_n$ in $O(n)$ or $Sp(2n)$ (then we have instead $I_{2n}$ and the product is up to $2n$) is taken as a function on the unit circle, where all its zeroes lie.

\begin{definition} For $n \in \mathbb{N}$, $\alpha \in \mathbb{R}$, $\beta \in i\mathbb{R}$ and $\theta \in [0,2\pi)$ let 
\begin{equation}
f_{n,\alpha,\beta}(\theta) = |p_n(\theta)|^{2\alpha} e^{2i\beta \Im \ln p_n(\theta)},
\end{equation}
where (with the sum being up to $2n$ for $U_n \in Sp(2n)$)
\begin{equation}
\Im \ln p_n(\theta):= \sum_{k = 1}^n \Im \ln (1-e^{i(\theta_k -\theta)}),
\end{equation} 
with the branches on the RHS being the principal branches, such that 
\begin{equation}
\Im \ln (1-e^{i(\theta_l - \theta)}) \in \left( - \frac{\pi}{2}, \frac{\pi}{2} \right],
\end{equation}
where $\Im \ln 0 := \pi/2$. Further we define the random Radon measures $\mu_{n,\alpha,\beta}$ on $S^1 \sim [0,2\pi)$ by 
\begin{equation}
\mu_{n,\alpha,\beta}(\text{d}\theta) = \frac{f_{n,\alpha,\beta}(\theta) }{\mathbb{E}\left(f_{n,\alpha,\beta}(\theta)\right)}\text{d}\theta.
\end{equation}
\end{definition}

Let $\left(\mathcal{N}_j\right)_{j \in \mathbb{N}}$ be a sequence of independent standard (real) normal random variables and denote 
\begin{equation}
\eta_j := 1_{j \text{ is even}}.
\end{equation}
We recall the following result from \cite{DiaconisEvans, DiaconisShahshahani}:

\begin{theorem} \label{thm:traces} (Diaconis and Shahshahani, Diaconis and Evans)
If $U_n$ is Haar distributed  on $O(n)$ we have for any fixed $k$:
\begin{equation}
\left( \text{Tr}(U_n), \frac{1}{\sqrt{2}} \text{Tr}(U_n^2), ... , \frac{1}{\sqrt{k}} \text{Tr}(U_n^k) \right) \xrightarrow[n\rightarrow \infty]{d} \left( \mathcal{N}_1 + \eta_1, \mathcal{N}_2 + \frac{\eta_2}{\sqrt{2}}, ... , \mathcal{N}_k + \frac{\eta_k}{\sqrt{k}} \right).
\end{equation}
Similarly, if $U_n$ is Haar distributed on $Sp(2n)$, we have for any fixed $k \in \mathbb{N}$:
\begin{equation}
\left( \text{Tr}(U_n), \frac{1}{\sqrt{2}} \text{Tr}(U_n^2), ... , \frac{1}{\sqrt{k}} \text{Tr}(U_n^k) \right) \xrightarrow[n\rightarrow \infty]{d} \left( \mathcal{N}_1 - \eta_1, \mathcal{N}_2 - \frac{\eta_2}{\sqrt{2}}, ... , \mathcal{N}_k - \frac{\eta_k}{\sqrt{k}} \right).
\end{equation}
Finally we have the bound
\begin{equation} \label{eqn:trace bound}
\mathbb{E}_{U_n}\left( \left( \text{Tr}(U_n^k)\right)^2 \right) \leq \text{const} \min\{k,n\},
\end{equation}
where const is independent of $k$ and $n$. 
\end{theorem}

Using this result, Assiotis and Keating have proved the following theorem\footnote{This result has not previously been published, and so with the kind agreement of Dr. Assiotis we set out the proof in Appendix \ref{appendix:Gaussian fields}.}: 

\begin{theorem} \label{thm:Gaussian field} (Assiotis, Keating)
Let $p_n$ be the characteristic polynomial of a random $U_n \in O(n)$, w.r.t. Haar measure. Then for any $\epsilon > 0$ the pair of fields $\left( \Re \ln p_n, \Im \ln p_n \right)$ converges in distribution in $H^{-\epsilon}_0 \times H^{-\epsilon}_0$ to the pair of Gaussian fields $\left( X - x, \hat{X} - \hat{x} \right)$, where
\begin{align} 
\begin{split}
X(\theta) &= \frac{1}{2} \sum_{j = 1}^\infty \frac{1}{\sqrt{j}} \mathcal{N}_{j} \left(e^{-ij\theta} + e^{ij\theta}\right) = \sum_{j=1}^\infty \frac{1}{\sqrt{j}} \mathcal{N}_{j} \cos(j\theta),\\
\hat{X}(\theta) &= \frac{1}{2i} \sum_{j = 1}^\infty \frac{1}{\sqrt{j}} \mathcal{N}_{j} \left(e^{-ij\theta} - e^{ij\theta}\right) = -\sum_{j=1}^\infty \frac{1}{\sqrt{j}} \mathcal{N}_{j} \sin(j\theta),\\
x(\theta) &= \frac{1}{2} \sum_{j = 1}^\infty \frac{\eta_j}{j} \left(e^{-ij\theta} + e^{ij\theta}\right) = \sum_{j=1}^\infty \frac{\eta_j}{j} \cos(j\theta),\\
\hat{x}(\theta) &= \frac{1}{2i} \sum_{j = 1}^\infty \frac{\eta_j}{j} \left(e^{-ij\theta} - e^{ij\theta}\right) = - \sum_{j=1}^\infty \frac{\eta_j}{j} \sin(j\theta).
\end{split}
\end{align}
Similarly, for $U_n \in Sp(2n)$ and any $\epsilon > 0$, the pair of fields $\left( \Re \ln p_n, \Im \ln p_n \right)$ converges in distribution in $H^{-\epsilon}_0 \times H^{-\epsilon}_0$ to the pair of Gaussian fields $\left( X + x, \hat{X} + \hat{x} \right)$.
\end{theorem}
\noindent The spaces $H^{-\epsilon}_0$ are defined in Appendix \ref{appendix:Gaussian fields} as certain closed subspaces of the negative Sobolev spaces $H^{-\epsilon}$.
 
\begin{remark} Formally one has 
\begin{align} \label{eqn:cov 1}
\begin{split}
\mathbb{E}\left( X(\theta)X(\theta') \right) &= \sum_{j=1}^\infty \frac{\cos(j\theta)\cos(j\theta')}{j}\\ &= 
\frac{1}{2} \left( \sum_{j=1}^\infty \frac{\cos(j(\theta+\theta'))}{j} + \sum_{j=1}^\infty \frac{\cos(j(\theta-\theta'))}{j} \right) \\
&= \frac{1}{4} \sum_{j=1}^\infty \frac{1}{j} \left( e^{ij(\theta+\theta')} + e^{-ij(\theta+\theta')} + e^{ij(\theta-\theta')} + e^{-ij(\theta - \theta')} \right) \\
&= - \frac{1}{2} \left( \ln |e^{i\theta} - e^{i\theta'}| + \ln |e^{i\theta} - e^{-i\theta'}| \right).
\end{split}
\end{align}
Similarly one has formally 
\begin{align} \label{eqn:cov 2}
\begin{split}
\mathbb{E}\left( \hat{X}(\theta)\hat{X}(\theta') \right) =& - \frac{1}{2} \left( \ln |e^{i\theta} - e^{i\theta')}| - \ln |e^{i\theta} - e^{-i\theta'}| \right),\\
\mathbb{E}\left( X(\theta)\hat{X}(\theta') \right) =& \frac{1}{2} \left( \Im \ln (1-e^{i(\theta+\theta')}) - \Im \ln (1-e^{i(\theta-\theta')})\right).
\end{split}
\end{align}
\end{remark} 

For $\alpha \in \mathbb{R}$ and $\beta \in i\mathbb{R}$ we define the field
\begin{equation} \label{eqn:Y}
Y_{\alpha,\beta}(\theta) = 2\alpha X(\theta) + 2i\beta \hat{X}(\theta).
\end{equation}
By (\ref{eqn:cov 1}) and \ref{eqn:cov 2}) we see that its covariance function is formally given by 
\begin{align} \label{eqn:Cov Y}
\begin{split}
\text{Cov}(Y_{\alpha,\beta} (\theta),Y_{\alpha,\beta} (\theta')) =& -2(\alpha^2-\beta^2) \ln |e^{i\theta}-e^{i\theta'}| - 2(\alpha^2+\beta^2) \ln |e^{i\theta}-e^{-i\theta'}| \\
& +4i\alpha \beta \Im \ln (1-e^{i(\theta+\theta')}).
\end{split}
\end{align}
Motivated by Theorem \ref{thm:Gaussian field} one expects that $\mu_{n,\alpha,\beta}$ behaves like $e^{Y_{\alpha,\beta}}$ for large $n$. Even though the covariance function of $Y_{\alpha,\beta}$ has logarithmic singularities, not only on the diagonal $\theta = \theta'$ but also on the anti-diagonal $\theta = - \theta'$, one can still construct a corresponding non-trivial Gaussian multiplicative chaos measure $\mu_{\alpha,\beta}$, which can formally be written as
\begin{align}
\mu_{\alpha,\beta}(\text{d}\theta) = \frac{ e^{Y_{\alpha,\beta}(\theta)} } {\mathbb{E}(e^{Y_{\alpha,\beta}(\theta)})} \text{d}\theta = e^{Y_{\alpha,\beta}(\theta) - \frac{1}{2}\mathbb{E}(Y_{\alpha,\beta}(\theta)^2)}\text{d}\theta.
\end{align} 
$\mu_{\alpha,\beta}$ is properly defined in Appendix \ref{appendix:GMC} as the almost sure limit in distribution of certain random Radon measures $\left(\mu^{(k)}_{\alpha,\beta} \right)_{k \in \mathbb{N}}$.\\

For $\epsilon \in (0,\pi/2)$ define $I_{\epsilon}:= (\epsilon, \pi - \epsilon) \cup (\pi + \epsilon, 2\pi - \epsilon)$. Then our first main result is the following:
\begin{theorem} \label{thm:main}
Let $\alpha^2 - \beta^2 < 1/2$ and $\alpha > - 1/4$. When restricting the random measures $\mu_{n,\alpha,\beta}$ and $\mu_{\alpha,\beta}$ to $I_\epsilon$, then as $n \rightarrow \infty$, for any fixed $\epsilon>0$, the sequence $\left(\mu_{n,\alpha,\beta}\right)_{n \in \mathbb{N}}$ converges weakly to $\mu_{\alpha,\beta}$ in the space of Radon measures on $I_\epsilon$ equipped with the topology of weak convergence, i.e. for any $F:\left\{\text{Radon measures on } I_\epsilon \right\} \rightarrow \mathbb{R}$ for which $F(\mu_n) \rightarrow F(\mu)$ whenever $\mu_n \xrightarrow{d} \mu$, it holds that 
\begin{equation}
\mathbb{E}\left( F(\mu_{n,\alpha,\beta}) \right) \xrightarrow{n \rightarrow \infty} \mathbb{E} \left( F(\mu_{\alpha,\beta}) \right).
\end{equation}
\end{theorem}  

The specialisation of the formulas in \cite{Claeys et al} to our situation is stated in Theorem \ref{thm:T+H Claeys}. Using these, we can prove our second main result, which extends Theorem \ref{thm:main} to the full circle, but for a slightly smaller set of parameters $\alpha,\beta$. The reason for the different sets of parameters $\alpha, \beta$ is explained in Remark \ref{remark:parameters}. 

\begin{theorem} \label{thm:main2}
Let $\alpha^2 - \beta^2 < 1/2$ and $0 \leq \alpha < 1/2$. Then the sequence of random measures $\left( \mu_{n,\alpha,\beta}\right)_{n \in \mathbb{N}}$ converges weakly to $\mu_{\alpha,\beta}$ in the space of Radon measures on $[0,2\pi)$ equipped with the topology of weak convergence. 
\end{theorem}

\noindent \textbf{Proof strategy:} Let $I$ denote either $I_\epsilon$ or $[0,2\pi)$. We first remark that by Theorem 4.2. in \cite{Kallenberg2}, weak convergence of $\mu_{n,\alpha,\beta}$ to $\mu_{\alpha,\beta}$ in the space of Radon measures on $I$ equipped with the topology of weak convergence is equivalent to
\begin{equation}\label{eqn:conv equiv}
\int_{I} g(\theta) \mu_{n,\alpha,\beta}(\text{d}\theta) \xrightarrow{d} \int_{I} g(\theta) \mu_{\alpha,\beta}(\text{d}\theta),
\end{equation}
as $n \rightarrow \infty$, for any bounded continuous non-negative function $g$ on $I$. \\

\noindent To prove (\ref{eqn:conv equiv}) we use Theorem 4.28 in \cite{Kallenberg1}:
\begin{theorem} \label{thm:approximation} For $k,n \in \mathbb{N}$ let $\xi$, $\xi_n$, $\eta^k$ and $\eta^k_n$ be random variables with values in a metric space $(S,\rho)$ such that $\eta_n^k \xrightarrow{d} \eta^k$ as $n\rightarrow \infty$ for any fixed $k$, and also $\eta^k \xrightarrow{d} \xi$ as $k\rightarrow \infty$. Then $\xi_n \xrightarrow{d} \xi$ holds under the further condition
\begin{equation}
\lim_{k\rightarrow \infty} \limsup_{n\rightarrow \infty} \mathbb{E}\left( \rho(\eta_n^k, \xi_n) \wedge 1 \right) = 0.
\end{equation}
\end{theorem}

Our setting corresponds to $S = \mathbb{R}$, $\rho = | \cdot |$, and 
\begin{align}
\begin{split}
\xi &= \int_{I} g(\theta) \mu_{\alpha,\beta}(\text{d}\theta), \quad \xi_n = \int_{I} g(\theta) \mu_{n,\alpha,\beta}(\text{d}\theta), \\
\eta^k &= \int_{I} g(\theta) \mu^{(k)}_{\alpha,\beta}(\text{d}\theta), \quad 
\eta^k_n = \int_{I} g(\theta) \mu^{(k)}_{n,\alpha,\beta}(\text{d}\theta),
\end{split}
\end{align}
where $\mu^{(k)}_{n,\alpha,\beta}$ will now be defined by truncating the Fourier series of $\ln f_{n,\alpha,\beta}$. We have 
\begin{align}
\begin{split}
\ln f_{n,\alpha,\beta}(\theta) &= -\sum_{j=1}^\infty \frac{1}{j} \left( (\alpha + \beta)\text{Tr}(U_n^j)e^{-ij\theta} + (\alpha - \beta)\text{Tr}(U_n^{-j})e^{ij\theta} \right) \\
&= - \sum_{j=1}^\infty \frac{\text{Tr}(U^j_n)}{j} \left( 2\alpha \cos(j\theta) - 2i\beta \sin (j\theta) \right),
\end{split}
\end{align}
where we used that for $U \in O(n)$ or $Sp(2n)$ we have $\text{Tr}(U^{-j}_n) = \text{Tr}(U^j_n)$.

\begin{definition}\label{def:mu_n^(k)} For $k,n \in \mathbb{N}$, $\alpha \in \mathbb{R}$, $\beta \in i\mathbb{R}$ and $\theta \in I$, let
\begin{equation}
f_{n,\alpha,\beta}^{(k)}(\theta) = e^{-\sum_{j=1}^k \frac{\text{Tr}(U^j_n)}{j} \left(2\alpha \cos(j\theta) - 2i\beta \sin(j\theta)\right)}, 
\end{equation}
and 
\begin{equation}
\mu_{n,\alpha,\beta}^{(k)}(\text{d}\theta) = \frac{f_{n,\alpha,\beta}^{(k)}(\theta)}{\mathbb{E}(f_{n,\alpha,\beta}^{(k)}(\theta))} \text{d}\theta.
\end{equation}
\end{definition}

In order to apply Theorem \ref{thm:approximation} to verify (\ref{eqn:conv equiv}), which then implies our main results, we thus need to examine the following three limits: for any bounded continuous non-negative function on $I$ 
\begin{align} \label{eqn:first limit} 
\lim_{k \rightarrow \infty} \int_{I} g(\theta) \mu^{(k)}_{\alpha,\beta} (\text{d}\theta) \overset{d}{=}& \int_{I} g(\theta) \mu_{\alpha,\beta} (\text{d}\theta), 
\end{align}
\begin{align} \label{eqn:second limit}
\lim_{n \rightarrow \infty} \int_{I} g(\theta) \mu_{n,\alpha,\beta}^{(k)} (\text{d}\theta) \overset{d}{=}& \int_{I} g(\theta) \mu^{(k)}_{\alpha,\beta} (\text{d}\theta),
\end{align}
and
\begin{align} \label{eqn:third limit}
\lim_{k\rightarrow \infty} \limsup_{n\rightarrow \infty} \mathbb{E}\left( \left| \int_{I} g(\theta) \mu_{n,\alpha,\beta}^{(k)} (\text{d}\theta) - \int_{I} g(\theta) \mu_{n,\alpha,\beta} (\text{d}\theta) \right| \wedge 1 \right) = 0.
\end{align}
The first limit (\ref{eqn:first limit}) follows immediately from the definitions of $\mu_{\alpha,\beta}^{(k)}$ and $\mu_{\alpha,\beta}$ in Appendix \ref{appendix:GMC}: almost surely 
\begin{equation}
\lim_{k\rightarrow \infty} \mu^{(k)}_{\alpha,\beta} \overset{d}{=} \mu_{\alpha,\beta},
\end{equation}
so in particular almost surely (and thus also in distribution)
\begin{equation}
\lim_{k \rightarrow \infty} \int_{I} g(\theta) \mu^{(k)}_{\alpha,\beta} (\text{d}\theta) = \int_{I} g(\theta) \mu_{\alpha,\beta} (\text{d}\theta).
\end{equation}
The second limit (\ref{eqn:second limit}) will be proved in Section \ref{section:second limit}, using previously established results on the asymptotics of Toeplitz+Hankel determinants.\\

\noindent To show that the third limit (\ref{eqn:third limit}) holds, we will prove the following lemma in Section \ref{section:L2}:
\begin{lemma}[The $L^2$-limit]\label{lem:L2 limit} 
Let $\alpha^2 - \beta^2 <1/2$. Further let $0 \leq \alpha < 1/2$ in the case $I = [0,2\pi)$, and $\alpha > -1/4$ in the case $I = I_\epsilon$. Then for any bounded continuous non-negative function $g$ on $I$ the following expectation goes to zero, as first $n \rightarrow \infty$ and then $k \rightarrow \infty$:
\begin{align} \label{eqn:L2 limit}
\begin{split}
&\mathbb{E} \left( \left( \int_{I} g(\theta) \mu_{n,\alpha,\beta}(\text{d}\theta) - \int_{I} g(\theta) \mu_{n,\alpha,\beta}^{(k)}(\text{d}\theta) \right)^2 \right) \\
=&\int_{I} \int_{I} g(\theta)g(\theta') \frac{\mathbb{E}\left(f_{n,\alpha,\beta}^{(k)}(\theta) f_{n,\alpha,\beta}^{(k)}(\theta') \right)}{\mathbb{E}\left(f_{n,\alpha,\beta}^{(k)}(\theta)\right)\mathbb{E}\left(f_{n,\alpha,\beta}^{(k)}(\theta')\right)} \text{d}\theta \text{d}\theta'\\
&-2 \int_{I} \int_{I} g(\theta)g(\theta') \frac{\mathbb{E}\left(f_{n,\alpha,\beta}^{(k)}(\theta) f_{n,\alpha,\beta}(\theta') \right)}{\mathbb{E}\left(f_{n,\alpha,\beta}^{(k)}(\theta)\right)\mathbb{E}\left(f_{n,\alpha,\beta}(\theta')\right)} \text{d}\theta \text{d}\theta' \\
&+ \int_{I} \int_{I} g(\theta)g(\theta') \frac{\mathbb{E}\left(f_{n,\alpha,\beta}(\theta) f_{n,\alpha,\beta}(\theta') \right)}{\mathbb{E}\left(f_{n,\alpha,\beta}(\theta)\right)\mathbb{E}\left(f_{n,\alpha,\beta}(\theta')\right)} \text{d}\theta \text{d}\theta'.
\end{split}
\end{align}
\end{lemma}
All the expectations inside the integrals can be expressed as (sums of) Toeplitz+Hankel determinants (see (\ref{eqn:T+H def}) and Theorem \ref{thm:average}). We will prove Lemma \ref{lem:L2 limit} explicitly only in the case $I = [0,2\pi)$, $\alpha^2 - \beta^2 <1/2$, $0 \leq \alpha < 1/2$, using Theorem \ref{thm:T+H Claeys} below. \\

The proof in the case $I = I_\epsilon$, $\alpha^2 - \beta^2 < 1/2$, $\alpha > -1/4$, works in almost exactly the same way. Instead of Theorem \ref{thm:T+H Claeys}, which only holds for $\alpha \geq 0$, we use new results on the uniform asymptotics of Toeplitz+Hankel determinants of symbols with two pairs of merging singularities bounded away from $\pm 1$, which are stated in Theorems \ref{thm:T+H uniform} and \ref{thm:T, T+H extended} in the next section. To the best of our knowledge these results have not previously been set out and we believe them to be of independent interest.

\section{Statement of Results on Toeplitz and Toeplitz+Hankel Determinants with Merging Singularities}

\begin{definition} \label{def:FH}(\cite{DeiftItsKrasovsky})
A function $f:S^1 \rightarrow \mathbb{C}$ is called a symbol with a fixed number of Fisher-Hartwig singularities if it has the following form:
\begin{equation}
f(z) = e^{V(z)} z^{\sum_{j=0}^{m} \beta_j} \prod_{j=0}^m |z-z_j|^{2\alpha_j} g_{z_j,\beta_j}(z) z_j^{-\beta_j}, \quad z = e^{i\theta}, \quad \theta \in [0,2\pi),
\end{equation}
for some $m = 0,1,...$, where 
\begin{align}
z_j = e^{i\theta_j}, \quad j = 0,...,m, \quad 0 = \theta_0 < \theta_1 < ... < \theta_m < 2\pi;\\
g_{z_j,\beta_j}(z) = \begin{cases} e^{i\pi\beta_j} & 0 \leq \text{arg} \, z < \theta_j \\ e^{-i\pi\beta_j} & \theta_j \leq \text{arg} \, z < 2\pi \end{cases},\\
\Re \alpha_j > -1/2, \quad \beta_j \in \mathbb{C}, \quad j = 0,...,m,
\end{align}
and $V(z)$ is analytic in a neighborhood of the unit circle. A point $z_j$, $j=1,...,m$, is included if and only if either $\alpha_j \neq 0$ or $\beta_j \neq 0$, while always $z_0 = 1$, even if $\alpha_0 = \beta_0 = 0$.
\end{definition} 

Under these assumptions $V$ has a Laurent series, convergent in a neighborhood of the unit circle,
\begin{equation}
V(z) = \sum_{k = -\infty}^\infty V_kz^k, \quad V_k = \frac{1}{2\pi} \int_0^{2\pi} V(e^{i\theta})e^{-ik\theta}\text{d}\theta,
\end{equation} 
and the function $e^{V(z)}$ allows the standard Wiener-Hopf decomposition:
\begin{equation}
e^{V(z)} = b_+(z) b_0 b_-(z), \quad b_+(z) = e^{\sum_{k=1}^\infty V_kz^k}, \quad b_0 = e^{V_0}, \quad b_-(z) = e^{\sum_{k= -\infty}^{-1} V_k z^k}.
\end{equation}
Note that a symbol $f$ as in Definition \ref{def:FH} that is real on the unit circle and fulfills $f(e^{i\theta}) = f(e^{-i\theta})$ is of the following form:  
\begin{align} \label{eqn:FH even}
f(z) = e^{V(z)} \prod_{j=0}^{r+1} |z - e^{i\theta_j}|^{2\alpha_j} |z - e^{-i\theta_j}|^{2\alpha_j} g_{e^{i\theta_j},\beta_j}(z) g_{e^{i(2\pi - \theta_j)}, - \beta_j}(z) e^{-i\theta_j\beta_j} e^{i(2\pi - \theta_j)\beta_j}, 
\end{align}
where $r \in \mathbb{N} \cup \{0\}$, $0 = \theta_0 < \theta_1 < ... < \theta_r < \theta_{r+1} = \pi$, $\alpha_j > -1/2$, $\beta_j \in i\mathbb{R}$ for $j = 0,...,r+1$, and $\beta_0 = \beta_{r+1} = 0$, and where $V(e^{i\theta}) = V(e^{-i\theta})$ such that $V_k = V_{-k}$.\\
 
The Toeplitz and Toeplitz+Hankel determinants we consider are defined as follows:
\begin{align} \label{eqn:T+H def}
\begin{split}
D_n(f) &:= \det \left( f_{j-k} \right)_{j,k=0}^{n-1},\\
D_n^{T+H,1}(f) &:= \det \left( f_{j-k} + f_{j+k} \right)_{j,k = 0}^{n-1}, \\
D_n^{T+H,2}(f) &:= \det \left( f_{j-k} - f_{j+k+2} \right)_{j,k = 0}^{n-1}, \\
D_n^{T+H,3}(f) &:= \det \left( f_{j-k} - f_{j+k+1} \right)_{j,k = 0}^{n-1}, \\
D_n^{T+H,4}(f) &:= \det \left( f_{j-k} + f_{j+k+1} \right)_{j,k = 0}^{n-1},
\end{split}
\end{align}
where 
\begin{equation}
f_j = \frac{1}{2\pi} \int_0^{2\pi} f(e^{i\theta})e^{-ij\theta}\text{d}\theta.
\end{equation}
The study of the asymptotics of Toeplitz determinants was initiated by Szeg\"{o}. The simplest case is the strong Szeg\"{o} limit theorem (see for example \cite{Simons} for the most general version), which states that for $f(z) = e^{V(z)}$
\begin{equation}
D_n(e^V) = e^{nV_0}e^{\sum_{k = 1}^\infty kV_kV_{-k}}(1 + o(1)), \quad n \rightarrow \infty.
\end{equation}
Subsequently the asymptotics of $D_n(f)$ and $D_n^{T+H,\kappa}(f)$ have been computed under various assumptions on the symbol $f$ \cite{Basor, Basor2, BasorEhrhardt, BasorEhrhardt2, BasorEhrhardt3, BasorEhrhardt4, BoettcherSilbermann, BoettcherSilbermann2, BoettcherSilbermann3, EhrhardtSilbermann, Ehrhardt, Widom}; see \cite{DeiftItsKrasovsky2} for a recent historical account and \cite{DeiftItsKrasovsky} for the most general results. The most general results for the case that $f$ is real-valued are stated below in Theorems \ref{thm:T non-uniform} and \ref{thm:T+H non-uniform}. However,  all these results are only valid if the singularities of $f$ are bounded away from each other as $n \rightarrow \infty$.\\
In recent years advances were made on the asymptotics of $D_n(f)$ when the singularities merge as $n \rightarrow \infty$. In \cite{ClaeysKrasovsky} the asymptotics were computed for two merging singularities, and were related to a solution of the Painlev\'{e} V equation. Using the techniques in \cite{ClaeysKrasovsky}, we establish new results on Toeplitz determinants in Theorems \ref{thm:T uniform} and \ref{thm:T, T+H extended} that give the asymptotics when there are two conjugate pairs of merging singularities which are bounded away from $\pm 1$. Again the asymptotics are related to a Painlev\'{e} V equation, in fact the same one as in \cite{ClaeysKrasovsky}. \\
In \cite{Fahs} the asymptotics of $D_n(f)$ for arbitrarily many singularities merging were computed up to a factor $e^{O(1)}$ which is uniformly bounded and bounded away from $0$. The precise factor is believed to be related to higher-dimensional analogues of Painlev\'{e} equations. Using the Riemann-Hilbert analysis of \cite{Fahs}, the asymptotics of $D_n^{T+H,\kappa}(f)$ for $f$ having arbitrarily many singularities merging were computed up to an $e^{O(1)}$ factor in \cite{Claeys et al}. This result is stated in Theorem \ref{thm:T+H Claeys}. \\
For two conjugate pairs of merging singularities which are bounded away from $\pm 1$ our Theorem \ref{thm:T+H uniform} expresses that factor in terms of a Painlev\'{e} transcendent, which again is the same one as in \cite{ClaeysKrasovsky}. \\

Closely related to Toeplitz and Toeplitz+Hankel determinants are the Hankel determinants
\begin{align}
\begin{split}
\det \left( \int_I x^{j+k} w(x) \text{d}x \right)_{j,k = 0}^{n-1},
\end{split}
\end{align}
where $I \subset \mathbb{R}$ is an interval on which the weight $w(x)$ is supported. Hankel determinants with Fisher-Hartwig singularities have a weight of the form 
\begin{equation}
w(x) = e^{-nV(x)} e^{W(x)} \omega(x),
\end{equation}
where $W(x)$ is continuous, $\omega(x)$ has Fisher-Hartwig singularities and $V(x)$ is a potential (in case $I$ is unbounded $W(x)$ and $V(x)$ need to fulfill certain integrability conditions). There are three canonical cases:
\begin{itemize}
\item $I = \mathbb{R}$, $V(x) = 2x^2$,
\item $I = [-1,\infty)$, $V(x) = 2(x+1)$,
\item $I = [-1,1]$, $V(x) = 0$.
\end{itemize}
Using the Heine identity those cases of Hankel determinants appear as averages over the (scaled and shifted) Gaussian Unitary Ensemble, Wishart Ensemble and Jacobi Ensemble respectively. 

\cite{ItsKrasovsky}, \cite{Garoni} and \cite{Krasovsky} are important early works on the large-$n$ asymptotics of Hankel determinants with Fisher-Hartwig singularities that have greatly contributed to the development of the theory. The most general results have been proven in \cite{Charlier} and \cite{CharlierGharakhloo}. \cite{ClaeysFahs} concerns the case when singularities are merging as $n \rightarrow \infty$. There are various formulas which relate Hankel determinants with $I=[-1,1]$ and $V(x) = 0$ to Toeplitz determinants and Toeplitz+Hankel determinants, for example Theorem 2.6 and Lemma 2.7 in \cite{DeiftItsKrasovsky}. A combination of those two results gives a relation between Toeplitz and Toeplitz+Hankel determinants and is stated in Lemma \ref{lem:connection between T and T+H} below. We use this relation to prove our results on uniform asymptotics of Toeplitz+Hankel determinants.\\

The following two theorems state the asymptotics of Toeplitz and Toeplitz+Hankel determinants with real-valued symbols, when the singularities are bounded away from each other: 

\begin{theorem}(Ehrhardt \cite{Ehrhardt}) \label{thm:T non-uniform}
Let $f$ be as in Definition \ref{def:FH}, with $\alpha_j \in \mathbb{R}$ and $\beta_j \in i\mathbb{R}$ for $j = 0,...,m$, and $V(S^1) \subset \mathbb{R}$. Let $\epsilon > 0$, then as $n \rightarrow \infty$, uniformly for $\min_{j,k = 0,...,m} |z_j - z_k| \geq \epsilon$: 
\begin{align}
\begin{split}
\ln D_n(f) =& n V_0 + \sum_{k = 0}^\infty k|V_k|^2 + \left( \ln n \right) \sum_{j = 0}^m (\alpha_j^2 - \beta_j^2)  \\
&- \sum_{j=0}^m (\alpha_j - \beta_j) \left(\sum_{k = 1}^\infty V_k z_j^k\right) + (\alpha_j + \beta_j) \left( \sum_{k = 1}^\infty V_{-k} \overline{z_j}^k \right) \\
&+ \sum_{0 \leq j < k \leq m} 2(\beta_j\beta_k - \alpha_j\alpha_k) \ln |z_j - z_k| + (\alpha_j\beta_k - \alpha_k \beta_j) \ln \frac{z_k}{z_j e^{i\pi}} \\
&+ \sum_{j = 0}^m \ln \frac{G(1+\alpha_j + \beta_j) G(1+\alpha_j - \beta_j)}{G(1+2\alpha_j)} \\
&+ o(1).
\end{split}
\end{align}
\end{theorem}
\begin{theorem} (Deift, Its, Krasovsky \cite{DeiftItsKrasovsky}) \label{thm:T+H non-uniform}
Let $f(z)$ be as in (\ref{eqn:FH even}) and let $\epsilon > 0$. Then as $n \rightarrow \infty$, uniformly for $\min_{j,k = 0,...,r+1} |z_j - z_k| \geq \epsilon$:
\begin{align}
\begin{split} 
D_n^{T+H,\kappa}(f) =& e^{nV_0 + \frac{1}{2}\left( (\alpha_0 + \alpha_{r+1} + s'+ t' )V_0 - (\alpha_0+s')V(1) - (\alpha_{r+1}+t')V(-1) + \sum_{k=1}^\infty k V_k^2 \right)} \\
&\times \prod_{j=1}^r b_+(z_j)^{-\alpha_j+\beta_j} b_- (z_j)^{-\alpha_j - \beta_j} \\
&\times e^{-i \pi \left( \left( \alpha_0 + s' + \sum_{j=1}^r \alpha_j \right) \sum_{j=1}^r \beta_j + \sum_{1 \leq j < k \leq r}(\alpha_j \beta_k - \alpha_k \beta_j) \right)} \\
&\times 2^{(1-s'-t')n+q+\sum_{j=1}^r(\alpha_j^2-\beta_j^2)- \frac{1}{2}(\alpha_0+\alpha_{r+1}+s'+t')^2 + \frac{1}{2}(\alpha_0+\alpha_{r+1}+s'+t')} \\
&\times n^{\frac{1}{2}(\alpha_0^2+\alpha_{r+1}^2) + \alpha_0 s' + \alpha_{r+1} t' + \sum_{j=1}^r (\alpha_j^2 - \beta_j^2)}\\
&\times \prod_{1 \leq j < k \leq r} |z_j - z_k|^{-2(\alpha_j \alpha_k - \beta_j \beta_k)} |z_j - z_k^{-1}|^{-2(\alpha_j \alpha_k + \beta_j \beta_k)} \\ 
&\times \prod_{j=1}^r z_j^{2\tilde{A}\beta_j} |1-z_j^2|^{-(\alpha_j^2 + \beta_j^2)} |1-z_j|^{-2\alpha_j(\alpha_0+s')} |1+z_j|^{-2\alpha_j(\alpha_{r+1}+t')} \\
&\times \frac{\pi^{\frac{1}{2}(\alpha_0+\alpha_{r+1}+s'+t'+1)} G(1/2)^2}{G(1+\alpha_0+s') G(1 + \alpha_{r+1} + t')} 
\prod_{j = 1}^r \frac{G(1+\alpha_j + \beta_j)G(1+\alpha_j - \beta_j)}{G(1+2\alpha_j)} (1+o(1)),
\end{split}
\end{align} 
where $\ln z_j = i\theta_j$, $b_+(z_j)^{-\alpha_j+\beta_j} = (-\alpha_j + \beta_j) \sum_{k = 1}^\infty V_k z_j^k$ and similarly for $b_-(z_j)$, $\tilde{A} = \frac{1}{2}(\alpha_0 + \alpha_{3} + s' + t') + \sum_{j=1}^2 \alpha_j$ and 
\begin{align} 
D_n^{T+H,1}(f) =& \det \left( f_{j-k} + f_{j+k} \right)_{j,k = 0}^{n-1}, \quad \text{with  } q = -2n + 2, \quad s' = t' = - \frac{1}{2}, \\
D_n^{T+H,2}(f) =& \det \left( f_{j-k} - f_{j+k+2} \right)_{j,k = 0}^{n-1}, \quad \text{with  } q = 0, \quad s' = t' = \frac{1}{2}, \\
D_n^{T+H,3}(f) =& \det \left( f_{j-k} - f_{j+k+1} \right)_{j,k = 0}^{n-1}, \quad \text{with  } q = -n, \quad s' = \frac{1}{2}, \quad t' = -\frac{1}{2}, \\
D_n^{T+H,4}(f) =& \det \left( f_{j-k} + f_{j+k+1} \right)_{j,k = 0}^{n-1}, \quad \text{with  } q = -n, \quad s' = -\frac{1}{2}, \quad t' = \frac{1}{2}.
\end{align}
\end{theorem}
\begin{remark}
The asymptotics in Theorem \ref{thm:T uniform} and Theorem \ref{thm:T+H uniform} were not stated in \cite{DeiftItsKrasovsky} to hold uniformly when the singularities are bounded away from each other, but looking at their proofs one can quickly see that this is the case.
\end{remark}
Theorem 2.6 and Lemma 2.7 in \cite{DeiftItsKrasovsky} relate Toeplitz+Hankel determinants to Toeplitz determinants and monic orthogonal polynomials, which is how Theorem \ref{thm:T+H non-uniform} was proven. This relation can be stated as follows:
\begin{lemma}[Deift, Its, Krasovsky] \label{lem:connection between T and T+H}
Let $f$ be as in (\ref{eqn:FH even}). Then for all $n \in \mathbb{N}$:
\begin{align}
\begin{split}
D_n^{T+H,1}(f(z))^2 =& 4 \frac{\left( 1 + \Phi_{2n}(0) \right)^2}{\Phi_{2n}(1) \Phi_{2n}(-1)} D_{2n}(f(z)),\\
D_n^{T+H,2}(f(z))^2 =& \frac{\left( 1 + \Phi_{2n}(0) \right)^2}{\Phi_{2n}(1) \Phi_{2n}(-1)} D_{2n}\left( f(z) |1 - z|^2 |1 + z|^2 \right),\\
D_n^{T+H,3}(f(z))^2 =& \frac{1}{4} \frac{\left( 1 + \Phi_{2n}(0) \right)^2}{\Phi_{2n}(1) \Phi_{2n}(-1)} D_{2n}\left( f(z) |1 - z|^2 \right),\\
D_n^{T+H,4}(f(z))^2 =& \frac{1}{4}  \frac{\left( 1 + \Phi_{2n}(0) \right)^2}{\Phi_{2n}(1) \Phi_{2n}(-1)} D_{2n}\left( f(z) |1 + z|^2 \right),
\end{split}
\end{align}
where $\Phi_{n}(z) = z^n + ...$ are the monic orthogonal polynomials w.r.t. the symbols on the RHS, i.e. $f(z)$, $f(z)|1-z|^2 |1+z|^2$, $f(z)|1-z|^2$ and $f(z) |1+z|^2$.
\end{lemma}
We will use Lemma \ref{lem:connection between T and T+H} in Section \ref{section:T+H} to prove our results on the asymptotics of Toeplitz+Hankel determinants with merging singularities, stated in Theorems \ref{thm:T+H uniform} and \ref{thm:T, T+H extended} below, using our results on the asymptotics of Toeplitz determinants with merging singularities, stated in Theorems \ref{thm:T uniform} and \ref{thm:T, T+H extended} below. In those theorems we consider the following class of symbols: let $\epsilon \in (0,\pi/2)$ and define
\begin{align} \label{eqn:f}
\begin{split}
f_{p,t}(z) :=& e^{V(z)} z^{\sum_{j=0}^{5} \beta_j} \prod_{j=0}^5 |z - z_j|^{2\alpha_j} g_{z_j,\beta_j}(z) z_j^{-\beta_j}, \quad z = e^{i\phi}, \quad \phi \in [0,2\pi), \\
=& e^{V(z)} |z - 1|^{2\alpha_0} |z+1|^{2\alpha_3} \prod_{j=1}^2 |z - z_j|^{2\alpha_j} |z - \overline{z_j}|^{2\alpha_j} g_{z_j,\beta_j}(z) g_{\overline{z_j}, - \beta_j}(z) z_j^{-\beta_j} \overline{z_j}^{\beta_j},
\end{split}
\end{align}
where
\begin{itemize}
\item
$z_0 = 1$, $z_1 = e^{i(p-t)}$, $z_2 = e^{i(p+t)}$, $z_3 = -1$, $z_4 = \overline{z_2} = e^{i(2\pi-p-t)}$, $z_5 = \overline{z_1} = e^{i(2\pi - p +t)}$, with $p \in (\epsilon, \pi - \epsilon)$, $0 < t < \epsilon$, 
\item $\alpha_j \in (-1/2,\infty)$ for $j = 0,...,5$, and $\alpha_1 = \alpha_5$, $\alpha_2 = \alpha_4$, 
\item $\beta_0 = \beta_3 = 0$, $\beta_1 = - \beta_5 \in i\mathbb{R}$, $\beta_2 = - \beta_4 \in i\mathbb{R}$,
\item $V(z)$ is real-valued on the unit circle, and satisfies $V(e^{i\theta}) = V(e^{-i\theta})$. 
\end{itemize}
To state our results on the uniform asymptotics of $D_n(f_{p,t})$ and $D_n^{T+H,\kappa}(f_{p,t})$ we further need the following theorem (not in its most general form) from \cite{ClaeysKrasovsky}, which describes the relevant Painlev\'e transcendents: 
\begin{theorem} (Claeys, Krasovsky) \label{thm:painleve} 
Let $\alpha_1, \alpha_2, \alpha_1 + \alpha_2 > -\frac{1}{2}$, $\beta_1,\beta_2 \in i\mathbb{R}$ and consider the $\sigma$-form of the Painleve V equation
\begin{equation} \label{eqn:painleve}
s^2 \sigma^2_{ss} = (\sigma - s \sigma_s + 2 \sigma_s^2)^2 - 4(\sigma_s - \theta_1)(\sigma_s - \theta_2)(\sigma_s - \theta_3)(\sigma_s - \theta_4),
\end{equation}
where the parameters $\theta_1,\theta_2,\theta_3,\theta_4$ are given by 
\begin{align}
\begin{split}
\theta_1 =& - \alpha_2 + \frac{\beta_1 + \beta_2}{2}, \quad \theta_2 = \alpha_2 + \frac{\beta_1 + \beta_2}{2}, \\
\theta_3 =& \alpha_1 - \frac{\beta_1 - \beta_2}{2}, \quad \theta_4 = - \alpha_1 - \frac{\beta_1 + \beta_2}{2}.
\end{split}
\end{align} 
Then there exists a solution $\sigma(s)$ to (\ref{eqn:painleve}) which is real and free of poles for $s \in -i\mathbb{R}_+$, and which has the following asymptotic behavior along the negative imaginary axis:
\begin{align}
\begin{split}
\sigma(s) =& 2\alpha_1\alpha_2 - \frac{(\beta_1 + \beta_2)^2}{2} + O(|s|^\delta), \quad s \rightarrow -i0_+,\\
\sigma(s) =& \frac{\beta_1 - \beta_2}{2} s - \frac{(\beta_1 - \beta_2)^2}{2} + O(|s|^{-\delta}), \quad s \rightarrow -i\infty,
\end{split}
\end{align}
for some $\delta > 0$.
\end{theorem}
Our result on the uniform asymptotics of $D_n(f_{p,t})$ is then the following, which we prove using the Riemann-Hilbert techniques in \cite{ClaeysKrasovsky}:
\begin{theorem} \label{thm:T uniform}
Let $f_{p,t}$ be as in (\ref{eqn:f}) with $\alpha_1 + \alpha_2 > -1/2$, and let $\sigma$ satisfy the conditions of Theorem \ref{thm:painleve}. Then we have the following large $n$ asymptotics, uniformly for $p \in (\epsilon, \pi - \epsilon)$ and $0 < t < t_0$, for a sufficiently small $t_0 \in (0,\epsilon)$:
\begin{small}
\begin{align}
\begin{split}
\ln D_n(f_{p,t}) =& 2int(\beta_1 - \beta_2) + n V_0 + \sum_{k = 1}^\infty kV_k^2 + \ln (n) \sum_{j = 0}^5 (\alpha_j^2 - \beta_j^2)   \\
&- \sum_{j=0}^5 (\alpha_j - \beta_j) \left(\sum_{k = 1}^\infty V_k z_j^k\right) + (\alpha_j + \beta_j) \left( \sum_{k = 1}^\infty V_k \overline{z_j}^k \right) \\
&+ \sum_{\substack{0 \leq j < k \leq 5 \\ (j,k) \neq (1,2),(4,5)}} 2(\beta_j\beta_k - \alpha_j\alpha_k) \ln |z_j - z_k| + (\alpha_j\beta_k - \alpha_k \beta_j) \ln \frac{z_k}{z_j e^{i\pi}} \\
&+4it(\alpha_1\beta_2 - \alpha_2\beta_1) \\
&+ 2\int_0^{-2int} \frac{1}{s} \left( \sigma(s) - 2\alpha_1\alpha_2 + \frac{1}{2}(\beta_1+\beta_2)^2 \right) \text{d}s + 4\left( \beta_1 \beta_2 - \alpha_1\alpha_2 \right) \ln \frac{\sin t}{nt}\\ 
&+ \ln \frac{G(1+\alpha_0)^2}{G(1+2\alpha_0)} + \ln \frac{G(1+\alpha_3)^2}{G(1+2\alpha_3)} \\
&+ \ln \frac{G(1+\alpha_1+\alpha_2+\beta_1+\beta_2)^2 G(1+\alpha_1+\alpha_2-\beta_1-\beta_2)^2}{G(1+2\alpha_1+2\alpha_2)^2}\\
&+ o(1),
\end{split}
\end{align}
\end{small}
where $\ln \frac{z_k}{z_j e^{i\pi}} = i(\theta_k - \theta_j - \pi)$ and $G$ denotes the Barnes $G$-function. 
\end{theorem}
Our result on the uniform asymptotics of $D_n^{T+H}(f_{p,t})$ is as follows, using the same notation as in Theorem \ref{thm:T+H non-uniform}:
\begin{theorem}  \label{thm:T+H uniform}
Let $f_{p,t}$ be as in (\ref{eqn:f}) with $\alpha_1 + \alpha_2 > -1/2$, and let $\sigma$ satisfy the conditions of Theorem \ref{thm:painleve}. For $D_n^{T+H,\kappa}(f_{p,t})$ we get, as $n \rightarrow \infty$, uniformly in $p \in (\epsilon, \pi - \epsilon)$ and $0 < t < t_0$, for a sufficiently small $t_0 \in (0,\epsilon)$:
\begin{align} \label{eqn:T+H uniform}
\begin{split} 
D_n^{T+H,\kappa}(f_{p,t}) =& e^{2int(\beta_1 - \beta_2) + nV_0 + \frac{1}{2}\left( (\alpha_0 + \alpha_{3} + s'+ t' )V_0 - (\alpha_0+s')V(1) - (\alpha_{3}+t')V(-1) + \sum_{k=1}^\infty k V_k^2 \right)} \\
&\times \prod_{j=1}^2 b_+(z_j)^{-\alpha_j+\beta_j} b_- (z_j)^{-\alpha_j - \beta_j} \\
&\times e^{-i \pi \left( \alpha_0 + s' + \sum_{j=1}^2 \alpha_j \right) \sum_{j=1}^2 \beta_j } \\
&\times 2^{(1-s'-t')n+q+\sum_{j=1}^2 (\alpha_j^2-\beta_j^2)- \frac{1}{2}(\alpha_0+\alpha_{3}+s'+t')^2 + \frac{1}{2}(\alpha_0+\alpha_{3}+s'+t')} \\
&\times n^{\frac{1}{2}(\alpha_0^2+\alpha_{3}^2) + \alpha_0 s' + \alpha_{3} t'+ \sum_{j=1}^2 (\alpha_j^2 - \beta_j^2)}\\
&\times \left| \frac{\sin t}{2nt} \right|^{-2(\alpha_1 \alpha_2 - \beta_1 \beta_2)} |2\sin p|^{-2(\alpha_1 \alpha_2 + \beta_1 \beta_2)} \\ 
&\times e^{ \int_0^{-4int} \frac{1}{s} \left( \sigma(s) - 2\alpha_1\alpha_2 + \frac{1}{2}(\beta_1+\beta_2)^2 \right) \text{d}s } \\
&\times \prod_{j=1}^2 z_j^{2\tilde{A}\beta_j} |1-z_j^2|^{-(\alpha_j^2 + \beta_j^2)} |1-z_j|^{-2\alpha_j(\alpha_0+s')} |1+z_j|^{-2\alpha_j(\alpha_{3}+t')} \\
&\times \frac{\pi^{\frac{1}{2}(\alpha_0+\alpha_{3}+s'+t'+1)} G(1/2)^2}{G(1+\alpha_0+s') G(1 + \alpha_{3} + t')} \\
&\times \frac{G(1 + \alpha_1 + \alpha_2 + \beta_1 + \beta_2) G(1 + \alpha_1 + \alpha_2 - \beta_1 - \beta_2)}{G(1 + 2 \alpha_1 + 2\alpha_2 )} (1+ o(1)).
\end{split}
\end{align} 
\end{theorem}
\begin{remark} One can probably get similar results if more generally one chooses complex $\alpha_j, \beta_j$ with $\Re(\alpha_j) > -1/2$, but to prove Theorem \ref{thm:main} this is not necessary. 
\end{remark} 
\begin{remark} 
The requirements $p \in (\epsilon, \pi - \epsilon)$, $t_0 \in (0,\epsilon)$ are necessary for us to be able to apply the proof techniques in \cite{ClaeysKrasovsky}. The results there only hold for two merging singularities, while if $p \rightarrow 0,\pi$ we have 5 singularities merging at $\pm 1$, and if $t \rightarrow \epsilon$ we can have $p \pm t \rightarrow 0,\pi$ which means 3 singularities are merging at $\pm 1$.
\end{remark}  
\begin{remark}
Comparing the uniform asymptotics of $D_n(f_{p,t})$ in Theorem \ref{thm:T uniform} with the non-uniform asymptotics one gets from Theorem \ref{thm:T non-uniform}, one can see that the different expansions are related in the following way:
\begin{align} \label{eqn:relation}
\begin{split}
&2\sum_{j=1}^2 \ln \frac{G(1+\alpha_j + \beta_j) G(1 + \alpha_j - \beta_j)}{G(1+2\alpha_j)} - 2i\pi (\alpha_1 \beta_2 -\alpha_2\beta_1) + o(1)_{non-uniform} \\
=& 2int(\beta_1 - \beta_2) + 2\int_0^{-2int} \frac{1}{s} \left( \sigma(s) - 2\alpha_1\alpha_2 + \frac{1}{2}(\beta_1+\beta_2)^2 \right) \text{d}s + 4(\beta_1\beta_2 - \alpha_1\alpha_2) \ln \frac{1}{2nt} \\
&+ 2\ln \frac{G(1 + \alpha_1 + \alpha_2 + \beta_1 + \beta_2) G(1 + \alpha_1 + \alpha_2 - \beta_1 - \beta_2)}{G(1 + 2 \alpha_1 + 2\alpha_2 )} + o(1)_{uniform}.
\end{split}
\end{align}
This is exactly the same relationship as the one between the non-uniform and uniform expansions of $D_n(f_t)$ in \cite{ClaeysKrasovsky} (see their (1.8), (1.24) and (1.26)). 
\end{remark}
\begin{remark}
The relationship between the uniform asymptotics of $D_n^{T+H,\kappa}(f_{p,t})$ in Theorem \ref{thm:T+H uniform} and the non-uniform asymptotics one gets from Theorem \ref{thm:T+H non-uniform} is given by (\ref{eqn:relation}), with both sides divided by $2$, and $n$ replaced by $2n$. This is because the uniform asymptotics of $D_n^{T+H,\kappa}(f_{p,t})$ are related to the uniform asymptotics of $D_{2n}(f_{p,t})^{1/2}$ (with added singularities at $\pm 1$) and $\Phi_n(\pm 1)^{1/2}$, $\Phi_n(0)$, by Lemma \ref{lem:connection between T and T+H}. As will be argued in Section \ref{section:T+H} the asymptotics of $\Phi_{2n}(\pm 1)$, computed as in \cite{DeiftItsKrasovsky} are unaffected by the merging of singularities away from $\pm 1$, and $\Phi_{2n}(0) = o(1)$ both when singularities merge or not. Thus only the asymptotics of $D_{2n}(f_{p,t})^{1/2}$ are different in the merging and non-merging regime, and their relationship is given by (\ref{eqn:relation}) with both sides divided by $2$, and $n$ replaced by $2n$.
\end{remark}  
Our last results corresponds to Theorem 1.11 in \cite{ClaeysKrasovsky}. It extends Theorem \ref{thm:T non-uniform} for the symbol $f_{p,t}$, and Theorem \ref{thm:T+H non-uniform} in the case $r = 2$: 
\begin{theorem} \label{thm:T, T+H extended}
Let $\omega(x)$ be a positive, smooth function for $x$ sufficiently large, s.t.
\begin{equation}
\omega(x) \rightarrow \infty, \quad \omega(x) = o(x), \quad \text{as } x \rightarrow \infty.
\end{equation}
Then for any $t_0 \in (0,\epsilon)$ the expansion of $D_n(f_{p,t})$ one gets from Theorem \ref{thm:T non-uniform} holds uniformly in $p \in (\epsilon, \pi - \epsilon)$ and $\omega(n)/n < t < t_0$. For $r = 2$, the expansion of Theorem \ref{thm:T+H non-uniform} holds uniformly in $\theta_1,\theta_2 \in (\epsilon, \pi - \epsilon)$ for which $\omega(n)/n < |\theta_1 - \theta_2|$. 
\end{theorem}
Finally, we state the following result which gives asymptotics of Toeplitz+Hankel determinants which hold uniformly for arbitrarily many merging singularities, but only up to a multiplicative constant:
\begin{theorem} (Claeys, Glesner, Minakov, Yang \cite{Claeys et al}) \label{thm:T+H Claeys}
Let $f$ be as in (\ref{eqn:FH even}) with $r \in \mathbb{N}$, $\alpha_0 = \alpha_{r+1} = 0$ and $\alpha_j \geq 0$, $j = 1,...,r$. Then we have uniformly over the entire region $0 < \theta_1 < ... < \theta_r < \pi$, as $n \rightarrow \infty$,
\begin{align}
\begin{split}
D_n^{T+H,1}(f) =& Fe^{nV_0} \prod_{j = 1}^r n^{\alpha_j^2 - \beta_j^2} \left( \sin \theta_j + \frac{1}{n} \right)^{\alpha_j - \alpha_j^2 - \beta_j^2} \times e^{O(1)},\\
D_n^{T+H,2}(f) =& Fe^{nV_0} \prod_{j = 1}^r n^{\alpha_j^2 - \beta_j^2} \left( \sin \theta_j + \frac{1}{n} \right)^{-\alpha_j - \alpha_j^2 - \beta_j^2} \times e^{O(1)},\\
D_n^{T+H,3}(f) =& Fe^{nV_0} \prod_{j = 1}^r n^{\alpha_j^2 - \beta_j^2} \left( \sin \frac{\theta_j}{2} + \frac{1}{n} \right)^{-\alpha_j - \alpha_j^2 - \beta_j^2} \left( \cos \frac{\theta_j}{2} + \frac{1}{n} \right)^{\alpha_j - \alpha_j^2 - \beta_j^2} \times e^{O(1)},\\
D_n^{T+H,4}(f) =& Fe^{nV_0} \prod_{j = 1}^r n^{\alpha_j^2 - \beta_j^2} \left( \sin \frac{\theta_j}{2} + \frac{1}{n} \right)^{\alpha_j - \alpha_j^2 - \beta_j^2} \left( \cos \frac{\theta_j}{2} + \frac{1}{n} \right)^{-\alpha_j - \alpha_j^2 - \beta_j^2} \times e^{O(1)},
\end{split}
\end{align}
where 
\begin{align}
F = \prod_{1 \leq j < k \leq r } \left( \sin \left| \frac{\theta_j - \theta_k}{2} \right| + \frac{1}{n} \right)^{-2( \alpha_j\alpha_k - \beta_j \beta_k)} \left( \sin \left| \frac{\theta_j +\theta_k}{2} \right| + \frac{1}{n} \right)^{-2(\alpha_j \alpha_k + \beta_j \beta_k)}. 
\end{align}
Here $e^{O(1)}$ denotes a function which is uniformly bounded and bounded away from $0$ as $n \rightarrow \infty$. 
\end{theorem}

\section{The $L^2$-Limit} \label{section:L2}
In this section we prove Lemma \ref{lem:L2 limit}, for which we need to compute asymptotics of all the expectation terms in the integrals, which hold uniformly in $\theta$ and $\theta'$ even as $\theta \rightarrow \theta'$. We use the following theorem to express all those expectations as Toeplitz+Hankel determinants, where we let $SO(n) := \left\{O \in O(n): \det O = 1 \right\}$, $O^-(n) := \left\{O \in O(n): \det O = -1 \right\}$ be the components of the orthogonal groups:

\begin{theorem}\label{thm:average}(Theorem 2.2 in \cite{BaikRains})
Let $h(z)$ be any function on the unit circle such that for $\iota(e^{i\theta}) := h(e^{i\theta})h(e^{-i\theta})$ the integrals 
\begin{equation}
\iota_j = \frac{1}{2\pi} \int_0^{2\pi} \iota(e^{i\theta}) e^{-ij\theta} \, \text{d}\theta
\end{equation} 
are well-defined. Then with $D_n^{T+H,\kappa}$ defined as in (\ref{eqn:T+H def}) we have
\begin{align}
\begin{split}
\mathbb{E}_{SO(2n)}\left( \text{det}(h(U)) \right) &= \frac{1}{2} D_n^{T+H,1}(\iota), \\
\mathbb{E}_{O^-(2n)}\left( \text{det}(h(U)) \right) &=  h(1)h(-1) D_{n-1}^{T+H,2}(\iota), \\
\mathbb{E}_{SO(2n+1)}\left( \text{det}(h(U)) \right) &=  h(1) D_n^{T+H,3}(\iota), \\
\mathbb{E}_{O^-(2n+1)}\left( \text{det}(h(U)) \right) &=  h(-1) D_n^{T+H,4}(\iota), \\
\mathbb{E}_{Sp(2n)}\left( \text{det}(h(U)) \right) &= D_n^{T+H,2}(\iota),
\end{split}
\end{align}
except that $\mathbb{E}_{SO(0)} \left( \det(h(U)) \right) = h(1)$.
\end{theorem} 

We define, for $\phi \in [0,2\pi)$ and $\theta, \theta' \in [0,2\pi)$:
\begin{align} \label{eqn:sigma hat}
\begin{split}
\hat{\sigma}_{1,\theta,\theta'}(e^{i\phi}) &= e^{-\sum_{j=1}^k \frac{2}{j} \Re \left( (\alpha - \beta)(e^{ij\theta} + e^{ij\theta'}) \right) e^{ij\phi} },\\
\hat{\sigma}_{2,\theta,\theta'}(e^{i\phi}) &= e^{-\sum_{j=1}^k \frac{2}{j} \Re \left( (\alpha - \beta)e^{ij\theta} \right) e^{ij\phi}} |e^{i\phi}-e^{i\theta'}|^{2\alpha} e^{2i\beta \Im \ln (1-e^{i(\phi-\theta')})}, \\
\hat{\sigma}_{3,\theta,\theta'}(e^{i\phi}) &= |e^{i\phi}-e^{i\theta}|^{2\alpha} e^{2i\beta \Im \ln (1-e^{i(\phi-\theta)})} |e^{i\phi}-e^{i\theta'}|^{2\alpha} e^{2i\beta \Im \ln (1-e^{i(\phi-\theta')})}, \\
\hat{\sigma}_{4,\theta}(e^{i\phi}) &= e^{-\sum_{j=1}^k \frac{2}{j} \Re \left( (\alpha - \beta)e^{ij\theta} \right) e^{ij\phi}}, \\
\hat{\sigma}_{5,\theta}(e^{i\phi}) &= |e^{i\phi}-e^{i\theta}|^{2\alpha} e^{2i\beta \Im \ln (1-e^{i(\phi-\theta)})}, 
\end{split}
\end{align}
where the branch of the logarithm is the principal one (so in particular $\Im \ln (1-e^{i(\phi-\theta)}) \in (-\pi/2,\pi/2]$). Then we have
\begin{align}
\begin{split}
f_{n,\alpha,\beta}^{(k)}(\theta)f_{n,\alpha,\beta}^{(k)}(\theta') &= \prod_{l = 1}^n \hat{\sigma}_{1,\theta,\theta'}(e^{i\theta_l}), \\
f_{n,\alpha,\beta}^{(k)}(\theta)f_{n,\alpha,\beta}(\theta') &= \prod_{l = 1}^n \hat{\sigma}_{2,\theta,\theta'}(e^{i\theta_l}), \\
f_{n,\alpha,\beta}(\theta)f_{n,\alpha,\beta}(\theta') &= \prod_{l = 1}^n \hat{\sigma}_{3,\theta,\theta'}(e^{i\theta_l}), \\
f_{n,\alpha,\beta}^{(k)}(\theta) &= \prod_{l = 1}^n \hat{\sigma}_{4,\theta}(e^{i\theta_l}), \\
f_{n,\alpha,\beta}(\theta) &= \prod_{l = 1}^n \hat{\sigma}_{5,\theta}(e^{i\theta_l}), 
\end{split}
\end{align}
where the product is up to $2n$ for $U_n \in Sp(2n)$. Further we define 
\begin{align}
\begin{split}
\sigma_{1,\theta,\theta'}(e^{i\phi}) =& \hat{\sigma}_{1,\theta,\theta'}(e^{i\phi})\hat{\sigma}_{1,\theta,\theta'}(e^{-i\phi})\\ 
=& e^{-\sum_{j=1}^k \frac{2}{j} \Re \left((\alpha - \beta)(e^{ij\theta}+e^{ij\theta'}) \right) (e^{ij\phi} + e^{-ij\phi})}, \\
\sigma_{2,\theta,\theta'}(e^{i\phi}) =& \hat{\sigma}_{2,\theta,\theta'}(e^{i\phi})\hat{\sigma}_{2,\theta,\theta'}(e^{-i\phi}) \\
=& e^{-\sum_{j=1}^k \frac{2}{j} \Re \left((\alpha - \beta)e^{ij\theta} \right) (e^{ij\phi} + e^{-ij\phi})} \\
&\times |e^{i\theta'}-e^{i\phi}|^{2\alpha} e^{2i\beta \Im \ln (1-e^{i(\phi-\theta')})} |e^{i\theta'}-e^{-i\phi}|^{2\alpha} e^{2i\beta \Im \ln (1-e^{i(-\phi-\theta')})},\\
\sigma_{3,\theta,\theta'}(e^{i\phi}) =& \hat{\sigma}_{3,\theta,\theta'}(e^{i\phi})\hat{\sigma}_{3,\theta,\theta'}(e^{-i\phi})\\
=& |e^{i\theta}-e^{i\phi}|^{2\alpha} e^{2i\beta \Im \ln (1-e^{i(\phi-\theta)})} |e^{i\theta'}-e^{i\phi}|^{2\alpha} e^{2i\beta \Im \ln (1-e^{i(\phi-\theta')})} \\
&\times |e^{i\theta}-e^{-i\phi}|^{2\alpha} e^{2i\beta \Im \ln (1-e^{i(-\phi-\theta)})} |e^{i\theta'}-e^{-i\phi}|^{2\alpha} e^{2i\beta \Im \ln (1-e^{i(-\phi-\theta')})},\\
\sigma_{4,\theta}(e^{i\phi}) =& \hat{\sigma}_{4,\theta}(e^{i\phi}) \hat{\sigma}_{4,\theta}(e^{-i\phi})\\ 
=&  e^{-\sum_{j=1}^k \frac{2}{j} \Re \left((\alpha - \beta)e^{ij\theta} \right) (e^{ij\phi} + e^{-ij\phi})}, \\
\sigma_{5,\theta}(e^{i\phi}) =& \hat{\sigma}_{5,\theta}(e^{i\phi})\hat{\sigma}_{5,\theta}(e^{-i\phi})\\ 
=& |e^{i\theta}-e^{i\phi}|^{2\alpha} e^{2i\beta \Im \ln (1-e^{i(\phi-\theta)})} |e^{i\theta}-e^{-i\phi}|^{2\alpha} e^{2i\beta \Im \ln (1-e^{i(-\phi-\theta)})}. 
\end{split}
\end{align}
Applying Theorem \ref{thm:average}, we obtain 
\begin{align} \label{eqn:averages sigma}
\begin{split}
&\mathbb{E}_{O(2n)} \left( f_{2n,\alpha, \beta}^{(k)}(\theta) f_{2n,\alpha,\beta}^{(k)}(\theta') \right) \\
=& \frac{1}{2} \mathbb{E}_{SO(2n)} \left( f_{2n,\alpha, \beta}^{(k)}(\theta) f_{2n,\alpha,\beta}^{(k)}(\theta') \right) + \frac{1}{2} \mathbb{E}_{O^-(2n)} \left( f_{2n,\alpha, \beta}^{(k)}(\theta) f_{2n,\alpha,\beta}^{(k)}(\theta') \right) \\
=& \frac{1}{4} D_n^{T+H,1}(\sigma_{1,\theta,\theta'}) + \frac{1}{2} \hat{\sigma}_{1,\theta,\theta'}(1)\hat{\sigma}_{1,\theta,\theta'}(-1) D_{n-1}^{T+H,2}(\sigma_{1,\theta,\theta'}), \\
&\mathbb{E}_{O(2n+1)} \left( f_{2n+1,\alpha, \beta}^{(k)}(\theta) f_{2n+1,\alpha,\beta}^{(k)}(\theta') \right) \\
=& \frac{1}{2} \mathbb{E}_{SO(2n+1)} \left( f_{2n+1,\alpha, \beta}^{(k)}(\theta) f_{2n+1,\alpha,\beta}^{(k)}(\theta') \right) + \frac{1}{2} \mathbb{E}_{O^-(2n+1)} \left( f_{2n+1,\alpha, \beta}^{(k)}(\theta) f_{2n+1,\alpha,\beta}^{(k)}(\theta') \right) \\
=& \frac{1}{2} \hat{\sigma}_{1,\theta,\theta'}(1) D_n^{T+H,3}(\sigma_{1,\theta,\theta'}) + \frac{1}{2} \hat{\sigma}_{1,\theta,\theta'}(-1) D_n^{T+H,4}(\sigma_{1,\theta,\theta'}), \\
&\mathbb{E}_{Sp(2n)} \left(  f_{n,\alpha, \beta}^{(k)}(\theta) f_{n,\alpha,\beta}^{(k)}(\theta') \right) = D_n^{T+H,2}(\sigma_{1,\theta,\theta'}),
\end{split}
\end{align}
and similarly for $\sigma_{2,\theta,\theta'}$, ..., $\sigma_{5,\theta}$.\\

The symbols $\sigma_{1,\theta,\theta'}$, ..., $\sigma_{5,\theta}$ can be written as symbols with Fisher-Hartwig singularities, i.e. as in Definition \ref{def:FH}. Due to our choice of logarithm, we have for $\phi, \theta \in [0,2\pi)$:
\begin{align}
\begin{split}
\Im \ln (1-e^{i(\phi-\theta)}) =& \begin{cases} -\frac{\pi}{2} + \frac{\phi-\theta}{2} & 0 \leq \theta < \phi < 2\pi \\ \frac{\pi}{2} + \frac{\phi-\theta}{2} & 0 \leq \phi \leq \theta < 2\pi \end{cases}, \\
\Im \ln (1-e^{i(-\phi-\theta)}) =& \Im \ln (1-e^{i(2\pi-\theta - \phi)})\\
=&\begin{cases} -\frac{\pi}{2} + \frac{2\pi - \theta - \phi}{2} & 0 \leq \phi < 2\pi - \theta \leq 2\pi \\ \frac{\pi}{2} + \frac{2\pi - \theta - \phi}{2} & 0 < 2\pi - \theta \leq \phi < 2\pi \end{cases},
\end{split}
\end{align}
which implies that we can write
\begin{align}
\begin{split}
\sigma_{2,\theta,\theta'}(e^{i\phi}) =&  e^{-\sum_{j=1}^k \frac{2}{j} \Re \left((\alpha-\beta)e^{ij\theta} \right) (e^{ij\phi} + e^{-ij\phi})} \\
&\times |e^{i\phi}-e^{i\theta'}|^{2\alpha} e^{i\beta (\phi - \theta')} g_{e^{i\theta'},\beta}(e^{i\phi}) |e^{i\phi}-e^{-i\theta'}|^{2\alpha} e^{-i\beta (\phi - (2\pi - \theta'))} g_{e^{i(2\pi - \theta')},-\beta}(e^{i\phi}),\\
\sigma_{3,\theta,\theta'}(e^{i\phi}) =& |e^{i\phi}-e^{i\theta}|^{2\alpha} e^{i\beta (\phi - \theta)} g_{e^{i\theta},\beta}(e^{i\phi}) |e^{i\phi}-e^{-i\theta}|^{2\alpha} e^{-i\beta (\phi - (2\pi - \theta))} g_{e^{i(2\pi - \theta)},-\beta}(e^{i\phi}) \\
& \times |e^{i\phi}-e^{i\theta'}|^{2\alpha} e^{i\beta (\phi - \theta')} g_{e^{i\theta'},\beta}(e^{i\phi}) |e^{i\phi}-e^{-i\theta'}|^{2\alpha} e^{-i\beta (\phi - (2\pi - \theta'))}, g_{e^{i(2\pi - \theta')},-\beta}(e^{i\phi}) \\
\sigma_{5,\theta}(e^{i\phi}) =& |e^{i\phi}-e^{i\theta}|^{2\alpha} e^{i\beta (\phi - \theta)} g_{e^{i\theta},\beta}(e^{i\phi}) |e^{i\phi}-e^{-i\theta}|^{2\alpha} e^{-i\beta (\phi - (2\pi - \theta'))} g_{e^{i(2\pi - \theta')},-\beta}(e^{i\phi}). \\
\end{split}
\end{align}

Theorem \ref{thm:T+H non-uniform} gives asymptotics for the Toeplitz+Hankel determinants of $\sigma_{1,\theta,\theta'}, \sigma_{2,\theta,\theta'}, \sigma_{3,\theta,\theta'}, \sigma_{4,\theta}, \sigma_{5,\theta}$, which are uniform when all of the singularities that appear are bounded away from each other. Since $\sigma_{1,\theta,\theta'}$ and $\sigma_{4,\theta}$ do not have any singularities the asymptotics of their Toeplitz+Hankel determinants are uniform in $\theta, \theta' \in [0,2\pi)$. To obtain asymptotics for the Toeplitz+Hankel determinants of $\sigma_{2,\theta,\theta'}$, $\sigma_{3,\theta,\theta'}$ and $\sigma_{5,\theta}$ that are also uniform when singularities merge we use Theorem \ref{thm:T+H Claeys}. \\

\noindent When applying Theorem \ref{thm:T+H uniform} we always have $\alpha_0 = \alpha_{r+1} = 0$, thus we get
\begin{equation}
\frac{\pi^{\frac{1}{2}(\alpha_0+\alpha_{r+1}+s'+t'+1)} G(1/2)^2}{G(1+\alpha_0+s') G(1 + \alpha_{r+1} + t')} = 1,
\end{equation}  
for any choices of $s',t' \in \{+ 1/2, - 1/2\}$. Further one has to be careful that $z_1$, $z_2$ always correspond to the singularities in the upper half circle, i.e. $\arg z_1, \arg z_2 \in (0, \pi)$. The asymptotics of the Toeplitz+Hankel determinants of $\sigma_{1,\theta,\theta'},\sigma_{2,\theta,\theta'},\sigma_{4,\theta},\sigma_{5,\theta'}$ obtained with Theorem \ref{thm:T+H uniform} are then as follows:

\begin{itemize}
\item For $f(e^{i\phi}) = \sigma_{1,\theta,\theta'}(e^{i\phi})$ we have $r = 0$, $\alpha_0 = \alpha_1 =0$ and
\begin{equation}
V(z) = -\sum_{j=1}^k \frac{2}{j} \Re \left((\alpha-\beta)(e^{ij\theta}+e^{ij\theta'}) \right) (e^{ij\phi} + e^{-ij\phi}).
\end{equation} 
Thus we obtain
\begin{align} \label{eqn:sigma1}
\begin{split}
D_n^{T+H,\kappa}(\sigma_{1,\theta,\theta'}) =& e^{\sum_{j=1}^k \frac{2}{j}   \Re \left((\alpha-\beta)(e^{ij\theta}+e^{ij\theta'}) \right)^2} \\
&\times e^{2\sum_{j=1}^k \frac{s'+(-1)^j t'}{j} \Re \left( (\alpha-\beta)(e^{ij\theta}+e^{ij\theta'}) \right)}\\
&\times 2^{(1-s'-t')n+q - \frac{1}{2}(s'+t')^2 + \frac{1}{2}(s'+t')} (1+o(1)),
\end{split}
\end{align}
uniformly for $\theta,\theta' \in [0,2\pi)$.

\item For $f(e^{i\phi}) = \sigma_{2,\theta,\theta'}(e^{i\phi})$ we have $r = 1$, $\alpha_0 = \alpha_2 = 0$, 
\begin{equation}
z_1 = \begin{cases} e^{i\theta'} & 0 < \theta' < \pi \\ e^{i(2\pi - \theta')} & \pi < \theta' < 2\pi \end{cases}, \quad \beta_1 = \begin{cases} \beta & 0 < \theta' < \pi \\ -\beta & \pi < \theta' < 2\pi \end{cases},
\end{equation}
$\alpha_1 = \alpha$, and 
\begin{equation}
V(z) = -\sum_{j=1}^k \frac{2}{j} \Re \left((\alpha-\beta)e^{ij\theta} \right) (e^{ij\phi} + e^{-ij\phi}).
\end{equation} 
Thus we obtain:
\begin{align} \label{eqn:sigma2}
\begin{split}
D_n^{T+H,\kappa}(\sigma_{2,\theta,\theta'}) =& (2n)^{(\alpha^2-\beta^2)} e^{\sum_{j=1}^k \frac{2}{j}   \Re \left((\alpha-\beta) e^{ij\theta} \right)^2} e^{\sum_{j=1}^k \frac{4}{j} \Re \left( (\alpha-\beta)e^{ij\theta} \right)\Re \left( (\alpha-\beta)e^{ij\theta'} \right)} \\
&\times e^{-i\pi \alpha \beta_1} z_1^{2\alpha \beta_1} |1- e^{2i\theta'}|^{-(\alpha^2+\beta^2)} \frac{G(1+\alpha+\beta) G(1+\alpha-\beta)}{G(1+2\alpha)} \\
&\times e^{2\sum_{j=1}^k \frac{s'+(-1)^j t'}{j} \Re \left( (\alpha-\beta)e^{ij\theta} \right)} 2^{(1-s'-t')n+q - \frac{1}{2}(s'+t')^2 + \frac{1}{2}(s'+t')} \\
&\times e^{-i\pi s' \beta_1} z_1^{\beta_1(s'+t')}  |1-e^{i\theta'}|^{-2\alpha s'} |1+e^{i\theta'}|^{-2\alpha t'} \\
&\times (1+o(1)),
\end{split}
\end{align}
uniformly for $\theta,\theta' \in [0,2\pi)$ s.t. $e^{i\theta'}$ stays bounded away from $\pm 1$.  

\item For $f(e^{i\phi}) = \sigma_{4,\theta}(e^{i\phi})$ we have $r = 0$, $\alpha_0 = \alpha_1 = 0$ and 
\begin{equation}
V(z) = -\sum_{j=1}^k \frac{2}{j} \Re \left((\alpha-\beta)e^{ij\theta} \right) (e^{ij\phi} + e^{-ij\phi}).
\end{equation} 
Thus we obtain:
\begin{align} \label{eqn:sigma4}
\begin{split}
D_n^{T+H,\kappa}(\sigma_{4,\theta}) =& e^{\sum_{j=1}^k \frac{2}{j} \Re \left((\alpha-\beta) e^{ij\theta} \right)^2} e^{2\sum_{j=1}^k \frac{s'+(-1)^j t'}{j} \Re \left( (\alpha-\beta)e^{ij\theta} \right)} \\
&\times 2^{(1-s'-t')n+q - \frac{1}{2}(s'+t')^2 + \frac{1}{2}(s'+t')} (1+o(1)),	
\end{split}
\end{align}
uniformly for $\theta \in [0,2\pi)$.  

\item For $f(e^{i\phi}) = \sigma_{5,\theta'}(e^{i\phi})$ with $e^{i\theta'} \neq \pm 1$, we have $r = 1$, $\alpha_0 = \alpha_2 = 0$, 
\begin{equation} \label{eqn:sigma5 decomposition}
z_1 = \begin{cases} e^{i\theta'} & 0 < \theta' < \pi \\ e^{i(2\pi - \theta')} & \pi < \theta' < 2\pi \end{cases}, \quad \beta_1 = \begin{cases} \beta & 0 < \theta' < \pi \\ -\beta & \pi < \theta' < 2\pi \end{cases}
\end{equation}
$\alpha_1 = \alpha$, and $V = 0$. Thus we obtain:
\begin{align} \label{eqn:sigma5}
\begin{split}
D_n^{T+H,\kappa}(\sigma_{5,\theta'}) =& (2n)^{(\alpha^2-\beta^2)} e^{-i\pi \alpha \beta_1} z_1^{2\alpha \beta_1}  |1- e^{2i\theta'}|^{-(\alpha^2+\beta^2)} \frac{G(1+\alpha+\beta) G(1+\alpha-\beta)}{G(1+2\alpha)} \\
&\times 2^{(1-s'-t')n+q - \frac{1}{2}(s'+t')^2 + \frac{1}{2}(s'+t')} e^{-i\pi s' \beta_1} z_1^{\beta_1(s'+t')}  |1-e^{i\theta'}|^{-2\alpha s'} |1+e^{i\theta'}|^{-2\alpha t'} \\
&\times (1+o(1)),
\end{split}
\end{align}
uniformly for $\theta' \in [0,2\pi)$ s.t. $e^{i\theta'}$ stays bounded away from $\pm 1$.    
\end{itemize}

For $f(e^{i\phi}) = \sigma_{3,\theta,\theta'}(e^{i\phi})$ we have $r = 2$, $\alpha_0 = \alpha_3 = V(z) = 0$ and $z_1,z_2,\alpha_1,\alpha_2,\beta_1,\beta_2$ chosen according to the following decomposition (always $\alpha_1 = \alpha_2 = \alpha$):
\begin{align} \label{eqn:decomposition}
\begin{split}
[0,2\pi)^2 =& \cup_{j=1}^8 J_j \cup \left\{ (\theta, \theta') \in [0,2\pi)^2:\theta = \theta', \text{ or } \theta = 2\pi - \theta', \text{ or } \theta,\theta' \in \{0,\pi\} \right\}, \\
J_1 =& \left\{ (\theta, \theta') \in (0, \pi) \times (0, \pi): \theta < \theta' \right\},  \\
&\text{for } (\theta, \theta') \in J_1: \quad z_1 = e^{i\theta}, z_2 = e^{i\theta'}, \beta_1 = \beta_2 = \beta. \\ 
J_2 =& \left\{ (\theta, \theta') \in (0, \pi) \times (0, \pi): \theta' < \theta \right\},  \\
&\text{for } (\theta, \theta') \in J_2: \quad z_1 = e^{i\theta'}, z_2 = e^{i\theta}, \beta_1 = \beta_2 = \beta. \\ 
J_3 =& \left\{ (\theta, \theta') \in (0, \pi) \times (\pi, 2\pi): \theta < 2\pi - \theta' \right\},  \\
&\text{for } (\theta, \theta') \in J_3: \quad z_1 = e^{i\theta}, z_2 = e^{i(2\pi - \theta')}, \beta_1 = -\beta_2 = \beta. \\ 
J_4 =& \left\{ (\theta, \theta') \in (\pi, 2\pi) \times (0, \pi): \theta' < 2\pi - \theta \right\},  \\
&\text{for } (\theta, \theta') \in J_4: \quad z_1 = e^{i\theta'}, z_2 = e^{i(2\pi - \theta)}, \beta_1 = - \beta_2 = \beta. \\ 
J_5 =& \left\{ (\theta, \theta') \in (\pi, 2\pi) \times (\pi, 2\pi): 2\pi - \theta < 2\pi - \theta' \right\},  \\
&\text{for } (\theta, \theta') \in J_5: \quad z_1 = e^{i(2\pi - \theta)}, z_2 = e^{i(2\pi - \theta')}, \beta_1 = \beta_2 = -\beta. \\ 
J_6 =& \left\{ (\theta, \theta') \in (\pi, 2\pi) \times (\pi, 2\pi): 2\pi - \theta' < 2\pi - \theta \right\},  \\
&\text{for } (\theta, \theta') \in J_6: \quad z_1 = e^{i(2\pi - \theta')}, z_2 = e^{i(2\pi - \theta)}, \beta_1 = \beta_2 = -\beta. \\ 
J_7 =& \left\{ (\theta, \theta') \in (\pi, 2\pi) \times (0, \pi): 2\pi - \theta < \theta' \right\},  \\
&\text{for } (\theta, \theta') \in J_7: \quad z_1 = e^{i(2\pi - \theta)}, z_2 = e^{i\theta'}, \beta_1 = -\beta_2 = -\beta. \\ 
J_8 =& \left\{ (\theta, \theta') \in (0, \pi) \times (\pi, 2\pi): 2\pi - \theta' < \theta \right\}, \\
&\text{for } (\theta, \theta') \in J_8: \quad z_1 = e^{i(2\pi - \theta')}, z_2 = e^{i\theta}, \beta_1 = -\beta_2 = -\beta.
\end{split}
\end{align}

\noindent In this notation we obtain by Theorem \ref{thm:T+H uniform} 
\begin{align} \label{eqn:sigma3 extended} 
\begin{split}
D_n^{T+H,\kappa}(\sigma_{3,\theta,\theta'}) =& (2n)^{2(\alpha^2-\beta^2)} e^{-i\pi \alpha \left( \beta_1 + \beta_2 + 2 \beta_2 \right)} |e^{i\theta} - e^{i\theta'}|^{-2(\alpha^2-\beta^2)} |e^{i\theta} - e^{-i\theta'}|^{-2(\alpha^2+\beta^2)} \\
&\times z_1^{4\beta_1 \alpha} z_2^{4\beta_2\alpha} |1- e^{2i\theta}|^{-(\alpha^2+\beta^2)} |1- e^{2i\theta'}|^{-(\alpha^2+\beta^2)} \frac{G(1+\alpha+\beta)^2 G(1+\alpha-\beta)^2}{G(1+2\alpha)^2} \\
&\times 2^{(1-s'-t')n+q - \frac{1}{2}(s'+t')^2 + \frac{1}{2}(s'+t')} e^{-i\pi s'(\beta_1 + \beta_2)} z_1^{\beta_1(s'+t')} z_2^{\beta_2(s'+t')} \\
&\times |1-e^{i\theta}|^{-2\alpha s'} |1+e^{i\theta}|^{-2\alpha t'} |1-e^{i\theta'}|^{-2\alpha s'} |1+e^{i\theta'}|^{-2\alpha t'} (1+o(1)),
\end{split}
\end{align}
uniformly for $\theta,\theta' \in [0,2\pi)$, s.t. $e^{i\theta}, e^{i\theta'}, e^{-i\theta}, e^{-i\theta'}$
stay bounded away from each other and $\pm 1$.\\

In the following sections we use the asymptotics obtained in this section to compute the asymptotics of the quotients of expectations that appear in Lemma \ref{lem:L2 limit}. 

\subsection{The Symplectic Case} 

By (\ref{eqn:averages sigma}), (\ref{eqn:sigma1}), (\ref{eqn:sigma4}) with $\kappa = 2$, $q=0$, $s' = t' = \frac{1}{2}$, we get 
\begin{align} \label{eqn:symplectic1}
\begin{split}
&\frac{\mathbb{E}_{Sp(2n)} \left( f_{n,\alpha,\beta}^{(k)}(\theta)f_{n,\alpha,\beta}^{(k)}(\theta') \right)}{\mathbb{E}_{Sp(2n)} \left( f_{n,\alpha,\beta}^{(k)}(\theta) \right)\mathbb{E}_{Sp(2n)} \left(f_{n,\alpha,\beta}^{(k)}(\theta') \right)} = \frac{D_n^{T+H,2}(\sigma_{1,\theta,\theta'})}{D_n^{T+H,2}(\sigma_{4,\theta})D_n^{T+H,2}(\sigma_{4,\theta'})} \\
=& e^{\sum_{j=1}^k \frac{4}{j} \Re \left((\alpha - \beta)e^{ij\theta}\right) \Re \left((\alpha - \beta)e^{ij\theta'}\right)} (1+o(1)),
\end{split}
\end{align}
uniformly for $\theta,\theta' \in [0,2\pi)$.\\

By (\ref{eqn:averages sigma}), (\ref{eqn:sigma2}), (\ref{eqn:sigma4}) and  (\ref{eqn:sigma5}) with $\kappa = 2$, $q = 0$, $s' = t' = \frac{1}{2}$, we get
\begin{align} \label{eqn:symplectic2}
\begin{split}
&\frac{\mathbb{E}_{Sp(2n)} \left( f_{n,\alpha,\beta}^{(k)}(\theta)f_{n,\alpha,\beta}(\theta') \right)}{\mathbb{E}_{Sp(2n)} \left( f_{n,\alpha,\beta}^{(k)}(\theta) \right)\mathbb{E}_{Sp(2n)} \left(f_{n,\alpha,\beta}(\theta') \right)} = \frac{D_n^{T+H,2}(\sigma_{2,\theta,\theta'})}{D_n^{T+H,2}(\sigma_{4,\theta})D_n^{T+H,2}(\sigma_{5,\theta'})} \\
=& e^{\sum_{j=1}^k \frac{4}{j} \Re \left((\alpha - \beta)e^{ij\theta}\right) \Re \left((\alpha -\beta)e^{ij\theta'}\right)} (1+o(1)),
\end{split}
\end{align}
uniformly for $\theta,\theta' \in [0,2\pi)$ s.t. $e^{i\theta'}$ stays bounded away from $\pm 1$. \\

By (\ref{eqn:averages sigma}), (\ref{eqn:sigma5}), (\ref{eqn:sigma3 extended}), with $\kappa = 2$, $q = 0$ and $s' = t' = \frac{1}{2}$, and $z_1,z_2,\beta_1,\beta_2$ chosen as in (\ref{eqn:decomposition}), a quick calculation results in
\begin{align} \label{eqn:symplectic3}
\begin{split}
&\frac{\mathbb{E}_{Sp(2n)} \left( f_{n,\alpha,\beta}(\theta)f_{n,\alpha,\beta}(\theta') \right)}{\mathbb{E}_{Sp(2n)} \left( f_{n,\alpha,\beta}(\theta) \right)\mathbb{E}_{Sp(2n)} \left(f_{n,\alpha,\beta}(\theta') \right)} = \frac{D_n^{T+H,2}(\sigma_{3,\theta,\theta'})}{D_n^{T+H,2}(\sigma_{5,\theta})D_n^{T+H,2}(\sigma_{5,\theta'})} \\
=& |e^{i\theta} - e^{i\theta'}|^{-2(\alpha^2-\beta^2)} |e^{i\theta} - e^{-i\theta'}|^{-2(\alpha^2+\beta^2)} z_1^{2\alpha\beta_1} z_2^{2\alpha \beta_2} e^{-2\pi i \alpha \beta_2} (1+o(1)) \\
=& |e^{i\theta} - e^{i\theta'}|^{-2(\alpha^2-\beta^2)} |e^{i\theta} - e^{-i\theta'}|^{-2(\alpha^2 + \beta^2)} e^{4i\alpha\beta \Im \ln (1-e^{i(\theta+\theta')})} (1+o(1)),
\end{split}
\end{align}
uniformly for $\theta,\theta' \in [0,2\pi)$, s.t. $e^{i\theta}, e^{i\theta'}, e^{-i\theta}, e^{-i\theta'}$
stay bounded away from each other and $\pm 1$.

\subsection{The Odd Orthogonal Case}
In the odd orthogonal case we always have $q = -n$, and $s' + t' = 0$, which implies that
\begin{equation}
2^{(1-s'-t')n+q - \frac{1}{2}(s'+t')^2 + \frac{1}{2}(s'+t')} = 1.
\end{equation}
We also note that by (\ref{eqn:sigma hat}) we have 
\begin{align} \label{eqn:sigma hat pm}
\begin{split}
\hat{\sigma}_{1,\theta,\theta'}(\pm 1) =& e^{-2\sum_{j=1}^k \frac{(\pm 1)^j}{j} \Re \left( (\alpha-\beta)\left( e^{ij\theta} + e^{ij\theta'} \right) \right)}, \\
\hat{\sigma}_{2,\theta,\theta'}(\pm 1) =& e^{-2\sum_{j=1}^k \frac{(\pm 1)^j}{j} \Re \left( (\alpha-\beta) e^{ij\theta} \right)} |1 \mp e^{i\theta'}|^{2\alpha} e^{i\beta(\pi - \theta')} g_{e^{i\theta'},\beta}(\pm 1), \\
\hat{\sigma}_{3,\theta,\theta'}(1) =& |1 - e^{i\theta}|^{2\alpha} |1 - e^{i\theta'}|^{2\alpha} e^{i\beta(\pi - \theta)} e^{i\beta(\pi - \theta')}, \\
\hat{\sigma}_{3,\theta,\theta'}(-1) =& |1 + e^{i\theta}|^{2\alpha} |1 + e^{i\theta'}|^{2\alpha} e^{i\beta(\pi - \theta)} e^{i\beta(\pi - \theta')} e^{-i\pi(\beta_1 + \beta_2)}, \\
\hat{\sigma}_{4,\theta}(\pm 1) =& e^{-2\sum_{j=1}^k \frac{(\pm 1)^j}{j} \Re \left( (\alpha-\beta) e^{ij\theta} \right)}, \\
\hat{\sigma}_{5,\theta'}(1) =& |1 - e^{i\theta'}|^{2\alpha} e^{i\beta(\pi - \theta')},\\
\hat{\sigma}_{5,\theta'}(-1) =& |1 + e^{i\theta'}|^{2\alpha} e^{i\beta(\pi - \theta')} e^{-i\pi \beta_1},
\end{split}
\end{align}
where $\beta_1,\beta_2$ are chosen as in (\ref{eqn:decomposition}) for $\hat{\sigma}_{3,\theta,\theta'}$, and as in (\ref{eqn:sigma5 decomposition}) for $\hat{\sigma}_{5,\theta'}$. Thus by (\ref{eqn:averages sigma}), (\ref{eqn:sigma1}) and (\ref{eqn:sigma hat pm}) we get 
\begin{align} \label{eqn:odd sigma1}
\begin{split}
&\mathbb{E}_{O(2n+1)}\left( f_{2n+1,\alpha,\beta}^{(k)}(\theta)f_{2n+1,\alpha,\beta}^{(k)}(\theta') \right) \\
=& \frac{1}{2} \hat{\sigma}_{1,\theta,\theta'}(1) D_n^{T+H,3}(\sigma_{1,\theta,\theta'}) + \frac{1}{2} \hat{\sigma}_{1,\theta,\theta'}(-1) D_n^{T+H,4}(\sigma_{1,\theta,\theta'}) \\
=& \frac{1}{2} e^{\sum_{j=1}^k \frac{2}{j}   \Re \left((\alpha-\beta)(e^{ij\theta}+e^{ij\theta'}) \right)^2} \\
&\times \Bigg( e^{-2\sum_{j=1}^k \frac{1}{j} \Re \left( (\alpha-\beta)\left( e^{ij\theta} + e^{ij\theta'} \right) \right)} e^{\sum_{j=1}^k \frac{1-(-1)^j}{j} \Re \left( (\alpha-\beta)(e^{ij\theta}+e^{ij\theta'}) \right)} \\
&+ e^{-2\sum_{j=1}^k \frac{(-1)^j}{j} \Re \left( (\alpha-\beta)\left( e^{ij\theta} + e^{ij\theta'} \right) \right)} e^{-\sum_{j=1}^k \frac{1-(-1)^j}{j} \Re \left( (\alpha-\beta)(e^{ij\theta}+e^{ij\theta'}) \right)} \Bigg) (1+o(1)) \\
=& e^{\sum_{j=1}^k \frac{2}{j}   \Re \left((\alpha-\beta)(e^{ij\theta}+e^{ij\theta'}) \right)^2} e^{-\sum_{j=1}^k \frac{1+(-1)^j}{j} \Re \left( (\alpha-\beta)(e^{ij\theta}+e^{ij\theta'}) \right)}(1+o(1)),
\end{split}
\end{align}
uniformly for $\theta,\theta' \in [0,2\pi)$. \\

Similarly we obtain from (\ref{eqn:averages sigma}), (\ref{eqn:sigma2}, (\ref{eqn:sigma4}), (\ref{eqn:sigma5}) and (\ref{eqn:sigma hat pm}), that
\begin{align} \label{eqn:odd sigma4}
\mathbb{E}_{O(2n+1)}\left( f_{2n+1,\alpha,\beta}^{(k)}(\theta) \right) = e^{\sum_{j=1}^k \frac{2}{j}   \Re \left((\alpha-\beta)e^{ij\theta} \right)^2} e^{-\sum_{j=1}^k \frac{1+(-1)^j}{j} \Re \left( (\alpha-\beta) e^{ij\theta}\right)} (1+o(1)),
\end{align}
uniformly for $\theta,\theta' \in [0,2\pi)$, and
\begin{align} \label{eqn:odd sigma25}
\begin{split}
\mathbb{E}_{O(2n+1)}\left( f_{2n+1,\alpha,\beta}^{(k)}(\theta)f_{2n+1,\alpha,\beta}(\theta') \right) =& (2n)^{(\alpha^2-\beta^2)} e^{\sum_{j=1}^k \frac{2}{j}   \Re \left((\alpha-\beta) e^{ij\theta} \right)^2} e^{\sum_{j=1}^k \frac{4}{j} \Re \left( (\alpha-\beta)e^{ij\theta} \right)\Re \left( (\alpha-\beta)e^{ij\theta'} \right)} \\
&\times e^{-i\pi \alpha \beta_1} z_1^{2\alpha \beta_1} |1- e^{2i\theta'}|^{-(\alpha^2+\beta^2)} \\
&\times \frac{G(1+\alpha+\beta) G(1+\alpha-\beta)}{G(1+2\alpha)} \\
&\times e^{-\sum_{j=1}^k \frac{1+(-1)^j}{j} \Re \left( (\alpha-\beta)e^{ij\theta} \right)} \\
&\times e^{i\beta(\pi - \theta')} e^{-\frac{i \pi}{2} \beta_1} |1-e^{i\theta'}|^{\alpha} |1+e^{i\theta'}|^{\alpha} (1+o(1)), \\
\mathbb{E}_{O(2n+1)}\left( f_{2n+1,\alpha,\beta}(\theta') \right) =& (2n)^{(\alpha^2-\beta^2)} e^{-i\pi \alpha \beta_1} z_1^{2\alpha \beta_1} |1- e^{2i\theta'}|^{-(\alpha^2+\beta^2)} \\
&\times \frac{G(1+\alpha+\beta) G(1+\alpha-\beta)}{G(1+2\alpha)} \\
&\times e^{i\beta(\pi - \theta')} e^{-\frac{i \pi}{2} \beta_1}
|1-e^{i\theta'}|^{\alpha} |1+e^{i\theta'}|^{\alpha} (1+o(1)),
\end{split}
\end{align}
uniformly for $\theta,\theta' \in [0,2\pi)$, s.t. $e^{i\theta'}$ stays bounded away from $\pm 1$, where
\begin{equation}
z_1 = \begin{cases} e^{i\theta'} & 0 < \theta' < \pi \\ e^{i(2\pi - \theta')} & \pi < \theta' < 2\pi \end{cases}, \quad \beta_1 = \begin{cases} \beta & 0 < \theta' < \pi \\ -\beta & \pi < \theta' < 2\pi \end{cases}.
\end{equation}

From (\ref{eqn:averages sigma}), (\ref{eqn:sigma3 extended}) and (\ref{eqn:sigma hat pm}) we obtain
\begin{align} \label{eqn:odd sigma3 extended}
\begin{split}
& \mathbb{E}_{O(2n+1)}\left( f_{2n+1,\alpha,\beta}(\theta)f_{2n+1,\alpha,\beta}(\theta') \right) \\
=& (2n)^{2(\alpha^2-\beta^2)} e^{-i\pi \alpha \left( \beta_1 + \beta_2 + 2 \beta_2 \right)} \\
&\times |e^{i\theta} - e^{i\theta'}|^{-2(\alpha^2-\beta^2)} |e^{i\theta} - e^{-i\theta'}|^{-2(\alpha^2+\beta^2)} z_1^{4\beta_1\alpha} z_2^{4\beta_2\alpha} |1- e^{2i\theta}|^{-(\alpha^2+\beta^2)} |1- e^{2i\theta}|^{-(\alpha^2+\beta^2)}\\
&\times \frac{G(1+\alpha+\beta)^2 G(1+\alpha-\beta)^2}{G(1+2\alpha)^2} e^{i\beta(\pi - \theta)} e^{i\beta(\pi - \theta')} e^{-\frac{i\pi}{2}(\beta_1 + \beta_2)} \\
&\times |1-e^{i\theta}|^{\alpha} |1+e^{i\theta}|^{\alpha} |1-e^{i\theta'}|^{\alpha} |1+e^{i\theta'}|^{\alpha} (1+o(1)),
\end{split}
\end{align}
uniformly for $\theta,\theta' \in [0,2\pi)$, s.t. $e^{i\theta}, e^{i\theta'}, e^{-i\theta}, e^{-i\theta'}$
stay bounded away from each other and $\pm 1$, and where $z_1,z_2,\beta_1,\beta_2$ are chosen as in (\ref{eqn:decomposition}).\\

Combining (\ref{eqn:odd sigma1}), (\ref{eqn:odd sigma4}) and (\ref{eqn:odd sigma25}), we obtain
\begin{align} \label{eqn:odd orthogonal1}
\begin{split}
&\frac{\mathbb{E}_{O(2n+1)} \left( f_{2n+1,\alpha,\beta}^{(k)}(\theta)f_{2n+1,\alpha,\beta}^{(k)}(\theta') \right)}{\mathbb{E}_{O(2n+1)} \left( f_{2n+1,\alpha,\beta}^{(k)}(\theta) \right)\mathbb{E}_{O(2n+1)} \left(f_{2n+1,\alpha,\beta}^{(k)}(\theta') \right)} \\
=& e^{\sum_{j=1}^k \frac{4}{j} \Re \left((\alpha - \beta)e^{ij\theta}\right) \Re \left((\alpha - \beta)e^{ij\theta'}\right)} (1+o(1)),
\end{split}
\end{align}
uniformly for $\theta,\theta' \in [0,2\pi)$, and 
\begin{align} \label{eqn:odd orthogonal2}
\begin{split}
&\frac{\mathbb{E}_{O(2n+1)} \left( f_{2n+1,\alpha,\beta}^{(k)}(\theta) f_{2n+1,\alpha,\beta}(\theta') \right)}{\mathbb{E}_{O(2n+1)} \left( f_{2n+1,\alpha,\beta}^{(k)}(\theta) \right)\mathbb{E}_{O(2n+1)} \left(f_{2n+1,\alpha,\beta}(\theta') \right)} \\
=& e^{\sum_{j=1}^k \frac{4}{j} \Re \left((\alpha - \beta)e^{ij\theta}\right) \Re \left((\alpha - \beta)e^{ij\theta'}\right)} (1+o(1)), 
\end{split}
\end{align}
uniformly for $\theta,\theta' \in [0,2\pi)$, s.t. $e^{i\theta'}$ stays bounded away from $\pm 1$. \\

By (\ref{eqn:odd sigma25}) and (\ref{eqn:odd sigma3 extended}) we obtain
\begin{align} \label{eqn:odd orthogonal3}
\begin{split}
&\frac{\mathbb{E}_{O(2n+1)} \left( f_{2n+1,\alpha,\beta}(\theta)f_{2n+1,\alpha,\beta}(\theta') \right)}{\mathbb{E}_{O(2n+1)} \left( f_{2n+1,\alpha,\beta}(\theta) \right)\mathbb{E}_{O(2n+1)} \left(f_{2n+1,\alpha,\beta}(\theta') \right)} \\
=& |e^{i\theta} - e^{i\theta'}|^{-2(\alpha^2-\beta^2)} |e^{i\theta} - e^{-i\theta'}|^{-2(\alpha^2+\beta^2)} z_1^{2\alpha\beta_1} z_2^{2\alpha \beta_2} e^{-2\pi i \alpha \beta_2} (1+o(1)), \\
=& e^{4i\alpha\beta \Im \ln (1-e^{i(\theta+\theta')})} |e^{i\theta} - e^{i\theta'}|^{-2(\alpha^2-\beta^2)} |e^{i\theta} - e^{-i\theta'}|^{-2(\alpha^2 + \beta^2)} (1+o(1)),
\end{split}
\end{align}
uniformly for $\theta,\theta' \in [0,2\pi)$, s.t. $e^{i\theta}, e^{i\theta'}, e^{-i\theta}, e^{-i\theta'}$
stay bounded away from each other and $\pm 1$.

\subsection{The Even Orthogonal Case}
In the same way as in the odd orthogonal case one can use (\ref{eqn:averages sigma}), (\ref{eqn:sigma1}) - (\ref{eqn:sigma3 extended}) and (\ref{eqn:sigma hat pm}) to obtain 
\begin{align} \label{eqn:even orthogonal1}
\begin{split}
&\frac{\mathbb{E}_{O(2n)} \left( f_{2n,\alpha,\beta}^{(k)}(\theta)f_{2n,\alpha,\beta}^{(k)}(\theta') \right)}{\mathbb{E}_{O(2n)} \left( f_{2n,\alpha,\beta}^{(k)}(\theta) \right)\mathbb{E}_{O(2n)} \left(f_{2n,\alpha,\beta}^{(k)}(\theta') \right)} = e^{\sum_{j=1}^k \frac{4}{j} \Re \left((\alpha - \beta)e^{ij\theta}\right) \Re \left((\alpha - \beta)e^{ij\theta'}\right)} (1+o(1)),
\end{split}
\end{align}
uniformly for $\theta,\theta \in [0,2\pi)$, and 
\begin{align}
\begin{split} \label{eqn:even orthogonal2}
&\frac{\mathbb{E}_{O(2n)} \left( f_{2n,\alpha,\beta}^{(k)}(\theta) f_{2n,\alpha,\beta}(\theta') \right)}{\mathbb{E}_{O(2n)} \left( f_{2n,\alpha,\beta}^{(k)}(\theta) \right)\mathbb{E}_{O(2n)} \left(f_{2n,\alpha,\beta}(\theta') \right)} = e^{\sum_{j=1}^k \frac{4}{j} \Re \left((\alpha - \beta)e^{ij\theta}\right) \Re \left((\alpha - \beta)e^{ij\theta'}\right)} (1+o(1)),
\end{split}
\end{align}
uniformly for $\theta, \theta' \in [0,2\pi)$, s.t. $e^{i\theta'}$ stays bounded away from $\pm 1$, and 
\begin{align} \label{eqn:even orthogonal3}
\begin{split}
&\frac{\mathbb{E}_{O(2n)} \left( f_{2n,\alpha,\beta}(\theta)f_{2n,\alpha,\beta}(\theta') \right)}{\mathbb{E}_{O(2n)} \left( f_{2n,\alpha,\beta}(\theta) \right)\mathbb{E}_{O(2n)} \left(f_{2n,\alpha,\beta}(\theta') \right)} \\
=& e^{4i\alpha\beta \Im \ln (1-e^{i(\theta+\theta')})} |e^{i\theta} - e^{i\theta'}|^{-2(\alpha^2-\beta^2)} |e^{i\theta} - e^{-i\theta'}|^{-2(\alpha^2 + \beta^2)} (1+o(1)),
\end{split}
\end{align}
uniformly for $\theta,\theta' \in [0,2\pi)$, s.t. $e^{i\theta}, e^{i\theta'}, e^{-i\theta}, e^{-i\theta'}$
stay bounded away from each other and $\pm 1$.\\

In Section \ref{section:second limit} we will also need that
\begin{align} \label{eqn:even sigma4}
\begin{split}
&\mathbb{E}_{O(2n)}\left( f_{2n,\alpha,\beta}^{(k)}(\theta) \right) = e^{- 2\sum_{j=1}^k \frac{\eta_j}{j} \Re\left( (\alpha-\beta)e^{ij\theta} \right) + \sum_{j=1}^k \frac{2}{j} \Re\left((\alpha-\beta)e^{ij\theta}\right)^2}(1+o(1)),
\end{split}
\end{align}
uniformly in $\theta \in [0,2\pi)$. 

\subsection{Pulling the large-$n$ limit inside the integral}
Using the asymptotics computed in the previous sections and Theorem \ref{thm:T+H Claeys}, we follow the proof of Corollary 2.1 in \cite{Fahs} to show that we can pull $\lim_{n \rightarrow \infty}$ inside the integral, i.e. that we can use the non-uniform asymptotics of the integrand:
\begin{lemma} \label{lem:L2 limit Claeys}
Let the expectations be over $O(n)$ or $Sp(2n)$, and $\alpha^2 - \beta^2 < 1/2$ and $0 \leq \alpha < 1/2$. Then 
\begin{align} \label{eqn:L2 limit Claeys1}
\begin{split}
\lim_{n \rightarrow \infty} &\int_0^{2\pi} \int_0^{2\pi} g(\theta) g(\theta') \frac{\mathbb{E} \left( f_{n,\alpha,\beta}(\theta) f_{n,\alpha,\beta}(\theta') \right)} {\mathbb{E} \left( f_{n,\alpha,\beta}(\theta) \right)\mathbb{E} \left( f_{n,\alpha,\beta}(\theta') \right)} \text{d}\theta \text{d}\theta' \\
=&\int_0^{2\pi} \int_0^{2\pi} g(\theta)g(\theta') e^{4i\alpha\beta \Im \ln (1-e^{i(\theta+\theta')})} |e^{i\theta} - e^{i\theta'}|^{-2(\alpha^2-\beta^2)} |e^{i\theta} - e^{-i\theta'}|^{-2(\alpha^2 + \beta^2)} \text{d}\theta \text{d}\theta',
\end{split}
\end{align}
and 
\begin{align} \label{eqn:L2 limit Claeys2}
\begin{split}
\lim_{n \rightarrow \infty} &\int_0^{2\pi} \int_0^{2\pi} g(\theta) g(\theta') \frac{\mathbb{E} \left( f^{(k)}_{n,\alpha,\beta}(\theta) f_{n,\alpha,\beta}(\theta') \right)} {\mathbb{E} \left( f^{(k)}_{n,\alpha,\beta}(\theta) \right)\mathbb{E} \left( f_{n,\alpha,\beta}(\theta') \right)} \text{d}\theta \text{d}\theta' \\
=&\int_0^{2\pi} \int_0^{2\pi} g(\theta)g(\theta') e^{\sum_{j=1}^k \frac{4}{j} \Re \left((\alpha - \beta)e^{ij\theta}\right) \Re \left((\alpha -\beta)e^{ij\theta'}\right)} \text{d}\theta \text{d}\theta'.
\end{split}
\end{align}
\end{lemma}
\begin{remark} \label{remark:parameters}
Note that the limit in (\ref{eqn:L2 limit Claeys1}) is finite since $2(\alpha^2 + \beta^2) \leq 2(\alpha^2 - \beta^2) < 1$ and $2(\alpha^2 + \beta^2) + 2(\alpha^2 - \beta^2) = 4\alpha^2 < 1$. If we restrict the integrals to $I_\epsilon = (\epsilon, \pi - \epsilon) \cup (\pi + \epsilon, 2\pi - \epsilon)$ then we only need that $2(\alpha^2 + \beta^2) \leq 2(\alpha^2 - \beta^2) < 1$ for the integral to be finite, since $|e^{i\theta} - e^{i\theta'}|$ and $|e^{i\theta} - e^{-i\theta'}|$ can only go to zero simultaneously when both $\theta, \theta'$ approach $\{0,\pi,2\pi\}$. Further, if we restrict to $I_\epsilon$ we can rely on Theorems \ref{thm:T+H uniform} and \ref{thm:T, T+H extended} for the proof, which only require $\alpha > -1/4$, and we will not need Theorem \ref{thm:T+H Claeys} anymore, which requires $\alpha \geq 0$. Thus when restricting to $I_\epsilon$, Lemma \ref{lem:L2 limit Claeys} holds for the larger set of parameters $\alpha^2 - \beta^2 < 1/2$ and $\alpha > - 1/4$. 
\end{remark} 
\noindent \textbf{Proof of Lemma \ref{lem:L2 limit Claeys}:} By Theorems \ref{thm:T+H Claeys} and \ref{thm:average} we have 
\begin{align} \label{eqn:quotients Sp(2n)}
\begin{split}
\frac{\mathbb{E}_{Sp(2n)} \left( f_{n,\alpha,\beta}(\theta) f_{n,\alpha,\beta}(\theta') \right)} {\mathbb{E}_{Sp(2n)} \left( f_{n,\alpha,\beta}(\theta) \right)\mathbb{E}_{Sp(2n)} \left( f_{n,\alpha,\beta}(\theta') \right)} =& \frac{D_n^{T+H,2}(\sigma_{3,\theta,\theta'}) }{ D_n^{T+H,2}(\sigma_{5,\theta}) D_n^{T+H,2}(\sigma_{5,\theta'})} \\
=& e^{O(1)} \left( \sin \left| \frac{\theta - \theta'}{2} \right| + \frac{1}{n} \right)^{-2( \alpha^2 - \beta^2)} \left( \sin \left| \frac{\theta +\theta'}{2} \right| + \frac{1}{n} \right)^{-2(\alpha^2 + \beta^2)}, 
\end{split}
\end{align}
as $n\rightarrow \infty$, uniformly in (Lebesgue almost all) $(\theta,\theta') \in [0,2\pi)^2$. By the same theorems we get 
\begin{align}
\begin{split}
&\mathbb{E}_{O(2n+1)} \left( f_{2n+1,\alpha,\beta}(\theta) f_{2n+1,\alpha,\beta}(\theta') \right) = \frac{1}{2} \left( \hat{\sigma}_{3,\theta,\theta'}(1) D_n^{T+H,3}(\sigma_{3,\theta,\theta'}) + \hat{\sigma}_{3,\theta,\theta'}(-1)D_n^{T+H,4}(\sigma_{3,\theta,\theta'}) \right) \\
=& \frac{1}{2} F_{\sigma_{3,\theta,\theta'}} n^{2(\alpha^2 - \beta^2)} \left( e^{O(1)} \prod_{j = 1}^2 \left( \frac{ \sin \frac{|\theta_j|}{2} + \frac{1}{n} }{ \sin^2 \frac{|\theta_j|}{2} \left( \cos \frac{|\theta_j|}{2} + \frac{1}{n} \right)} \right)^{-\alpha} + e^{O(1)} \prod_{j = 1}^2 \left( \frac{ \cos \frac{|\theta_j|}{2} + \frac{1}{n} }{ \cos^2 \frac{|\theta_j|}{2} \left( \sin \frac{|\theta_j|}{2} + \frac{1}{n} \right)} \right)^{-\alpha} \right),
\end{split}
\end{align}
and
\begin{align}
\begin{split}
&\mathbb{E}_{O(2n+1)} \left( f_{2n+1,\alpha,\beta}(\theta) \right) \\
=& \frac{1}{2} F_{\sigma_{5,\theta}} n^{\alpha^2 - \beta^2} \left( e^{O(1)} \left( \frac{ \sin \frac{|\theta|}{2} + \frac{1}{n} }{ \sin^2 \frac{|\theta|}{2} \left( \cos \frac{|\theta|}{2} + \frac{1}{n} \right)} \right)^{-\alpha} + e^{O(1)} \left( \frac{ \cos \frac{|\theta|}{2} + \frac{1}{n} }{ \cos^2 \frac{|\theta|}{2} \left( \sin \frac{|\theta|}{2} + \frac{1}{n} \right)} \right)^{-\alpha} \right),
\end{split}
\end{align}
as $n \rightarrow \infty$, uniformly for (Lebesgue almost all) $(\theta,\theta') \in [0,2\pi)^2$, where $\hat{\sigma}_{1,\theta,\theta'}$,...,$\hat{\sigma}_{5,\theta}$ are defined in (\ref{eqn:sigma hat}). Thus we can see that also for the expectations over $O(2n+1)$ we get
\begin{align} \label{eqn:quotients O(2n+1)}
\begin{split}
&\frac{ \mathbb{E}_{O(2n+1)} \left( f_{2n+1,\alpha,\beta}(\theta) f_{2n+1,\alpha,\beta}(\theta') \right) }{ \mathbb{E}_{O(2n+1)} \left( f_{2n+1,\alpha,\beta}(\theta) \right) \mathbb{E}_{O(2n+1)} \left( f_{2n+1,\alpha,\beta}(\theta') \right)} \\
=& e^{O(1)} \left( \sin \left| \frac{\theta - \theta'}{2} \right| + \frac{1}{n} \right)^{-2( \alpha^2 - \beta^2)} \left( \sin \left| \frac{\theta +\theta'}{2} \right| + \frac{1}{n} \right)^{-2(\alpha^2 + \beta^2)}, 
\end{split}
\end{align}
as $n\rightarrow \infty$, uniformly in (Lebesgue almost all) $(\theta,\theta') \in [0,2\pi)^2$. Similarly, Theorem \ref{thm:T+H Claeys} gives
\begin{align} \label{eqn:quotients O(2n)}
\begin{split}
&\frac{ \mathbb{E}_{O(2n)} \left( f_{2n,\alpha,\beta}(\theta) f_{2n,\alpha,\beta}(\theta') \right) }{ \mathbb{E}_{O(2n)} \left( f_{2n,\alpha,\beta}(\theta) \right) \mathbb{E}_{O(2n)} \left( f_{2n,\alpha,\beta}(\theta') \right)} \\
=& e^{O(1)} \left( \sin \left| \frac{\theta - \theta'}{2} \right| + \frac{1}{n} \right)^{-2( \alpha^2 - \beta^2)} \left( \sin \left| \frac{\theta +\theta'}{2} \right| + \frac{1}{n} \right)^{-2(\alpha^2 + \beta^2)}, 
\end{split}
\end{align}
as well as 
\begin{align} \label{eqn:quotients}
\begin{split}
&\frac{ \mathbb{E}_{Sp(2n)} \left( f^{(k)}_{n,\alpha,\beta}(\theta) f_{n,\alpha,\beta}(\theta') \right) }{ \mathbb{E}_{Sp(2n)} \left( f^{(k)}_{n,\alpha,\beta}(\theta) \right) \mathbb{E}_{Sp(2n)} \left( f_{n,\alpha,\beta}(\theta') \right)} = e^{O(1)},\\
&\frac{ \mathbb{E}_{O(2n+1)} \left( f^{(k)}_{2n+1,\alpha,\beta}(\theta) f_{2n+1,\alpha,\beta}(\theta') \right) }{ \mathbb{E}_{O(2n+1)} \left( f^{(k)}_{2n+1,\alpha,\beta}(\theta) \right) \mathbb{E}_{O(2n+1)} \left( f_{2n+1,\alpha,\beta}(\theta') \right)} = e^{O(1)},\\
&\frac{ \mathbb{E}_{O(2n)} \left( f^{(k)}_{2n,\alpha,\beta}(\theta) f_{2n,\alpha,\beta}(\theta') \right) }{ \mathbb{E}_{O(2n)} \left( f^{(k)}_{2n,\alpha,\beta}(\theta) \right) \mathbb{E}_{O(2n)} \left( f_{2n,\alpha,\beta}(\theta') \right)} = e^{O(1)},
\end{split}
\end{align}
as $n\rightarrow \infty$, uniformly in (Lebesgue almost all) $(\theta,\theta') \in [0,2\pi)^2$.\\

Now, for a given measureable subset $R \subset [0,2\pi)^2$, we denote
\begin{align}
\begin{split}
L_\epsilon(R) =& \int_R \left( \sin \left| \frac{\theta - \theta'}{2} \right| + \epsilon \right)^{-2( \alpha^2 - \beta^2)} \left( \sin \left| \frac{\theta +\theta'}{2} \right| + \epsilon \right)^{-2(\alpha^2 + \beta^2)} \text{d}\theta \text{d}\theta',\\
K_\epsilon(R) =& \int_R \left( \sin \left| \frac{\theta - \theta'}{2} \right| + \epsilon \right)^{-2( \alpha^2 - \beta^2)} \text{d}\theta \text{d}\theta'.
\end{split}
\end{align}
In the case $\alpha^2 + \beta^2 > 0$ we have $L_\epsilon(R) < L_0(R) < \infty$ for any $\epsilon > 0$ (since $2(\alpha^2 - \beta^2), 2(\alpha^2 + \beta^2), 4\alpha^2 < 1$), while in the case $\alpha^2 + \beta^2 \leq 0$ we have $K_\epsilon(R) < K_0(R) < \infty$ for any $\epsilon > 0$. For $\eta > 0$ we define 
\begin{align}
\begin{split}
R_1(\eta) =& \left\{ (\theta,\theta') \in [0,2\pi)^2: \sin \frac{|\theta - \theta'|}{2}, \sin \frac{|\theta + \theta'|}{2} > \mu \right\} \\
R_2(\eta) =& R_1(\eta)^c.
\end{split}
\end{align}
It follows by (\ref{eqn:quotients Sp(2n)}), (\ref{eqn:quotients O(2n+1)}) and (\ref{eqn:quotients O(2n)}) that for any $\eta > 0$ there exists a $C > 0$ and $N_0 \in \mathbb{N}$ such that 
\begin{align}
\begin{split}
\int_{R_2(\eta)} g(\theta) g(\theta') \frac{\mathbb{E} \left( f_{n,\alpha,\beta}(\theta) f_{n,\alpha,\beta}(\theta') \right)} {\mathbb{E} \left( f_{n,\alpha,\beta}(\theta) \right)\mathbb{E} \left( f_{n,\alpha,\beta}(\theta') \right)} \leq \begin{cases} C L_0(R_2(\eta)), & \alpha^2 + \beta^2 > 0, \\
C K_0(R_2(\eta)), & \alpha^2 + \beta^2 \leq 0, \end{cases}
\end{split}
\end{align}
for $n > N_0$. Fix $\delta > 0$. Since $L_0(R_2(\eta)), K_0(R_2(\eta) \rightarrow 0$ as $\eta \rightarrow 0$, it follows that there exists an $\eta_0 > 0$ and an $N_0 \in \mathbb{N}$ such that
\begin{align} \label{eqn:bound R_2}
\begin{split}
\int_{R_2(\eta)} g(\theta) g(\theta') \frac{\mathbb{E} \left( f_{n,\alpha,\beta}(\theta) f_{n,\alpha,\beta}(\theta') \right)} {\mathbb{E} \left( f_{n,\alpha,\beta}(\theta) \right)\mathbb{E} \left( f_{n,\alpha,\beta}(\theta') \right)} < \delta/2
\end{split}
\end{align}
for $n > N_0$ and $\eta < \eta_0$.  \\

Using (\ref{eqn:symplectic3}), (\ref{eqn:odd orthogonal3}) and (\ref{eqn:even orthogonal3}), we get that for any fixed $\eta > 0$ it holds that
\begin{align}
\begin{split}
&\int_{R_1(\eta)} g(\theta) g(\theta') \frac{\mathbb{E} \left( f_{n,\alpha,\beta}(\theta) f_{n,\alpha,\beta}(\theta') \right)} {\mathbb{E} \left( f_{n,\alpha,\beta}(\theta) \right)\mathbb{E} \left( f_{n,\alpha,\beta}(\theta') \right)} = (1+ o(1)) L_0(R_1(\eta)) \\
=& (1+o(1)) L_0([0,2\pi)^2) - (1+o(1)) L_0(R_2(\eta)).
\end{split}
\end{align}
We pick $\eta < \eta_0$ such that $I_0(R_2(\eta)) < \delta/4$, then we have
\begin{align}
\begin{split}
\left| \int_{R_1(\eta)} g(\theta) g(\theta') \frac{\mathbb{E} \left( f_{n,\alpha,\beta}(\theta) f_{n,\alpha,\beta}(\theta') \right)} {\mathbb{E} \left( f_{n,\alpha,\beta}(\theta) \right)\mathbb{E} \left( f_{n,\alpha,\beta}(\theta') \right)} - L_0([0,2\pi)^2) \right| = \left| L_0(R_2(\eta)) +o(1) \right| < \delta/4 + o(1).
\end{split}
\end{align}
Together with (\ref{eqn:bound R_2}) we obtain that there exists an $N \in \mathbb{N}$ such that 
\begin{align}
\begin{split}
\left| \int_0^{2\pi} \int_0^{2\pi} g(\theta) g(\theta') \frac{\mathbb{E} \left( f_{n,\alpha,\beta}(\theta) f_{n,\alpha,\beta}(\theta') \right)} {\mathbb{E} \left( f_{n,\alpha,\beta}(\theta) \right)\mathbb{E} \left( f_{n,\alpha,\beta}(\theta') \right)} - L_0([0,2\pi)^2) \right| < \delta
\end{split}
\end{align}
for all $n > N$. Since $\delta >0 $ is arbitrary, this shows the first part of (\ref{eqn:L2 limit Claeys1}). (\ref{eqn:L2 limit Claeys2}) follows from (\ref{eqn:quotients}) in a similar way. \qed 

\subsection{Proof of Lemma \ref{lem:L2 limit}}
Now we have all the ingredients necessary to prove Lemma \ref{lem:L2 limit}. We will only prove it for $I = [0,2\pi)$, the proof for $I = I_\epsilon = (\epsilon, \pi - \epsilon) \cup (\pi + \epsilon, 2\pi + \epsilon)$ is completely analogous and relies on the fact that Lemma \ref{lem:L2 limit Claeys} also holds for $I = I_\epsilon$, $\alpha^2 - \beta^2 < 1/2$ and $\alpha > -1/4$, as explained in Remark \ref{remark:parameters}.\\ 

\noindent \textbf{Proof of Lemma \ref{lem:L2 limit}:} From Lemma \ref{lem:L2 limit Claeys} and (\ref{eqn:L2 limit}), (\ref{eqn:symplectic1}), (\ref{eqn:odd orthogonal1}) and (\ref{eqn:even orthogonal1}) it follows that 
\begin{align} 
\begin{split}
& \lim_{n\rightarrow \infty} \mathbb{E} \left( \left( \int_0^{2\pi} g(\theta) \mu_{n,\alpha,\beta}(\text{d}\theta) - \int_0^{2\pi} g(\theta) \mu_{n,\alpha,\beta}^{(k)}(\text{d}\theta) \right)^2 \right) \\
=& \int_0^{2\pi} \int_0^{2\pi} g(\theta)g(\theta') e^{4i\alpha\beta \Im \ln (1-e^{i(\theta+\theta')})} |e^{i\theta} - e^{i\theta'}|^{-2(\alpha^2-\beta^2)} |e^{i\theta} - e^{-i\theta'}|^{-2(\alpha^2 + \beta^2)} \text{d}\theta \text{d}\theta'\\
& - \int_0^{2\pi} \int_0^{2\pi} g(\theta)g(\theta') e^{\sum_{j=1}^k \frac{4}{j} \Re \left((\alpha - \beta)e^{ij\theta}\right) \Re \left((\alpha - \beta)e^{ij\theta'}\right)}  \text{d}\theta \text{d}\theta'.
\end{split}
\end{align}
We have
\begin{align}\label{eqn:bounded in L2 - 1}
\begin{split}
&4\Re \left((\alpha - \beta)e^{ij\theta}\right) \Re \left((\alpha - \beta)e^{ij\theta'}\right) \\
=& (\alpha - \beta)^2 e^{ij(\theta+\theta')} + (\alpha + \beta)^2 e^{-ij(\theta+\theta')} + (\alpha^2- \beta^2)(e^{ij(\theta-\theta')} + e^{-ij(\theta-\theta')})\\
=& 2(\alpha^2-\beta^2) \cos(j(\theta-\theta')) + 2(\alpha^2+\beta^2) \cos(j(\theta+\theta')) - 2 \alpha \beta \left(e^{ij(\theta+\theta')} - e^{-ij(\theta + \theta')}\right).  
\end{split}
\end{align}
Since
\begin{align}
\begin{split}
\ln |e^{i\theta} - e^{i\theta'}| &= - \sum_{j=1}^\infty \frac{1}{j} \cos(j(\theta-\theta')), \quad \quad \ln |e^{i\theta} - e^{-i\theta'}| = - \sum_{j=1}^\infty \frac{1}{j} \cos(j(\theta+\theta')),
\end{split}
\end{align}
and 
\begin{align}
\begin{split}
-\sum_{j=1}^\infty \frac{1}{j}(e^{ij(\theta+\theta')} - e^{-ij(\theta+\theta')}) =& \ln  (1-e^{i(\theta+\theta')}) - \ln (1-e^{-i(\theta+\theta')}) \\
=& 2i\Im \ln (1-e^{i(\theta+\theta')}),
\end{split}
\end{align}
we see that
\begin{align} \label{eqn:series for covariance function}
\begin{split}
&e^{\sum_{j=1}^\infty \frac{1}{j} \Re \left((\alpha + i\beta)e^{ij\theta}\right) \Re \left((\alpha + i\beta)e^{ij\theta'}\right)} \\
=& |e^{i\theta} - e^{i\theta'}|^{-2(\alpha^2-\beta^2)} |e^{i\theta} - e^{-i\theta'}|^{-2(\alpha^2+\beta^2)} e^{4i\alpha \beta \Im \ln (1-e^{i(\theta+\theta')})}.
\end{split}
\end{align}
Because $\mathbb{E}\left( (...)^2\right) \geq 0$ we have
\begin{align} \label{eqn:bounded in L2 - 2}
\begin{split}
&\int_{0}^{2\pi} \int_{0}^{2\pi} g(\theta)g(\theta') |e^{i\theta} - e^{i\theta'}|^{-2(\alpha^2-\beta^2)} |e^{i\theta} - e^{-i\theta'}|^{-2(\alpha^2+\beta^2)} e^{4i\alpha \beta \Im \ln (1-e^{i(\theta+\theta')})} \text{d}\theta \text{d}\theta'\\
\geq &\limsup_{k\rightarrow \infty} \int_{0}^{2\pi} \int_{0}^{2\pi} g(\theta)g(\theta') e^{\sum_{j=1}^k \frac{4}{j} \Re \left((\alpha - \beta)e^{ij\theta}\right) \Re \left((\alpha - \beta)e^{ij\theta'}\right)}  \text{d}\theta \text{d}\theta'.
\end{split}
\end{align}
Now we use that $g$ is non-negative to apply Fatou's lemma to get the other inequality, which finishes the proof. \qed 

\section{Proof of the Second Limit} \label{section:second limit}

In this section we prove (\ref{eqn:second limit}) for $I = [0,2\pi)$, i.e. that for any fixed $k \in \mathbb{N}$ and bounded continuous function $g:[0,2\pi) \rightarrow \mathbb{R}$ it holds that
\begin{equation}
\int_0^{2\pi} g(\theta) \mu_{n,\alpha,\beta}^{(k)}(\text{d}\theta) \xrightarrow{d} \int_0^{2\pi} g(\theta) \mu_{\alpha,\beta}^{(k)}(\text{d}\theta),
\end{equation}
as $n\rightarrow \infty$, where $\mu_{n,\alpha,\beta}^{(k)}$ is defined in Definition \ref{def:mu_n^(k)} and $\mu^{(k)}_{\alpha,\beta}$ is defined in Appendix \ref{appendix:GMC}. For $I = I_\epsilon = (\epsilon, \pi - \epsilon) \cup (\pi + \epsilon, 2\pi - \epsilon)$ the proof is exactly the same.   \\

We consider the function $F:\mathbb{R}^k \rightarrow \mathbb{R}$,
\begin{equation}
F(z_1,...,z_k) = \int_0^{2\pi} \frac{g(\theta) e^{-2\sum_{j=1}^k \frac{z_j}{\sqrt{j}} \left( \alpha \cos(j\theta) - i\beta \sin(j\theta) \right)}}{e^{\pm 2\sum_{j=1}^k \frac{\eta_j}{j} \Re\left( (\alpha-\beta)e^{ij\theta} \right) + \sum_{j=1}^k \frac{2}{j} \Re\left((\alpha-\beta)e^{ij\theta}\right)^2}} \text{d}\theta, 
\end{equation}
which is continuous since the integrand is continuous in $z_1,...,z_n$ and $\theta$, and bounded in $\theta$ for any fixed $z_1,...,z_n$. Then we have, with $\pm$ corresponding to symplectic/orthogonal:
\begin{align}
\begin{split}
\int_0^{2\pi} g(\theta) \mu_{n,\alpha,\beta}^{(k)}(\text{d}\theta) =& \int_0^{2\pi} \frac{g(\theta) e^{-2\sum_{j=1}^k \frac{\text{Tr}(U_n^j)}{\sqrt{j}} \left( \alpha \cos(j\theta) - i\beta \sin(j\theta) \right)}}{\mathbb{E}(f_{n,\alpha,\beta}^{(k)}(\theta))} \text{d}\theta \\
=& \frac{1}{1+o(1)} \int_0^{2\pi} \frac{g(\theta) e^{-2\sum_{j=1}^k \frac{\text{Tr}(U_n^j)}{\sqrt{j}} \left( \alpha \cos(j\theta) - i\beta \sin(j\theta) \right)}}{e^{\pm 2\sum_{j=1}^k \frac{\eta_j}{j} \Re\left( (\alpha-\beta)e^{ij\theta} \right) + \sum_{j=1}^k \frac{2}{j} \Re\left((\alpha-\beta)e^{ij\theta}\right)^2}} \text{d}\theta \\
&\xrightarrow{d} \int_0^{2\pi} \frac{g(\theta) e^{- 2\sum_{j=1}^k \frac{\mathcal{N}_j \mp \frac{\eta_j}{\sqrt{j}}}{\sqrt{j}} \left( \alpha \cos(j\theta) - i\beta \sin(j\theta) \right)}}{e^{\pm 2\sum_{j=1}^k \frac{\eta_j}{j} \Re\left( (\alpha-\beta)e^{ij\theta} \right) + \sum_{j=1}^k \frac{2}{j} \Re\left((\alpha-\beta)e^{ij\theta}\right)^2}} \text{d}\theta \\
\overset{d}{=}& \int_0^{2\pi} g(\theta) e^{ 2\sum_{j=1}^k \frac{\mathcal{N}_j}{\sqrt{j}} \left( \alpha \cos(j\theta) - i\beta \sin(j\theta) \right) -\sum_{j=1}^k \frac{2}{j} \Re\left((\alpha-\beta)e^{ij\theta}\right)^2} \text{d}\theta \\
=& \int_0^{2\pi} g(\theta) \mu_{\alpha,\beta}^{(k)}(\text{d}\theta), \\
\end{split}
\end{align} 
where in the second equality we used (\ref{eqn:sigma4}), (\ref{eqn:odd sigma4}), and (\ref{eqn:even sigma4}), where the convergence in distribution follows from Theorem \ref{thm:traces} and the continuous mapping theorem, and where the penultimate equality follows from the fact that $- \mathcal{N}_j \overset{d}{=} \mathcal{N}_j$.

\section{Riemann-Hilbert Problem for a System of Orthogonal Polynomials and a Differential Identity}
In this section and and the following Sections \ref{section:asymptotics of polynomials}, \ref{section:Toeplitz} and \ref{section:T+H} we prove Theorems \ref{thm:T uniform}, \ref{thm:T+H uniform} and \ref{thm:T, T+H extended}, by following the Riemann-Hilbert analysis of \cite{ClaeysKrasovsky}. 

\subsection{RHP for Orthogonal Polynomials}
By the integral representation for a Toeplitz-determinant and since $f_{p,t} > 0$ except at $z_0,...,z_5$, it holds that $D_{n}(f_{p,t}) \in (0,\infty)$ for all $n\in \mathbb{N}$. Thus we can define the polynomials
\begin{align}
\begin{split}
\phi_n(z) &= \frac{1}{\sqrt{D_n(f_{p,t}) D_{n+1}(f_{p,t})}} 
\left| 
\begin{array} {cccc} f_{p,t,0} & f_{p,t,-1} & \dots & f_{p,t,-n} \\
f_{p,t,1} & f_{p,t,0} & \dots & f_{p,t,-n+1} \\
\dots & \dots & & \dots \\
f_{p,t,n-1} & f_{p,t,n-2} & \dots & f_{p,t,-1} \\
1 & z & \dots & z^n 
\end{array}
\right| = \chi_n z^n + ..., \\
\hat{\phi}_n(z) &= \frac{1}{\sqrt{D_n(f_{p,t}) D_{n+1}(f_{p,t})}} 
\left| 
\begin{array} {ccccc} f_{p,t,0} & f_{p,t,-1} & \dots & f_{p,t,-n+1} & 1 \\
f_{p,t,1} & f_{p,t,0} & \dots & f_{p,t,-n+2} & z \\
\dots & \dots & & \dots & \dots \\
f_{p,t,n} & f_{p,t,n-1} & \dots & f_{p,t,1} & z^n 
\end{array}
\right| = \chi_n z^n + ..., 
\end{split}
\end{align}
where the leading coefficient $\chi_n$ is given by 
\begin{equation}
\chi_n = \sqrt{\frac{D_n(f_{p,t})}{D_{n+1}(f_{p,t,})}}.
\end{equation}
The above polynomials satisfy the orthogonality relations
\begin{align}
\begin{split}
\frac{1}{2\pi} \int_0^{2\pi} \phi_n(e^{i\theta}) e^{-ik\theta} f_{p,t}(e^{i\theta}) \text{d}\theta =& \chi_n^{-1} \delta_{nk}, \\
\frac{1}{2\pi} \int_0^{2\pi} \hat{\phi}_n(e^{-i\theta}) e^{ik\theta} f_{p,t}(e^{i\theta}) \text{d}\theta =& \chi_n^{-1} \delta_{nk},
\end{split}
\end{align}
for $k = 0,1,...,n$, which implies that they are orthonormal w.r.t. the weight $f_{p,t}$. \\

Let $C$ denote the unit circle, oriented counterclockwise. It can easily be verified that the matrix-valued function $Y(z) = Y(z;n,p,t)$ given by 
\begin{equation}
Y(z) = \left( \begin{array}{cc} \chi_n^{-1} \phi_n(z) & \chi_n^{-1} \int_C \frac{\phi_n(\xi)}{\xi - z} \frac{f_{p,t}(\xi) \text{d}{\xi}}{2\pi i \xi^n} \\
- \chi_{n-1}z^{n-1} \hat{\phi}_{n-1}(z^{-1}) &  -\chi_{n-1} \int_C \frac{\hat{\phi}_{n-1}(\xi^{-1})}{\xi - z} \frac{f_{p,t}(\xi) \text{d}{\xi}}{2\pi i \xi} \end{array} \right)
\end{equation}
is the unique solution of the following Riemann-Hilbert problem:\\

\noindent \textbf{RH problem for} $Y$

\begin{enumerate}[label=(\alph*)]
\item $Y:\mathbb{C}\setminus C \rightarrow \mathbb{C}^{2\times 2}$ is analytic.

\item The continuous boundary values of $Y$ from inside the unit circle, denoted $Y_+$, and from outside, denoted $Y_-$, exist on $C\setminus \{ z_0,...,z_5 \}$, and are related by the jump condition
\begin{equation}
Y_+(z) = Y_-(z) \left( \begin{array} {cc} 1 & z^{-n} f_{p,t}(z) \\ 0 & 1 \end{array} \right), \quad z \in C\setminus \{z_0,...,z_5\}.
\end{equation}

\item $Y(z) = (I + O(1/z)) \left( \begin{array}{cc} z^n & 0 \\ 0 & z^{-n} \end{array} \right)$, as $ z \rightarrow \infty$.

\item As $z \rightarrow z_k$, $z \in \mathbb{C}\setminus C$, $k= 0,...,5$, we have 
\begin{equation}
Y(z) = \left( \begin{array}{cc} O(1) & O(1) + O(|z-z_k|^{2\alpha_k}) \\ O(1) & O(1) + O(|z-z_k|^{2\alpha_k}) \end{array}\right), \quad \text{if } \alpha_k \neq 0,
\end{equation}
and 
\begin{equation}
Y(z) = \left( \begin{array}{cc} O(1) & O(\ln |z-z_k|) \\ O(1) & O(\ln |z-z_k| ) \end{array}\right), \quad \text{if } \alpha_k = 0.
\end{equation}
\end{enumerate}

\noindent From the RHP and Liouville's theorem it follows that $\det Y(z) = 1$ for all $z \in \mathbb{C}\setminus C$. Using this, one can see quickly that the solution is unique. 

We have $Y(z;n,p,t)_{21}(0) = \chi_{n-1}^2$ and $Y(z;n,p,t)_{11}(z) = \chi_n^{-1} \phi_n(z) = \Phi_n(z)$, thus if we know the asymptotics of $Y$, we know the asymptotics of $\Phi_n$, $\phi_n$ and $\chi_n$.

\subsection{Differential Identity}

The Fourier coefficients are differentiable in $t$, thus $\ln D_n(f_{p,t})$ is differentiable in $t$ for all $p \in (\epsilon, \pi - \epsilon)$ and $n\in \mathbb{N}$. We calculate:
\begin{equation}
\frac{\partial}{\partial t} \ln \left|z - e^{i(p-t)}\right|^{2\alpha_1} = \frac{\partial}{\partial t} \ln \left| 2\sin \frac{\theta - (p-t)}{2} \right|^{2\alpha_1} = \alpha_1 \cot \frac{\theta - (p-t)}{2} = i\alpha_1 \frac{z + e^{i(p-t)}}{z - e^{i(p-t)}}.
\end{equation}
Similarly we obtain 
\begin{align}
\begin{split}
\frac{\partial}{\partial t} \ln \left|z - e^{i(p+t)} \right|^{2\alpha_2} &= - i\alpha_2 \frac{z + e^{i(p+t)}}{z - e^{i(p+t)}},\\
\frac{\partial}{\partial t} \ln \left|z - e^{i(2\pi - (p+t))} \right|^{2\alpha_4} &= i\alpha_4 \frac{z + e^{i(2\pi - (p+t))}}{z - e^{i(2\pi - (p+t))}},\\
\frac{\partial}{\partial t} \ln \left|z - e^{i(2\pi - (p-t)} \right|^{2\alpha_2} &= - i\alpha_5 \frac{z + e^{i(2\pi - (p-t))}}{z - e^{i(2\pi - (p-t))}}.
\end{split}
\end{align}
Therefore we get
\begin{align} \label{eqn:f diff}
\begin{split}
\frac{\partial f(z)}{\partial t} =& if(z) \sum_{k = 1,2,4,5} q_k \left( \alpha_k \frac{z+z_k}{z-z_k} + \beta_k \right)\\ 
=& if(z) \sum_{k = 1,2,4,5} q_k \left( \alpha_k + \beta_k + \frac{2\alpha_k z_k}{z - z_k} \right) \\
=& if(z) \sum_{k = 1,2,4,5} q_k \left( \beta_k + \frac{2\alpha_k z_k}{z - z_k} \right) \\
\end{split}
\end{align}
where $q_k = 1$ for $k = 1,4$ and $q_k = -1$ for $k = 2,5$. In the last line we used that $\sum_{k = 1,2,4,5} q_k \alpha_k = 0$.\\

Set $\tilde{Y}(z) = Y(z)$ in a neighborhood of $z_k$ if $\alpha_k > 0$. If $\alpha_k < 0$ the second column of $Y$ has a term of order $(z - z_k)^{2\alpha_k}$, which explodes as $z \rightarrow z_k$. We set $\tilde{Y}_{j1} = Y_{j1}$, $j = 1,2$, $\tilde{Y}_{j2} = Y_{j2} - c_j(z-z_k)^{2\alpha_k}$ in a neighborhood of $z_k$, with $c_j$ such that $\tilde{Y}$ is bounded in that neighborhood. Then we have

\begin{proposition}
Let $n \in \mathbb{N}$ and $\alpha_k \neq 0$ for $k = 1,2,4,5$. Then the following differential identity holds:
\begin{align} \label{eqn:diff id}
\begin{split}
\frac{1}{i} \frac{\text{d}}{\text{d}t} \ln D_n(f_{p,t}) =& \sum_{k = 1,2,4,5} q_k \left( n\beta_k - 2\alpha_k z_k \left( \frac{\text{d}Y^{-1}}{\text{d}z} \tilde{Y} \right)_{22} (z_k) \right), 
\end{split}
\end{align}
with $q_k$ as above and $\left( \frac{\text{d}Y^{-1}}{\text{d}z} \tilde{Y} \right)_{22} (z_k) = \lim_{z \rightarrow z_k} \left(\frac{\text{d}Y^{-1}}{\text{d}z} \tilde{Y} \right)_{22} (z)$ with $z \rightarrow z_k$ non-tangentially to the unit circle.  
\end{proposition}

\noindent \textbf{Proof:} The proof for $\alpha_k \neq 0$, $k = 1,2,4,5$ works exactly like the proof of Proposition 2.1 in \cite{ClaeysKrasovsky}. We have to modify their (2.16), which we replace with our (\ref{eqn:f diff}). The singularities at $\pm 1$ are independent of $p$ and $t$ and thus always stay within $f_{p,t}$. \qed

\begin{remark}
As in Remark 2.2 in \cite{ClaeysKrasovsky} one can also get a differential identity for $\ln D_n(f_{p,t})$ in the case where $\alpha_k = 0$ for some $k \in \{1,2,4,5\}$, by letting those $\alpha_k$'s go to zero in (\ref{eqn:f diff}), which is continuous in $\alpha_k$ on both sides. 
\end{remark}

\section{Aymptotics of the Orthogonal Polynomials} \label{section:asymptotics of polynomials}

\subsection{Normalization of the RHP}

Set 
\begin{equation}
T(z) = \begin{cases} Y(z) \left( \begin{array}{cc} z^{-n} & 0 \\ 0 & z^n \end{array} \right), & |z| >1, \\ Y(z), & |z| < 1. \end{cases}
\end{equation}
Then by the RH conditions for $Y$, we obtain the following RH condition for $T$:\\

\noindent \textbf{RH problem for $T$}

\begin{enumerate}[label=(\alph*)]
\item $T:\mathbb{C}\setminus C \rightarrow \mathbb{C}^{2\times 2}$ is analytic.

\item The continuous boundary values of $T$ from the inside, $T_+$, and from outside, $T_-$, of the unit circle exist on $C\setminus \{ z_0,...,z_5 \}$, and are related by the jump condition
\begin{equation}
T_+(z) = T_-(z) \left( \begin{array} {cc} z^n & f_{p,t}(z) \\ 0 & z^{-n} \end{array} \right), \quad z \in C\setminus \{z_0,...,z_5\}.
\end{equation}

\item $T(z) = I + O(1/z), \quad \text{as } z \rightarrow \infty$.

\item As $z \rightarrow z_k$, $z \in \mathbb{C}\setminus C$, $k= 0,...,5$, we have 
\begin{equation}
T(z) = \left( \begin{array}{cc} O(1) & O(1) + O(|z-z_k|^{2\alpha_k}) \\ O(1) & O(1) + O(|z-z_k|^{2\alpha_k}) \end{array}\right), \quad \text{if } \alpha_k \neq 0,
\end{equation}
and 
\begin{equation}
T(z) = \left( \begin{array}{cc} O(1) & O(\ln |z-z_k|) \\ O(1) & O(\ln |z-z_k| ) \end{array}\right), \quad \text{if } \alpha_k = 0.
\end{equation}
\end{enumerate}

\subsection{Opening of the Lens}

Define the Szeg\"{o} function 
\begin{equation}
D(z) = \exp \left( \frac{1}{2\pi i} \int_C \frac{\ln f_{p,t}(\xi)}{\xi - z} \text{d}\xi\right),
\end{equation}
which is analytic inside and outside of $C$ and satisfies
\begin{equation} \label{eqn:D pm f}
D_+(z) = D_-(z) f_{p,t}(z), \quad z \in C\setminus \{z_0,...,z_5\}.
\end{equation}
We have (see (4.9)-(4.10) in \cite{DeiftItsKrasovsky}):
\begin{equation} \label{eqn:D in}
D(z) = e^{\sum_0^\infty V_j z^j} \prod_{k = 0}^5 \left( \frac{z - z_k}{z_ke^{i\pi}} \right)^{\alpha_k + \beta_k} =: D_{\text{in},p,t}(z), \quad |z| < 1,
\end{equation}
and
\begin{equation} \label{eqn:D out}
D(z) = e^{-\sum_{-\infty}^{-1} V_j z^j} \prod_{k = 0}^5 \left( \frac{z - z_k}{z} \right)^{-\alpha_k + \beta_k} =: D_{\text{out},p,t}(z), \quad |z| > 1,
\end{equation}
and thus
\begin{equation}
D_{\text{out},p,t}(z)^{-1} = e^{\sum_{-\infty}^{-1} V_j z^j} \prod_{k = 0}^5 \left( \frac{z - z_k}{z} \right)^{\alpha_k - \beta_k}.
\end{equation}
The branch of $(z-z_k)^{\alpha_k \pm \beta_k}$ is fixed by the condition that $\arg (z-z_k) = 2\pi$ on the line going from $z_k$ to the right parallel to the real axis, and the branch cut is the line $\theta = \theta_k$ going from $z = z_k = e^{i\theta}$ to infinity. For any $k$, the branch cut of the root $z^{\alpha_k - \beta_k}$ is the line $\theta = \theta_k$ from $z = 0$ to infinity, and $\theta_k < \arg z < \theta_k + 2\pi$. By (\ref{eqn:D pm f}) we have that 
\begin{equation}
f_{p,t}(e^{i\theta}) = D_{\text{in},p,t}(e^{i\theta}) D_{\text{out},p,t}(e^{i\theta})^{-1},
\end{equation}
and this function extends analytically to a neighborhood $\mathcal{S}$ of the unit circle with the 6 branch cuts $z_k\mathbb{R}^+ \cap \mathcal{S}$, $k = 0,...5$, which we orient away from zero. Then we obtain for the jumps of $f_{p,t}$:
\begin{align} \label{eqn:f jumps}
\begin{split}
f_{p,t+}(z) =& f_{p,t-}(z) e^{2\pi i (\alpha_j - \beta_j)}, \quad \text{on } z_j(0,1) \cap \mathcal{S},\\
f_{p,t+}(z) =& f_{p,t-}(z) e^{-2\pi i (\alpha_j + \beta_j)}, \quad \text{on } z_j(1,\infty) \cap \mathcal{S}.
\end{split}
\end{align}

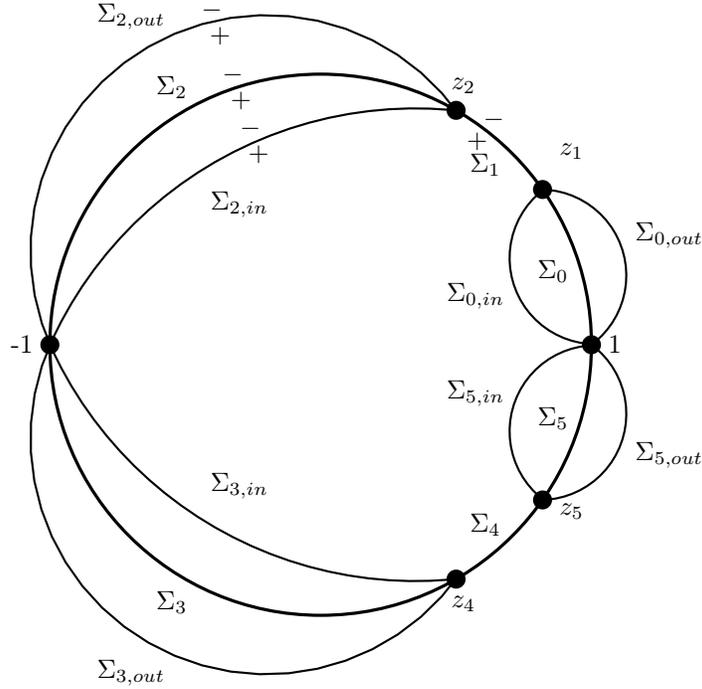
\begin{figure}[H] 
\centering
\begin{tikzpicture}[scale = 1.2]

\def\a{35} \def\b{60} \def\r{3}

\def\rsmallA{( (0.85*\r*sin(0.5*\a))^2 + (\r -0.85*\r*cos(0.5*\a))^2)^0.5}

\def\rsmallB{( (0.95*\r*sin(0.5*(\b-\a)))^2 + (\r - 0.95*\r*cos(0.5*(\b-\a)))^2)^0.5}

\def\rsmallC{( (0.4*\r*sin(0.5*(180-\b)))^2 + (\r - 0.4*\r*cos(0.5*(180-\b)))^2)^0.5}

\def\rbigC{( (0.7*\r*sin(0.5*(180-\b)))^2 + (\r + 0.7*\r*cos(0.5*(180-\b)))^2)^0.5}

\draw[name path=ellipse,black,very thick]
(0,0) circle[x radius = \r cm, y radius = \r cm];

\coordinate (1) at ({\r*cos(\a)}, {\r*sin(\a)});
\coordinate (2) at ({\r*cos(\b)}, {\r*sin(\b)});
\coordinate (4) at ({\r*cos(\b)}, {-\r*sin(\b)});
\coordinate (5) at ({\r*cos(\a)}, {-\r*sin(\a)});
	
\fill (1) circle (3pt) node[right,xshift=0.1cm,yshift=0.51cm] {$z_1$};
\fill (2) circle (3pt) node[above,xshift=0.1cm,yshift=0.1cm] {$z_2$};
\fill (4) circle (3pt) node[below,xshift=0.1cm,yshift=-0.1cm] {$z_4$};
\fill (5) circle (3pt) node[right,xshift=0.1cm,yshift=-0.15cm] {$z_5$};
\fill (\r,0) circle (3pt) node[right,xshift=0.1cm] {1};
\fill (-\r,0) circle (3pt) node[left,xshift=-0.1cm] {-1};

\draw [black,thick,domain=-50:85] plot ({0.85*\r*cos(0.5*\a) + \rsmallA*cos(\x)},{0.85*\r*sin(0.5*\a) + \rsmallA*sin(\x)});

\draw [black,thick,domain=131:265] plot ({1.065*\r*cos(0.5*\a) + \rsmallA*cos(\x)},{1.065*\r*sin(0.5*\a) + \rsmallA*sin(\x)});

\fill ({1.35*\r*cos(0.5*\a)},{1.35*\r*sin(0.5*\a)}) node[] {$\Sigma_{0,out}$};
\fill ({0.9*\r*cos(0.5*\a)},{0.9*\r*sin(0.5*\a)}) node[] {$\Sigma_{0}$};
\fill ({0.6*\r*cos(0.5*\a)},{0.6*\r*sin(0.5*\a)}) node[] {$\Sigma_{0,in}$};

\draw [black,thick,domain=-50:85] plot ({0.85*\r*cos(0.5*\a) + \rsmallA*cos(\x)},{0.85*\r*sin(-0.5*\a) + \rsmallA*sin(-\x)});

\draw [black,thick,domain=131:265] plot ({1.065*\r*cos(0.5*\a) + \rsmallA*cos(\x)},{1.065*\r*sin(-0.5*\a) + \rsmallA*sin(-\x)});

\fill ({1.35*\r*cos(0.5*\a)},{-1.35*\r*sin(0.5*\a)}) node[] {$\Sigma_{5,out}$};
\fill ({0.9*\r*cos(0.5*\a)},{-0.9*\r*sin(0.5*\a)}) node[] {$\Sigma_{5}$};
\fill ({0.6*\r*cos(0.5*\a)},{-0.6*\r*sin(0.5*\a)}) node[] {$\Sigma_{5,in}$};

\draw [black,thick,domain=35:205] plot ({0.4*\r*cos(90+0.5*\b) + \rsmallC*cos(\x)},{0.4*\r*sin(90+0.5*\b) + \rsmallC*sin(\x)});

\draw [black,thick,domain=85:155] plot ({-0.7*\r*cos(90+0.5*\b) + \rbigC*cos(\x)},{-0.7*\r*sin(90+0.5*\b) + \rbigC*sin(\x)});

\fill ({1.4*\r*cos(0.5*\b+90)},{1.4*\r*sin(0.5*\b+90)}) node[] {$\Sigma_{2,out}$};
\fill ({1.1*\r*cos(0.5*\b+90)},{1.1*\r*sin(0.5*\b+90)}) node[] {$\Sigma_{2}$};
\fill ({0.6*\r*cos(0.5*\b+90)},{0.6*\r*sin(0.5*\b+90)}) node[] {$\Sigma_{2,in}$};

\fill ({0.74*\r*cos(0.6*\b+0.4*180)},{0.74*\r*sin(0.6*\b+0.4*180)}) node[] {$+$};
\fill ({0.84*\r*cos(0.6*\b+0.4*180)},{0.84*\r*sin(0.6*\b+0.4*180)}) node[] {$-$};

\fill ({0.95*\r*cos(0.6*\b+0.4*180)},{0.95*\r*sin(0.6*\b+0.4*180)}) node[] {$+$};
\fill ({1.05*\r*cos(0.6*\b+0.4*180)},{1.05*\r*sin(0.6*\b+0.4*180)}) node[] {$-$};

\fill ({1.2*\r*cos(0.6*\b+0.4*180)},{1.2*\r*sin(0.6*\b+0.4*180)}) node[] {$+$};
\fill ({1.3*\r*cos(0.6*\b+0.4*180)},{1.3*\r*sin(0.6*\b+0.4*180)}) node[] {$-$};

\draw [black,thick,domain=35:205] plot ({0.4*\r*cos(90+0.5*\b) + \rsmallC*cos(\x)},{0.4*\r*sin(-(90+0.5*\b)) + \rsmallC*sin(-\x)});

\draw [black,thick,domain=85:155] plot ({-0.7*\r*cos(90+0.5*\b) + \rbigC*cos(\x)},{-0.7*\r*sin(-(90+0.5*\b)) + \rbigC*sin(-\x)});

\fill ({1.4*\r*cos(0.5*\b+90)},{-1.4*\r*sin(0.5*\b+90)}) node[] {$\Sigma_{3,out}$};
\fill ({1.1*\r*cos(0.5*\b+90)},{-1.1*\r*sin(0.5*\b+90)}) node[] {$\Sigma_{3}$};
\fill ({0.6*\r*cos(0.5*\b+90)},{-0.6*\r*sin(0.5*\b+90)}) node[] {$\Sigma_{3,in}$};

\fill ({0.9*\r*cos(0.5*(\b+\a))},{0.9*\r*sin(0.5*(\b+\a))}) node[] {$\Sigma_{1}$};

\fill ({0.95*\r*cos(0.7*\b+0.3*\a)},{0.95*\r*sin(0.7*\b+0.3*\a)}) node[] {$+$};
\fill ({1.05*\r*cos(0.7*\b+0.3*\a)},{1.05*\r*sin(0.7*\b+0.3*\a)}) node[] {$-$};

\fill ({0.9*\r*cos(0.5*(\b+\a))},{-0.9*\r*sin(0.5*(\b+\a))}) node[] {$\Sigma_{4}$};
\end{tikzpicture}
\caption{The jump contour $\Sigma_S$ of $S$.} \label{figure:S1}
\end{figure}

We factorize the jump matrix of $T$ as follows:
\begin{equation}
\left( \begin{array}{cc} z^n & f_{p,t}(z) \\ 0 & z^{-n} \end{array} \right) = \left( \begin{array}{cc} 1 & 0 \\ z^{-n} f_{p,t}(z)^{-1} & 1 \end{array} \right) \left( \begin{array}{cc} 0 & f_{p,t}(z) \\ - f_{p,t}(z)^{-1} & 0 \end{array} \right) \left( \begin{array}{cc} 1 & 0 \\ z^n f_{p,t}(z)^{-1} & 1 \end{array} \right).
\end{equation}
We then fix a lens-shaped region as in Figure \ref{figure:S1} and define 
\begin{equation}
S(z) = \begin{cases} T(z), & \text{outside the lens} \\ 
T(z) \left( \begin{array}{cc} 1 & 0 \\ z^{-n} f_{p,t}(z)^{-1} & 1 \end{array} \right), & \text{in the parts of the lenses outside the unit circle},\\
T(z) \left( \begin{array}{cc} 1 & 0 \\ -z^n f_{p,t}(z)^{-1} & 1 \end{array} \right), & \text{in the parts of the lenses inside the unit circle}.
\end{cases}
\end{equation}
The following RH conditions for $S$ can be verified directly:\\

\noindent \textbf{RH problem for $S$}
\begin{enumerate}[label=(\alph*)]
\item $S:\mathbb{C}\setminus \Sigma_S \rightarrow \mathbb{C}^{2\times 2}$ is analytic.\\

\item $S_+(z) = S_-(z) J_S(z)$ for $z \in \Sigma_S\setminus \{z_0,...,z_5\}$, where $J_S$ is given by  
\begin{equation}
J_S(z) = \begin{cases} \left( \begin{array}{cc} 1 & 0 \\ z^{-n} f_{p,t}(z)^{-1} & 1 \end{array} \right), & \text{on } \Sigma_{0,out} \cup \Sigma_{2,out} \cup \Sigma_{3,out} \cup \Sigma_{5,out}, \\
\left( \begin{array}{cc} 0 & f_{p,t}(z) \\ -f_{p,t}(z)^{-1} & 0 \end{array} \right), & \text{on } \Sigma_0 \cup \Sigma_2 \cup \Sigma_3 \cup \Sigma_5,\\
\left( \begin{array}{cc} 1 & 0 \\ z^n f_{p,t}(z)^{-1} & 1 \end{array} \right), & \text{on } \Sigma_{0,in} \cup \Sigma_{2,in} \cup \Sigma_{3,in} \cup \Sigma_{5,in}, \\
\left( \begin{array}{cc} z^n & f_{p,t}(z) \\ 0 & z^{-n} \end{array} \right), & \text{on } \Sigma_1 \cup \Sigma_4.
\end{cases}
\end{equation}

\item $S(z) = I + O(1/z), \quad \text{as } z \rightarrow \infty$.\\

\item As $z \rightarrow z_k$ from outside the lenses, $k= 0,...,5$, we have 
\begin{equation}
S(z) = \left( \begin{array}{cc} O(1) & O(1) + O(|z-z_k|^{2\alpha_k}) \\ O(1) & O(1) + O(|z-z_k|^{2\alpha_k}) \end{array}\right), \quad \text{if } \alpha_k \neq 0,
\end{equation}
and 
\begin{equation}
S(z) = \left( \begin{array}{cc} O(1) & O(\ln |z-z_k|) \\ O(1) & O(\ln |z-z_k| ) \end{array}\right), \quad \text{if } \alpha_k = 0.
\end{equation}
The behaviour of $S(z)$ as $z \rightarrow z_k$ from the other regions is obtained from these expressions by application of the appropriate jump conditions. 
\end{enumerate}

Fix $\delta_1,\delta_2 > 0$ such that the discs 
\begin{equation}
U_{\pm 1}: = \{z: |z-\pm 1| < \delta_1\}, \quad U_{\pm} := \{z: |z - e^{\pm ip}| < \delta_2\}
\end{equation}
are disjoint for any $p \in (\epsilon, \pi - \epsilon)$. Let $t_0 \in (0,\epsilon)$ such that $e^{i(p \pm t)} \in U_{+}$ and $e^{i(2\pi - (p \pm t))} \in U_{-}$ for one and hence for all $p \in (\epsilon, \pi - \epsilon)$. Then one observes that on the inner and out jump contours and outside of $U_1 \cup U_{-1} \cup U_{+} \cup U_{-}$ the jump matrix $J_S(z)$ converges to the identity matrix as $n\rightarrow \infty$, uniformly in $z$, $t < t_0$ and $p \in (\epsilon, \pi - \epsilon)$.

\subsection{Global Parametrix}

Define the function
\begin{align} \label{eqn:N}
\begin{split}
N(z) =& \begin{cases} \left( \begin{array}{cc} D_{\text{in},p,t}(z) & 0 \\ 0 & D_{\text{in},p,t}(z)^{-1} \end{array} \right) \left( \begin{array}{cc} 0 & 1 \\ -1 & 0 \end{array} \right), & |z| < 1, \\
\left( \begin{array}{cc} D_{\text{out},p,t}(z) & 0 \\ 0 & D_{\text{out},p,t}(z)^{-1} \end{array} \right), & |z| > 1. 
\end{cases}
\end{split}
\end{align}
One can easily verify that $N$ satisfies the following RH conditions:\\

\noindent \textbf{RH problem for $N$}
\begin{enumerate}[label=(\alph*)]
\item $N:\mathbb{C}\setminus{C} \rightarrow \mathbb{C}^{2\times 2}$ is analytic.\\

\item $N_+(z) = N_-(z)\left( \begin{array}{cc} 0 & f_{p,t}(z) \\ - f_{p,t}(z)^{-1} & 0 \end{array} \right)$ for $z \in C\setminus \{z_0,...,z_5\}$. \\

\item $N(z) = I + O(1/z)$ as $z \rightarrow \infty$. 
\end{enumerate} 

\subsection{Local Parametrix near $\pm 1$} \label{section:local plus minus 1}

The local parametrix near $\pm 1$ are constructed in exactly the same way as in \cite{DeiftItsKrasovsky}. We are looking for a solution of the following RHP:\\

\noindent \textbf{RH problem for $P_{\pm 1}(z)$}
\begin{enumerate}[label=(\alph*)]
\item $P_{\pm 1}: U_{\pm 1} \setminus \Sigma_S \rightarrow \mathbb{C}^{2\times 2}$ is analytic.

\item $P_{\pm 1}(z)_+ = P_{\pm 1}(z)_-J_S(z)$ for $z \in U_{\pm 1} \cap \Sigma_S$. 

\item As $z \rightarrow \pm 1$, $S(z)P_{\pm 1}(z)^{-1} = O(1)$. 

\item $P_{\pm 1}$ satisfies the matching condition $P_{\pm 1}(z) N^{-1}(z) = I + o(1)$ as $n \rightarrow \infty$, uniformly in $z \in \partial U_{\pm 1}$, $p \in (\epsilon, \pi - \epsilon)$ and $0 < t < t_0$.
\end{enumerate}

$P_{\pm 1}$ is given by (4.15), (4.23), (4.24), (4.47)-(4.50) in \cite{DeiftItsKrasovsky} and one can see from their construction that when all the other singularities are bounded away from $\pm 1$, then the matching condition is uniform in the location of the other singularities, i.e. holds uniformly in $p \in (\epsilon, \pi - \epsilon)$ and $0 < t < t_0$.

\subsection{$0 < t \leq \omega(n)/n$. Local Parametrices near $e^{\pm ip}$} \label{section:local 1}

Let $\omega(x)$ be a positive, smooth function for $x$ sufficiently large, s.t. 
\begin{equation}
\omega(x) \rightarrow \infty, \quad \omega(x) = o(x), \quad \text{as } x \rightarrow \infty.
\end{equation}
For $0 < t \leq 1/n$ and $1/n < t \leq \omega(n)/n$ we will construct local parametrices in $U_{\pm}$ which satisfy the same jump and growth conditions as $S$ inside $U_{\pm}$, and which match with the global parametrix $N$ on the boundaries $\partial U_{\pm}$ for large $n$. To be precise, we will construct $P_{\pm}$ satisfying the following conditions:\\

\noindent \textbf{RH problem for $P_{\pm}(z)$}
\begin{enumerate}[label=(\alph*)]
\item $P_{\pm}: U_{\pm} \setminus \Sigma_S \rightarrow \mathbb{C}^{2\times 2}$ is analytic.

\item $P_{\pm}(z)_+ = P_{\pm}(z)_-J_S(z)$ for $z \in U_{\pm} \cap \Sigma_S$. 

\item As $n \rightarrow \infty$, we have
\begin{equation}
P_{\pm}(z) N^{-1}(z) = (I+o(1)), \quad z \in \partial U_{\pm}, 
\end{equation}
uniformly for $p \in (\epsilon, \pi - \epsilon)$ and $0 < t < t_0$.

\item As $z \rightarrow z_k$, $S(z)P_{\pm}(z)^{-1} = O(1)$, $k = 1,2$ for $+$ and $k = 4,5$ for $-$. 
\end{enumerate}

\subsubsection{RH problem for $\Phi_\pm$}
Define
\begin{align} \label{eqn:Phi Psi}
\begin{split}
\Phi_{+} (\zeta,s) =& \Psi_+(\zeta,s) \begin{cases} 1, & -1 < \Im < 1, \\
e^{\pi i(\alpha_2 - \beta_2)\sigma_3}, & \Im \zeta > 1, \\
e^{-\pi i(\alpha_1 - \beta_1)\sigma_3}, & \Im \zeta < - 1,
\end{cases}\\
\Phi_{-} (\zeta,s) =& \Psi_+(\zeta,s) \begin{cases} 1, & -1 < \Im < 1, \\
e^{\pi i(\alpha_1 + \beta_1)\sigma_3}, & \Im \zeta > 1, \\
e^{-\pi i(\alpha_2 + \beta_2)\sigma_3}, & \Im \zeta < - 1,
\end{cases}
\end{split}
\end{align}
where $\Psi_{+}(\zeta,s)$ equals $\Psi(\zeta,s)$, defined in Appendix \ref{appendix:Psi}, and $\Psi_-(\zeta,s)$ equals $\Psi(\zeta,s)$ with $(\alpha_1,\alpha_2,\beta_1,\beta_2)$ in the appendix changed to $(\alpha_4,\alpha_5,\beta_4,\beta_5) = (\alpha_2, \alpha_1, -\beta_2, -\beta_1)$. The RH conditions for $\Phi_\pm$ follow directly from the RHP for $\Psi$. \\ 

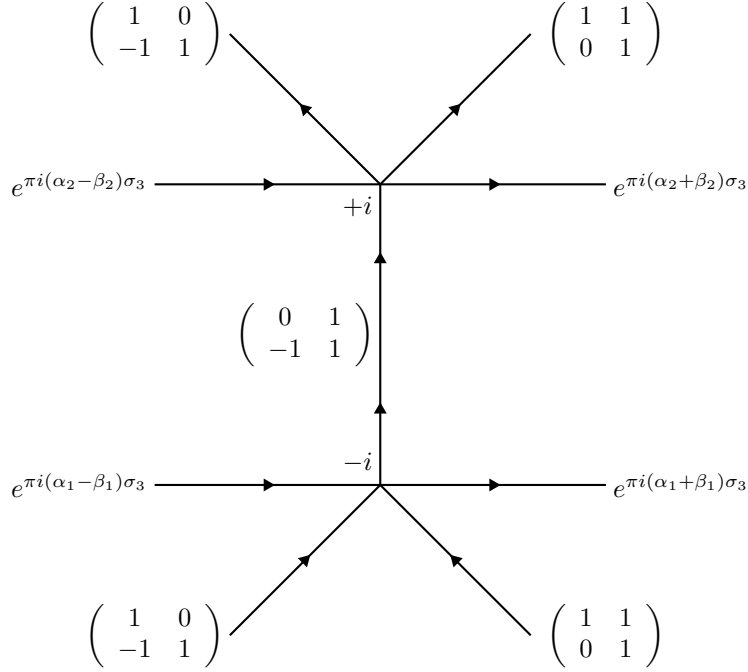
\begin{figure} [H]
\centering
\begin{tikzpicture}
\draw [thick] (0,-2) -- (0,2);
\draw [thick] (-3,-2) -- (3,-2);
\draw [thick] (-3,2) -- (3,2);
\draw [thick] (0,-2) -- (2,-4);
\draw [thick] (0,-2) -- (-2,-4);
\draw [thick] (0,2) -- (2,4);
\draw [thick] (0,2) -- (-2,4);

\fill (-0.3,1.7) node[] {$+i$};
\fill (-0.3,-1.7) node[] {$-i$};
\fill (-1,0) node[] {$\left( \begin{array}{cc} 0 & 1 \\ -1 & 1 \end{array} \right)$};
\fill (3,4) node[] {$\left( \begin{array}{cc} 1 & 1 \\ 0 & 1 \end{array} \right)$};
\fill (-3,4) node[] {$\left( \begin{array}{cc} 1 & 0 \\ -1 & 1 \end{array} \right)$};
\fill (4,2) node[] {$e^{\pi i(\alpha_2 + \beta_2) \sigma_3}$};
\fill (-4,2) node[] {$e^{\pi i(\alpha_2 - \beta_2) \sigma_3}$};
\fill (4,-2) node[] {$e^{\pi i(\alpha_1 + \beta_1) \sigma_3}$};
\fill (-4,-2) node[] {$e^{\pi i(\alpha_1 - \beta_1) \sigma_3}$};
\fill (3,-4) node[] {$\left( \begin{array}{cc} 1 & 1 \\ 0 & 1 \end{array} \right)$};
\fill (-3,-4) node[] {$\left( \begin{array}{cc} 1 & 0 \\ -1 & 1 \end{array} \right)$};

\node[fill=black,regular polygon, regular polygon sides=3,inner sep=1.pt, shape border rotate = -45] at (1,3) {};
\node[fill=black,regular polygon, regular polygon sides=3,inner sep=1.pt, shape border rotate = 45] at (-1,3) {};
\node[fill=black,regular polygon, regular polygon sides=3,inner sep=1.pt, shape border rotate = 0] at (0,1) {};
\node[fill=black,regular polygon, regular polygon sides=3,inner sep=1.pt, shape border rotate = 0] at (0,-1) {};
\node[fill=black,regular polygon, regular polygon sides=3,inner sep=1.pt, shape border rotate = 45] at (1,-3) {};
\node[fill=black,regular polygon, regular polygon sides=3,inner sep=1.pt, shape border rotate = -45] at (-1,-3) {};
\node[fill=black,regular polygon, regular polygon sides=3,inner sep=1.pt, shape border rotate = -90] at (1.5,2) {};
\node[fill=black,regular polygon, regular polygon sides=3,inner sep=1.pt, shape border rotate = -90] at (1.5,-2) {};
\node[fill=black,regular polygon, regular polygon sides=3,inner sep=1.pt, shape border rotate = -90] at (-1.5,2) {};
\node[fill=black,regular polygon, regular polygon sides=3,inner sep=1.pt, shape border rotate = -90] at (-1.5,-2) {};
\end{tikzpicture}
\caption{The jump contour $\Sigma$ and the jump matrices for $\Phi_+$. $\Phi_-$ has the same jump contour, while the jump matrices of $\Phi_-$ are given by replacing $(\alpha_1,\alpha_2, \beta_1,\beta_2)$ with $(\alpha_2,\alpha_1,-\beta_2,-\beta_1)$ in the jump matrices of $\Phi_+$.} \label{figure:Phi}
\end{figure}

\newpage 
\noindent \textbf{RH Problem for $\Phi_\pm$}

\begin{enumerate}[label = (\alph*)]
\item $\Phi_\pm:\mathbb{C} \setminus \Sigma \rightarrow \mathbb{C}^{2\times 2}$ is analytic, where 
\begin{align*}
\Sigma =& \cup_{k = 1}^9 \Sigma_k,  &&\Sigma_1 = i + e^{\frac{i\pi}{4}} \mathbb{R}_+, &&\Sigma_2 = i + e^{\frac{3i\pi}{4}\mathbb{R}_+} \\
\Sigma_3 =& i - \mathbb{R}_+, &&\Sigma_4 = -i - \mathbb{R}_+, &&\Sigma_5 = -i + e^{- \frac{3i\pi}{4}}\mathbb{R}_+, \\
\Sigma_6 =& -i + e^{\frac{i\pi}{4}} \mathbb{R}_+, &&\Sigma_7 = -i + \mathbb{R}_+, && \Sigma_8 = i + \mathbb{R}_+,\\
\Sigma_9 =& [-i,i],
\end{align*}
with the orientation chosen as in Figure \ref{figure:Phi} ("-" is always on the RHS of the contour).
\item $\Phi_+$ satisfies the jump conditions
\begin{equation}
\Phi_+(\zeta)_+ = \Phi_+(\zeta)_- V_k, \quad \zeta \in \Sigma_k,
\end{equation}
where 
\begin{align}
V_1 =& \left( \begin{array}{cc} 1 & 1 \\ 0 & 1 \end{array} \right), && V_2 = \left( \begin{array}{cc} 1 & 0 \\ - 1 & 1 \end{array} \right), \\
V_3 =& e^{\pi i(\alpha_2 - \beta_2)\sigma_3}, && V_4 = e^{\pi i(\alpha_1 + \beta_1)\sigma_3}, \\
V_5 =& \left( \begin{array}{cc} 1 & -1 \\ 0 & 1 \end{array} \right), && V_6 = \left( \begin{array}{cc} 1 & 1 \\ 0 & 1 \end{array} \right), \\
V_7 =& e^{\pi i(\alpha_2 + \beta_2)\sigma_3}, && V_8 = e^{\pi i(\alpha_1 + \beta_1)\sigma_3}, \\
V_9 =& \left( \begin{array}{cc} 0 & 1 \\ -1 & 1 \end{array} \right).
\end{align}
The jump conditions of $\Phi_-$ are given by replacing $(\alpha_1,\alpha_2,\beta_1,\beta_2)$ with $(\alpha_2,\alpha_1,-\beta_2,-\beta_1)$ in the jump matrices of $\Phi_+$.

\item We have in all regions:
\begin{equation} \label{eqn:Phi asymptotics}
\Phi_\pm(\zeta) = \left( I + \frac{\Psi_{\pm,1}}{\zeta} + \frac{\Psi_{\pm,2}}{\zeta^2} + O(\zeta^{-3}) \right) \hat{P}^{(\infty)}_\pm (\zeta) e^{-\frac{is}{4} \zeta \sigma_3}, \quad \text{as } \zeta \rightarrow \infty,
\end{equation}
where 
\begin{align}
\begin{split}
\hat{P}^{(\infty)}_+(\zeta,s) =& P^{(\infty)}_+(\zeta,s) \begin{cases} 1, & -1 < \Im < 1, \\
e^{\pi i(\alpha_2 - \beta_2)\sigma_3}, & \Im \zeta > 1, \\
e^{-\pi i(\alpha_1 - \beta_1)\sigma_3}, & \Im \zeta < - 1,
\end{cases}\\
\hat{P}^{(\infty)}_{-} (\zeta,s) =& P^{(\infty)}_+(\zeta,s) \begin{cases} 1, & -1 < \Im < 1, \\
e^{\pi i(\alpha_1 + \beta_1)\sigma_3}, & \Im \zeta > 1, \\
e^{-\pi i(\alpha_2 + \beta_2)\sigma_3}, & \Im \zeta < - 1,
\end{cases}
\end{split}
\end{align}
with 
\begin{align}
\begin{split}
P_+^{(\infty)}(\zeta,s) =& \left( \frac{is}{2} \right)^{-(\beta_1+\beta_2)\sigma_3} (\zeta - i)^{-\beta_2 \sigma_3} (\zeta + i)^{-\beta_1 \sigma_3}, \\
P_-^{(\infty)}(\zeta,s) =& \left( \frac{is}{2} \right)^{(\beta_1+\beta_2)\sigma_3} (\zeta - i)^{\beta_1 \sigma_3} (\zeta + i)^{\beta_2 \sigma_3},
\end{split}
\end{align}
with the branches corresponding to the arguments between $0$ and $2\pi$, and where $s \in -i\mathbb{R}_+$.

\item $\Phi_\pm$ has singular behaviour near $\pm i$ which is inherited from $\Psi$. The precise conditions follow from (\ref{eqn:Phi Psi}), (\ref{eqn:F_1 neq 0}), (\ref{eqn:F_2 neq 0}), (\ref{eqn:F_1 0}) and (\ref{eqn:F_2 0}). 
\end{enumerate}

\subsubsection{$0 < t \leq \omega(n)/n$. Construction of a Local Parametrix near $e^{\pm ip}$ in terms of $\Phi_\pm$}

We choose $P_{\pm}$ as in (7.21) in \cite{ClaeysKrasovsky}, i.e. 
\begin{equation} \label{eqn:P1}
P_{\pm}(z) = E_\pm (z) \Phi_{\pm} \left( \zeta;s \right) W_\pm(z), \quad \zeta = \frac{1}{t} \ln \frac{z}{e^{\pm ip}}, \quad s = -2int,
\end{equation}
\begin{itemize}
\item where $\Im \ln$ takes values in $(-\sigma, \sigma)$ for some $\sigma > 0$,
\item where $E_\pm$ is an analytic matrix-valued function in $U_{\pm}$, 
\item and where $W$ is given by 
\begin{align} \label{eqn:W}
\begin{split}
W_\pm(z) = \begin{cases} -z^{\frac{n}{2}\sigma_3} f_{p,t}(z)^{-\frac{\sigma_3}{2}} \sigma_3, & \text{for } |z| < 1, \\
z^{\frac{n}{2}\sigma_3} f_{p,t}(z)^{\frac{\sigma_3}{2}} \sigma_1, & \text{for } |z| > 1,
\end{cases}
\quad \sigma_1 = \left( \begin{array}{cc} 0 & 1 \\ 1 & 0 \end{array} \right), \quad \sigma_3 = \left( \begin{array}{cc} 1 & 0 \\ 0 & -1 \end{array} \right).
\end{split}
\end{align} 
\end{itemize}
The singularities $z = e^{i(p \pm t)}$ for $+$ and $z = e^{i(2\pi - (p \mp t))}$ for $-$ correspond to the values $\zeta = \pm i$. The jumps of $W_\pm$ follow from (\ref{eqn:f jumps}):
\begin{align}
\begin{split}
W_\pm(z)_+ =& W_\pm(z)_- \left( \begin{array}{cc} 0 & f_{p,t}(z) \\ - f_{p,t}^{-1}(z) & 0 \end{array} \right), \quad z \in C, \\
W_\pm(z)_+ =& W_\pm(z)_- e^{-\pi i(\alpha_j - \beta_j)\sigma_3}, \quad z \in z_j(0,1),\\
W_\pm(z)_+ =& W_\pm(z)_- e^{\pi i(\alpha_j + \beta_j)\sigma_3}, \quad z \in z_j(1,\infty).
\end{split}
\end{align}

Choose $\Sigma_S$ such that $\frac{1}{t} \ln \left( \frac{\Sigma_S \cap U_{\pm}}{e^{\pm ip}} \right) \subset \Sigma \cup i\mathbb{R}$, where $\Sigma$ is the contour of the RHP for $\Phi_\pm$, as shown in Figure \ref{figure:Phi}. Inside $U_{\pm}$ the combinded jumps of $W(z)$ and $\Phi$ are the same as the jumps of $S$:
\begin{align} \label{eqn:jumps}
\begin{split}
P_\pm(z)_+ =& E_\pm(z) \Psi_\pm(z)_- e^{\pi i(\alpha_j - \beta_j)\sigma_3} W_\pm(z)_-e^{-\pi i(\alpha_j - \beta_j)\sigma_3} = P_\pm(z)_-, \quad z \in z_j(0,1),\\
P_\pm(z)_+ =& E_\pm(z) \Psi_\pm(z)_- e^{\pi i(\alpha_j + \beta_j)\sigma_3} W_\pm(z)_-e^{\pi i(\alpha_j + \beta_j)\sigma_3} = P_\pm(z)_-, \quad z \in z_j(1,\infty),\\
P_\pm(z)_+ =& E_\pm(z) \Psi_\pm(z) W_\pm(z)_- \left( \begin{array}{cc} 0 & f_{p,t}(z) \\ - f_{p,t}^{-1}(z) & 0 \end{array} \right) \\
=& P_\pm(z)_- \left( \begin{array}{cc} 0 & f_{p,t}(z) \\ - f_{p,t}^{-1}(z) & 0 \end{array} \right), \quad z \in \Sigma_k \cap U_\pm, k = 0,2,3,5,\\  
P_\pm(z)_+ =& E_\pm(z) \Psi_\pm(z)_- \left( \begin{array}{cc} 0 & 1 \\ -1 & 1 \end{array} \right) W_\pm(z)_- \left( \begin{array}{cc} 0 & f_{p,t}(z) \\ - f_{p,t}^{-1}(z) & 0 \end{array} \right) \\
=& P_\pm(z)_- \left( \begin{array}{cc} z^n & f_{p,t}(z) \\ 0 & z^{-n} \end{array} \right), \quad z \in \Sigma_k \cap U_\pm, k = 1,4,\\
P_\pm(z)_+ =& E_\pm(z) \Psi_\pm(z)_- \left( \begin{array}{cc} 1 & 0 \\ -1 & 1 \end{array} \right) W_\pm(z)\\
=& P_\pm(z)_- \left( \begin{array}{cc} 1 & 0 \\ z^n f_{p,t}(z)^{-1} & 1 \end{array} \right), \quad z \in \Sigma_{k,in} \cap U_\pm, k = 0,2,3,5,\\
P_\pm(z)_+ =& E_\pm(z) \Psi_\pm(z)_- \left( \begin{array}{cc} 1 & 1 \\ 0 & 1 \end{array} \right) W_\pm(z)\\
=& P_\pm(z)_- \left( \begin{array}{cc} 1 & 0 \\ z^{-n} f_{p,t}(z)^{-1} & 1 \end{array} \right), \quad z \in \Sigma_{k,out} \cap U_\pm, k = 0,2,3,5.\\
\end{split}
\end{align} 
By the condition (d) of the RHP for $S$, the singular behaviour of $W$ near $z_k$, $k = 1,2,4,5$ and condition (d) of the RHP for $\Phi_\pm$, the singularities of $S(z)P_{\pm}(z)^{-1}$ at $z_1,z_2$ for $+$, and at $z_4,z_5$ for $-$, are removable.\\

What remains is to choose $E_\pm$ such that the matching condition (c) of $P_\pm$ holds. Define
\begin{equation} \label{eqn:E}
E_\pm(z) = \sigma_1 (D_{in,p,t}(z) D_{out,p,t}(z))^{-\frac{1}{2}\sigma_3} e^{\mp ip \frac{n}{2} \sigma_3} \hat{P}_\pm^{(\infty)} (\zeta,s)^{-1},
\end{equation}
From (\ref{eqn:D in}) and (\ref{eqn:D out}) one quickly sees that the branch cuts and singularities of $(D_{in,p,t}(z)D_{out,p,t}(z))^{-\frac{1}{2}\sigma_3}$ cancel out with those of $\hat{P}_\pm^{(\infty)}(z)^{-1}$, so that $E_\pm$ is analytic in $U_\pm$. \\

In exactly the same way as in the proof of Proposition 7.1 in \cite{ClaeysKrasovsky} one can see that the matching condition (c) is satisfied, i.e. we get: 
\begin{proposition} \label{prop:matching}
As $n \rightarrow \infty$ we have
\begin{align} \label{eqn:matching1}
\begin{split}
P_\pm(z) N(z)^{-1} =& (I+O(n^{-1})),
\end{split}
\end{align}
uniformly for $z \in \partial U_\pm$, $p \in (\epsilon, \pi - \epsilon)$ and $0 < t < t_0$ with $t_0$ sufficiently small.
\end{proposition}

\noindent \textbf{Proof:} Consider first the case where $c_0 \leq nt \leq C_0$, with some $c_0>0$ small and some $C_0>0$ large, which will be fixed below. Then $|\zeta| = |\frac{1}{t} \ln \frac{z}{e^{\pm ip}}| > \delta n$ for $z  \in \partial U_{\pm}$, and $s = -2int$ remains bounded and bounded away from zero. Thus by (\ref{eqn:Phi Psi}), (\ref{eqn:P1}) and (\ref{eqn:Psi asymptotics}) we have 
\begin{equation}
P(z) N(z)^{-1} = E_\pm(z) (I + O(n^{-1})) \hat{P}_\pm^{(\infty)}(\zeta,s) \left(\frac{z}{e^{\pm ip}} \right)^{-\frac{n}{2}\sigma_3} W_\pm(z) N(z)^{-1}, \quad z \in \partial U_\pm, \quad n \rightarrow \infty.  
\end{equation} 
Since the RHP for $\Psi$ is solvable for $c_0 \leq nt \leq C_0$, general properties of Painlev\'{e} RHPs imply that the error term is valid uniformly for $c_0 \leq nt \leq C_0$. By (\ref{eqn:N}) and (\ref{eqn:W}) we obtain
\begin{equation}
z^{-\frac{n}{2}\sigma_3} W_\pm(z) N(z)^{-1} = \left( D_{in,p,t}(z) D_{out,p,t}(z) \right)^{\frac{1}{2} \sigma_3} \sigma_1, \quad z \in U_\pm.
\end{equation}
Thus we have 
\begin{equation}
P_{\pm}(z) N(z)^{-1} = E_\pm(z)( I + (O(n^{-1})) E_\pm(z)^{-1},
\end{equation}
and since one can quickly see that $E_\pm(z)$ is bounded uniformly for $z \in \partial U_\pm$, $p \in (\epsilon, \pi - \epsilon)$ and $c_0 \leq nt \leq C_0$, we get that (\ref{eqn:matching1}) holds uniformly $z \in \partial U_\pm$, $p \in (\epsilon, \pi - \epsilon)$ and $c_0 \leq nt \leq C_0$.

Now consider the case $C_0 < nt < \omega(n)$. In this case we cannot use the expansion (\ref{eqn:Psi asymptotics}) since the argument $s$ of $\Psi_1$ is not bounded. Instead we need to use the large $|s| = 2nt$ asymptotics for $\Psi$, which were computed in \cite{ClaeysKrasovsky}[Section 5]. As is apparent from their (7.30) - (7.33), we have for $C_0$ sufficiently large
\begin{equation}
P_{\pm}(z) N(z)^{-1} = ( I + (O(n^{-1})), \quad n \rightarrow \infty, \quad z \in \partial U_\pm,
\end{equation}
uniformly for $C_0/n < t < t_0$, $z \in \partial U_\pm$ and $p \in (\epsilon, \pi - \epsilon)$. 

If $nt < c_0$, we can use the small $|s|$ asymptotics for $\Psi(\zeta;s)$ for large values of $\zeta = \frac{1}{t} \ln \frac{z}{e^{\pm ip}}$, as calculated in Section 6 of \cite{ClaeysKrasovsky}. From their (7.34) and (7.35) we see that 
\begin{equation}
P_{\pm}(z) N(z)^{-1} = ( I + (O(n^{-1})), \quad n \rightarrow \infty, \quad z \in \partial U_\pm,
\end{equation}
uniformly for $0 < t < C_0/n$, $z \in \partial U_\pm$ and $p \in (\epsilon, \pi - \epsilon)$. \qed

\subsubsection{$0 < t \leq \omega(n)/n$. Final Transformation}

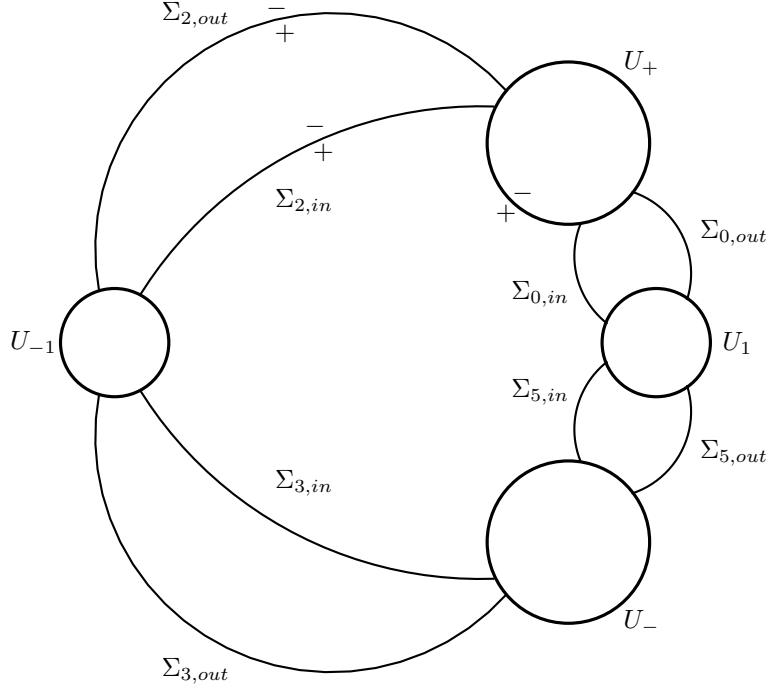
\begin{figure}
\centering
\begin{tikzpicture}[scale = 1.2]

\def\a{35} \def\b{60} \def\r{3}

\def\rsmallA{( (0.85*\r*sin(0.5*\a))^2 + (\r -0.85*\r*cos(0.5*\a))^2)^0.5}

\def\rsmallB{( (0.95*\r*sin(0.5*(\b-\a)))^2 + (\r - 0.95*\r*cos(0.5*(\b-\a)))^2)^0.5}

\def\rsmallC{( (0.4*\r*sin(0.5*(180-\b)))^2 + (\r - 0.4*\r*cos(0.5*(180-\b)))^2)^0.5}

\def\rbigC{( (0.7*\r*sin(0.5*(180-\b)))^2 + (\r + 0.7*\r*cos(0.5*(180-\b)))^2)^0.5}

\draw[name path=ellipse,black,very thick]
(\r,0) circle[x radius = 0.2*\r cm, y radius = 0.2*\r cm];

\fill ({1.3*\r},0) node[] {$U_1$};

\draw[name path=ellipse,black,very thick]
(-\r,0) circle[x radius = 0.2*\r cm, y radius = 0.2*\r cm];

\fill ({-1.3*\r},0) node[] {$U_{-1}$};

\draw[name path=ellipse,black,very thick]
({\r*cos(0.5*(\a+\b))},{\r*sin(0.5*(\a+\b))}) circle[x radius = 0.3*\r cm, y radius = 0.3*\r cm];

\fill ({1.40*\r*cos(0.5*(\b+\a))},{1.40*\r*sin(0.5*(\b+\a))}) node[] {$U_+$};
\fill ({0.75*\r*cos(0.5*(\b+\a))},{0.75*\r*sin(0.5*(\b+\a))}) node[] {$-$}; 
\fill ({0.66*\r*cos(0.5*(\b+\a))},{0.65*\r*sin(0.5*(\b+\a))}) node[] {$+$}; 

\draw[name path=ellipse,black,very thick]
({\r*cos(0.5*(\a+\b))},{-\r*sin(0.5*(\a+\b))}) circle[x radius = 0.3*\r cm, y radius = 0.3*\r cm];

\fill ({1.40*\r*cos(0.5*(\b+\a))},{-1.40*\r*sin(0.5*(\b+\a))}) node[] {$U_-$}; 

\draw [black,thick,domain=-18:72] plot ({0.85*\r*cos(0.5*\a) + \rsmallA*cos(\x)},{0.85*\r*sin(0.5*\a) + \rsmallA*sin(\x)});

\draw [black,thick,domain=158:232] plot ({1.065*\r*cos(0.5*\a) + \rsmallA*cos(\x)},{1.065*\r*sin(0.5*\a) + \rsmallA*sin(\x)});

\fill ({1.35*\r*cos(0.5*\a)},{1.35*\r*sin(0.5*\a)}) node[] {$\Sigma_{0,out}$};
\fill ({0.6*\r*cos(0.5*\a)},{0.6*\r*sin(0.5*\a)}) node[] {$\Sigma_{0,in}$};

\draw [black,thick,domain=-18:72] plot ({0.85*\r*cos(0.5*\a) + \rsmallA*cos(\x)},{0.85*\r*sin(-0.5*\a) + \rsmallA*sin(-\x)});

\draw [black,thick,domain=158:232] plot ({1.065*\r*cos(0.5*\a) + \rsmallA*cos(\x)},{1.065*\r*sin(-0.5*\a) + \rsmallA*sin(-\x)});

\fill ({1.35*\r*cos(0.5*\a)},{-1.35*\r*sin(0.5*\a)}) node[] {$\Sigma_{5,out}$};
\fill ({0.6*\r*cos(0.5*\a)},{-0.6*\r*sin(0.5*\a)}) node[] {$\Sigma_{5,in}$};

\draw [black,thick,domain=42:190.5] plot ({0.4*\r*cos(90+0.5*\b) + \rsmallC*cos(\x)},{0.4*\r*sin(90+0.5*\b) + \rsmallC*sin(\x)});

\draw [black,thick,domain=88:148] plot ({-0.7*\r*cos(90+0.5*\b) + \rbigC*cos(\x)},{-0.7*\r*sin(90+0.5*\b) + \rbigC*sin(\x)});

\fill ({1.4*\r*cos(0.5*\b+90)},{1.4*\r*sin(0.5*\b+90)}) node[] {$\Sigma_{2,out}$};
\fill ({0.6*\r*cos(0.5*\b+90)},{0.6*\r*sin(0.5*\b+90)}) node[] {$\Sigma_{2,in}$};

\fill ({0.74*\r*cos(0.6*\b+0.4*180)},{0.74*\r*sin(0.6*\b+0.4*180)}) node[] {$+$};
\fill ({0.84*\r*cos(0.6*\b+0.4*180)},{0.84*\r*sin(0.6*\b+0.4*180)}) node[] {$-$};

\fill ({1.2*\r*cos(0.6*\b+0.4*180)},{1.2*\r*sin(0.6*\b+0.4*180)}) node[] {$+$};
\fill ({1.3*\r*cos(0.6*\b+0.4*180)},{1.3*\r*sin(0.6*\b+0.4*180)}) node[] {$-$};

\draw [black,thick,domain=42:190.5] plot ({0.4*\r*cos(90+0.5*\b) + \rsmallC*cos(\x)},{0.4*\r*sin(-(90+0.5*\b)) + \rsmallC*sin(-\x)});

\draw [black,thick,domain=88:148] plot ({-0.7*\r*cos(90+0.5*\b) + \rbigC*cos(\x)},{-0.7*\r*sin(-(90+0.5*\b)) + \rbigC*sin(-\x)});

\fill ({1.4*\r*cos(0.5*\b+90)},{-1.4*\r*sin(0.5*\b+90)}) node[] {$\Sigma_{3,out}$};
\fill ({0.6*\r*cos(0.5*\b+90)},{-0.6*\r*sin(0.5*\b+90)}) node[] {$\Sigma_{3,in}$};
\end{tikzpicture}
\caption{The jump contour $\Sigma_R$ of $R$.} \label{figure:R1}
\end{figure}

Define 
\begin{equation} \label{eqn:R}
R(z) = \begin{cases} S(z)N^{-1}(z) & z \in U_\infty \setminus \Sigma_S, \quad U_\infty := \mathbb{C}\setminus \text{local parametrices}, \\
S(z)P_{\pm }^{-1}(z), & z \in U_{\pm}\setminus \Sigma_S, \\
S(z)P_{\pm 1}^{-1}(z), & z \in U_{\pm 1} \setminus \Sigma_S.
\end{cases}
\end{equation}
Then $R$ solves the following RHP:\\
\newpage
\noindent \textbf{RH problem for $R$}
\begin{enumerate}[label=(\alph*)]
\item $R:\mathbb{C}\setminus \Sigma_R \rightarrow \mathbb{C}^{2\times 2}$ is analytic, where $\Sigma_R$ is shown in Figure \ref{figure:R1}\\

\item $R(z)$ has the following jumps: 
\begin{align}
\begin{split}
R_+(z) &=  R_-(z)N(z)\left( \begin{array}{cc} 1 & 0 \\ f_{p,t}(z)^{-1}z^{-n} & 1 \end{array} \right) N(z)^{-1}, \quad z \in \Sigma_{j,out}, \, j = 0,2,3,5, \\
R_+(z) &=  R_-(z)N(z)\left( \begin{array}{cc} 1 & 0 \\ f_{p,t}(z)^{-1}z^{n} & 1 \end{array} \right) N(z)^{-1}, \quad z \in \Sigma_{j,in}, \, j = 0,2,3,5, \\
R_+(z) &=  R_-(z) P_{\pm}(z) N(z)^{-1}, \quad z \in \partial U_{\pm} \setminus \text{intersection points}, \\
R_+(z) &=  R_-(z) P_{\pm 1}(z) N(z)^{-1}, \quad z \in \partial U_{\pm 1} \setminus \text{intersection points}. \\
\end{split}
\end{align}

\item $R(z) = I + O(1/z)$ as $z \rightarrow \infty$. 
\end{enumerate}

One quickly sees that uniformly in $z \in \Sigma_{j,out} \cup \Sigma_{j,in} \setminus \overline{U_\infty}$ we have
\begin{equation}
R_+(z) = R_-(z)(I+O(e^{-\delta n}))
\end{equation}
for some $\delta > 0$ and uniformly in $p \in (\epsilon, \pi - \epsilon)$, $0 < t < t_0$.
By Proposition \ref{prop:matching} we have that 
\begin{equation}
R_+(z) = R_-(z) P_{\pm}(z) N(z)^{-1} = R_-(z)(I+O(n^{-1})),
\end{equation}
uniformly for $z \in \partial U_{\pm}$, $p \in (\epsilon, \pi - \epsilon)$ and $0 < t < t_0$. Because of the matching condition (d) of $P_{\pm 1}$ we have 
\begin{equation}
R_+(z) = R_-(z) P_{\pm 1}(z) N(z)^{-1} = R_-(z)(I+O(n^{-1})),
\end{equation}
uniformly for $z \in \partial U_{\pm 1}$, $p \in (\epsilon, \pi - \epsilon)$ and $0 < t < t_0$.\\

We see that we have a normalized RHP with small jumps, which by the standard theory on RHP implies that 
\begin{equation}
R(z) = I + O(n^{-1}), \quad \frac{\text{d}R(z)}{\text{d}z} = O(n^{-1}),
\end{equation}
as $n \rightarrow \infty$, uniformly for $z$ off the jump contour and uniformly in $p \in (\epsilon, \pi - \epsilon)$, $0 < t < t_0$.

\subsection{$\omega(n)/n < t < t_0$. Local Parametrices near $e^{\pm ip}$} \label{section:local 2}
We now transfer the construction from Section 7.5 in \cite{ClaeysKrasovsky} to our setting in a completely straightforward manner. Although the parametrices $P_{\pm}$ from the previous section are valid for the whole region $0 < t < t_0$ we need to construct more explicit parametrices for the case $\omega(n)/n < t < t_0$ to get a simpler large $n$ expansion for $Y$, which is needed for the analysis in the next section. \\

In the case $\omega(n)/n < t < t_0$ $\zeta = \frac{1}{t} \ln \frac{z}{e^{\pm ip}}$ is not necessarily large on $\partial U_{\pm}$. But we can construct a large $s = -int$ expansion for $Y$, as $|s| = nt$ is large. \\

We modify the $S$-RHP by now also opening up lenses around the arcs $(p-t,p+t)$ and $(2\pi - p - t, 2\pi - p + t)$, i.e. we choose the contour $\Sigma_S$ as in Figure \ref{figure:S2}.

\begin{figure}[H] 
\centering
\begin{tikzpicture}[scale = 1.2]

\def\a{35} \def\b{60} \def\r{3}

\def\rsmallA{( (0.85*\r*sin(0.5*\a))^2 + (\r -0.85*\r*cos(0.5*\a))^2)^0.5}

\def\rsmallB{( (0.95*\r*sin(0.5*(\b-\a)))^2 + (\r - 0.95*\r*cos(0.5*(\b-\a)))^2)^0.5}

\def\rsmallC{( (0.4*\r*sin(0.5*(180-\b)))^2 + (\r - 0.4*\r*cos(0.5*(180-\b)))^2)^0.5}

\def\rbigC{( (0.7*\r*sin(0.5*(180-\b)))^2 + (\r + 0.7*\r*cos(0.5*(180-\b)))^2)^0.5}

\draw[name path=ellipse,black,very thick]
(0,0) circle[x radius = \r cm, y radius = \r cm];

\coordinate (1) at ({\r*cos(\a)}, {\r*sin(\a)});
\coordinate (2) at ({\r*cos(\b)}, {\r*sin(\b)});
\coordinate (4) at ({\r*cos(\b)}, {-\r*sin(\b)});
\coordinate (5) at ({\r*cos(\a)}, {-\r*sin(\a)});
	
\fill (1) circle (3pt) node[right,xshift=0.1cm,yshift=0.51cm] {$z_1$};
\fill (2) circle (3pt) node[above,xshift=0.1cm,yshift=0.1cm] {$z_2$};
\fill (4) circle (3pt) node[below,xshift=0.1cm,yshift=-0.1cm] {$z_4$};
\fill (5) circle (3pt) node[right,xshift=0.1cm,yshift=-0.15cm] {$z_5$};
\fill (\r,0) circle (3pt) node[right,xshift=0.1cm] {1};
\fill (-\r,0) circle (3pt) node[left,xshift=-0.1cm] {-1};

\draw [black,thick,domain=-50:85] plot ({0.85*\r*cos(0.5*\a) + \rsmallA*cos(\x)},{0.85*\r*sin(0.5*\a) + \rsmallA*sin(\x)});

\draw [black,thick,domain=131:265] plot ({1.065*\r*cos(0.5*\a) + \rsmallA*cos(\x)},{1.065*\r*sin(0.5*\a) + \rsmallA*sin(\x)});

\fill ({1.35*\r*cos(0.5*\a)},{1.35*\r*sin(0.5*\a)}) node[] {$\Sigma_{0,out}$};
\fill ({0.9*\r*cos(0.5*\a)},{0.9*\r*sin(0.5*\a)}) node[] {$\Sigma_{0}$};
\fill ({0.6*\r*cos(0.5*\a)},{0.6*\r*sin(0.5*\a)}) node[] {$\Sigma_{0,in}$};

\draw [black,thick,domain=-50:85] plot ({0.85*\r*cos(0.5*\a) + \rsmallA*cos(\x)},{0.85*\r*sin(-0.5*\a) + \rsmallA*sin(-\x)});

\draw [black,thick,domain=131:265] plot ({1.065*\r*cos(0.5*\a) + \rsmallA*cos(\x)},{1.065*\r*sin(-0.5*\a) + \rsmallA*sin(-\x)});

\fill ({1.35*\r*cos(0.5*\a)},{-1.35*\r*sin(0.5*\a)}) node[] {$\Sigma_{5,out}$};
\fill ({0.9*\r*cos(0.5*\a)},{-0.9*\r*sin(0.5*\a)}) node[] {$\Sigma_{5}$};
\fill ({0.6*\r*cos(0.5*\a)},{-0.6*\r*sin(0.5*\a)}) node[] {$\Sigma_{5,in}$};

\draw [black,thick,domain=35:205] plot ({0.4*\r*cos(90+0.5*\b) + \rsmallC*cos(\x)},{0.4*\r*sin(90+0.5*\b) + \rsmallC*sin(\x)});

\draw [black,thick,domain=85:155] plot ({-0.7*\r*cos(90+0.5*\b) + \rbigC*cos(\x)},{-0.7*\r*sin(90+0.5*\b) + \rbigC*sin(\x)});

\fill ({1.4*\r*cos(0.5*\b+90)},{1.4*\r*sin(0.5*\b+90)}) node[] {$\Sigma_{2,out}$};
\fill ({1.1*\r*cos(0.5*\b+90)},{1.1*\r*sin(0.5*\b+90)}) node[] {$\Sigma_{2}$};
\fill ({0.6*\r*cos(0.5*\b+90)},{0.6*\r*sin(0.5*\b+90)}) node[] {$\Sigma_{2,in}$};

\fill ({0.74*\r*cos(0.6*\b+0.4*180)},{0.74*\r*sin(0.6*\b+0.4*180)}) node[] {$+$};
\fill ({0.84*\r*cos(0.6*\b+0.4*180)},{0.84*\r*sin(0.6*\b+0.4*180)}) node[] {$-$};

\fill ({0.95*\r*cos(0.6*\b+0.4*180)},{0.95*\r*sin(0.6*\b+0.4*180)}) node[] {$+$};
\fill ({1.05*\r*cos(0.6*\b+0.4*180)},{1.05*\r*sin(0.6*\b+0.4*180)}) node[] {$-$};

\fill ({1.2*\r*cos(0.6*\b+0.4*180)},{1.2*\r*sin(0.6*\b+0.4*180)}) node[] {$+$};
\fill ({1.3*\r*cos(0.6*\b+0.4*180)},{1.3*\r*sin(0.6*\b+0.4*180)}) node[] {$-$};

\draw [black,thick,domain=35:205] plot ({0.4*\r*cos(90+0.5*\b) + \rsmallC*cos(\x)},{0.4*\r*sin(-(90+0.5*\b)) + \rsmallC*sin(-\x)});

\draw [black,thick,domain=85:155] plot ({-0.7*\r*cos(90+0.5*\b) + \rbigC*cos(\x)},{-0.7*\r*sin(-(90+0.5*\b)) + \rbigC*sin(-\x)});

\fill ({1.4*\r*cos(0.5*\b+90)},{-1.4*\r*sin(0.5*\b+90)}) node[] {$\Sigma_{3,out}$};
\fill ({1.1*\r*cos(0.5*\b+90)},{-1.1*\r*sin(0.5*\b+90)}) node[] {$\Sigma_{3}$};
\fill ({0.6*\r*cos(0.5*\b+90)},{-0.6*\r*sin(0.5*\b+90)}) node[] {$\Sigma_{3,in}$};

\draw [black,thick,domain=-30:128] plot ({0.95*\r*cos(0.5*(\a+\b)) + \rsmallB*cos(\x)},{0.95*\r*sin(0.5*(\a+\b)) + \rsmallB*sin(\x)});

\draw [black,thick,domain=142:300] plot ({1.02*\r*cos(0.5*(\a+\b))+\rsmallB*cos(\x)},{1.02*\r*sin(0.5*(\a+\b)) + \rsmallB*sin(\x)});

\fill ({1.3*\r*cos(0.5*(\b+\a))},{1.3*\r*sin(0.5*(\b+\a))}) node[] {$\Sigma_{1,out}$};
\fill ({0.9*\r*cos(0.5*(\b+\a))},{0.9*\r*sin(0.5*(\b+\a))}) node[] {$\Sigma_{1}$};
\fill ({0.7*\r*cos(0.5*(\b+\a))},{0.7*\r*sin(0.5*(\b+\a))}) node[] {$\Sigma_{1,in}$};

\draw [black,thick,domain=-30:128] plot ({0.95*\r*cos(0.5*(\a+\b)) + \rsmallB*cos(\x)},{-0.95*\r*sin(0.5*(\a+\b)) + \rsmallB*sin(-\x)});

\draw [black,thick,domain=142:300] plot ({1.02*\r*cos(0.5*(\a+\b))+\rsmallB*cos(\x)},{-1.02*\r*sin(0.5*(\a+\b)) + \rsmallB*sin(-\x)});

\fill ({1.3*\r*cos(0.5*(\b+\a))},{-1.3*\r*sin(0.5*(\b+\a))}) node[] {$\Sigma_{4,out}$};
\fill ({0.9*\r*cos(0.5*(\b+\a))},{-0.9*\r*sin(0.5*(\b+\a))}) node[] {$\Sigma_{4}$};
\fill ({0.7*\r*cos(0.5*(\b+\a))},{-0.7*\r*sin(0.5*(\b+\a))}) node[] {$\Sigma_{4,in}$};
\end{tikzpicture}
\caption{The modified jump contour $\Sigma_S$ of $S$ in the case $\omega(n)/n < t < t_0$. The difference compared to Figure \ref{figure:S1} is that here there are also lenses around the arcs $(p-t,p+t)$ and $(2\pi - p - t, 2\pi - p + t)$.} \label{figure:S2}
\end{figure}
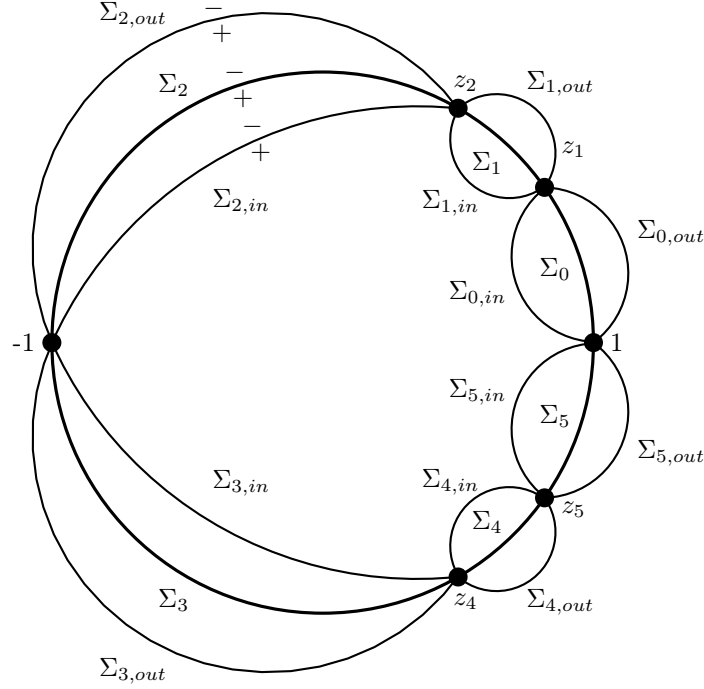

The points $z_0 = 1$, $z_3 = -1$ we surround with disks $U_1$, $U_{-1}$, small enough such for all $p \in (\epsilon, \pi - \epsilon)$ and $\omega(n)/n < t < t_0$ they are disjoint with the neighborhoods $\tilde{\mathcal{U}}_1, \tilde{\mathcal{U}}_2, \tilde{\mathcal{U}}_4, \tilde{\mathcal{U}}_5$ defined in the next paragraph, and we take the same local parametrices $P_{\pm 1}$ in $U_1$, $U_{-1}$ as in Section \ref{section:local plus minus 1}.\\ 

Let $\mathcal{U}_1, \mathcal{U}_2$ be small non-intersecting disks around $\pm i$, those are the same neighborhoods as in Section 5 of \cite{ClaeysKrasovsky}. We surround the points $z_1 = e^{i(p-t)}, z_2 = e^{i(p+t)}$ by small neighborhoods $\tilde{\mathcal{U}}_1, \tilde{\mathcal{U}}_2$, with $\tilde{\mathcal{U}}_1$ being the image of $\mathcal{U}_2$ under the inverse of the map $\zeta = \frac{1}{t} \ln \frac{z}{e^{ip}}$, and $\tilde{\mathcal{U}}_2$ being the image of $\mathcal{U}_1$ under the same map. Similarly we surround $z_4 = e^{i(2\pi - (p+t))}$ by a small neighborhood $\tilde{\mathcal{U}}_4$, which is the image of $\mathcal{U}_2$ under the inverse of the map $\zeta = \frac{1}{t} \ln \frac{z}{e^{-ip}}$, and $z_5 = e^{i(2\pi - (p-t))}$ we surround by $\tilde{\mathcal{U}}_5$ which is the image of $\mathcal{U}_1$ under the same map. Since the disks $\mathcal{U}_1$, $\mathcal{U}_2$ are fixed in the $\zeta$-plane, the neighborhoods $\tilde{\mathcal{U}}_1, \tilde{\mathcal{U}}_2, \tilde{\mathcal{U}}_4,\tilde{\mathcal{U}}_5$ contract in the $z$-plane if $t$ decreases with $n$.  \\

As global parametrix outside these neighborhoods we choose $N(z)$ as in the previous section. For $k = 1,2,4,5$ we choose the local parametrices in $\tilde{\mathcal{U}}_k$ as follows:
\begin{align} \label{eqn:P tilde}
\begin{split}
\tilde{P}_1(z) =& \tilde{E}_1(z) M^{(\alpha_1,\beta_1)}(nt(\zeta(z) + i)) \Omega_1(z) W_+(z), \quad \zeta = \frac{1}{t} \ln \frac{z}{e^{ip}}, \\
\tilde{P}_2(z) =& \tilde{E}_2(z) M^{(\alpha_2,\beta_2)}(nt(\zeta(z) - i)) \Omega_2(z) W_+(z), \quad \zeta = \frac{1}{t} \ln \frac{z}{e^{ip}}, \\
\tilde{P}_4(z) =& \tilde{E}_4(z) M^{(\alpha_4,\beta_4)}(nt(\zeta(z) + i)) \Omega_4(z) W_-(z), \quad \zeta = \frac{1}{t} \ln \frac{z}{e^{-ip}}, \\
\tilde{P}_5(z) =& \tilde{E}_5(z) M^{(\alpha_5,\beta_5)}(nt(\zeta(z) - i)) \Omega_5(z) W_-(z), \quad \zeta = \frac{1}{t} \ln \frac{z}{e^{-ip}},
\end{split}
\end{align}
where
\begin{align}
\tilde{E}_k(z) = \sigma_1 \left( D_{in,p,t}(z) D_{out,p,t}(z) \right)^{-\sigma_3/2} \Omega_k(z) (nt(\zeta \pm i))^{\beta_k \sigma_3} z_k^{-\frac{n}{2}\sigma_3},
\end{align}
with $+$ for $k = 1,4$ and $-$ for $k = 2,5$, where $M^{(\alpha_k,\beta_k)}(\lambda)$ is given in Appendix \ref{appendix:M} with $\alpha = \alpha_k, \beta = \beta_k$, where
\begin{align}
\Omega_k(z) = \begin{cases} e^{i\frac{\pi}{2}(\alpha_k - \beta_k)\sigma_3} , & \Im \zeta > 1 \\
e^{-i\frac{\pi}{2}(\alpha_k - \beta_k)\sigma_3} , & \Im \zeta < 1 \end{cases},
\end{align}
and where $W_\pm(z)$ is given in (\ref{eqn:W}). \\

By (\ref{eqn:N}) and (\ref{eqn:W}) we obtain
\begin{equation}
z^{-\frac{n}{2}\sigma_3} W_\pm(z) N(z)^{-1} = \left( D_{in,p,t}(z) D_{out,p,t}(z) \right)^{\frac{1}{2} \sigma_3} \sigma_1, \quad z \in U_\pm.
\end{equation}
Using the large argument expansion (\ref{eqn:M asymptotics}) for $M^{(\alpha_1,\beta_1)}(nt(\zeta + i))$ for $z \in \partial \tilde{\mathcal{U}}_1$, we see that
\begin{align} \label{eqn:matching tilde}
\begin{split}
\tilde{P}_1(z) N(z)^{-1} =& \tilde{E}_1(z) \left( I + \frac{M^{(\alpha_1,\beta_1)}}{nt(\zeta + i)} + O((nt)^{-2} \right) \\
&\times (nt(\zeta + i))^{-\beta_1 \sigma_3} \left( \frac{z}{e^{i(p+t)}} \right)^{-\frac{n}{2}\sigma_3} \Omega_1(z) W_+(z) N(z)^{-1}\\
=& \tilde{E}_1(z) \left( I + \frac{M^{(\alpha_1,\beta_1)}}{nt(\zeta + i)} + O((nt)^{-2} \right) \\
&\times (nt(\zeta + i))^{-\beta_1 \sigma_3}  z_1^{\frac{n}{2}\sigma_3} \Omega_1(z) \left( D_{in,p,t}(z) D_{out,p,t}(z) \right)^{\frac{1}{2} \sigma_3} \sigma_1 \\
=& \tilde{E}_1(z) \left( I + O((nt)^{-1}) \right) \tilde{E}_1(z)^{-1} \\
=& \left( I + O((nt)^{-1}) \right),
\end{split}
\end{align}
uniformly in $z \in \partial \tilde{\mathcal{U}}_1$, $p \in (\epsilon, \pi - \epsilon)$ and $\omega(n)/n < t < t_0$, since $\tilde{E}_1(z)$ is uniformly bounded for $\omega(n)/n < t < t_0$, $p \in (\epsilon, \pi - \epsilon)$ and $z \in \tilde{\mathcal{U}}_1$. Similarly one obtains that for $k = 2,4,5$
\begin{align}
\tilde{P}_k(z) N(z)^{-1} = \left( I + O((nt)^{-1}) \right),
\end{align}
uniformly in $z \in \partial \tilde{\mathcal{U}}_k$, $p \in (\epsilon, \pi - \epsilon)$ and $\omega(n)/n < t < t_0$.\\

Choose $\Sigma_S$ such that $\frac{1}{t} \ln \left( \frac{\Sigma_S}{z_k} \right) \subset \left( e^{\pm \frac{\pi i}{4}} \mathbb{R} \cup i\mathbb{R} \cup \mathbb{R} \right)$ in $\tilde{\mathcal{U}}_k$. Then one can easily verify, as in (\ref{eqn:jumps}),  that $\tilde{P}_k$ has the same jumps as $S$ in $\tilde{\mathcal{U}}_k$, so that $S(z) \tilde{P}_k(z)^{-1}$ is meromorphic in $\tilde{\mathcal{U}}_k$, with at most an isolated singulary at $z_k$. The singular behaviour of $S$ and $W_\pm$ near $z_k$, and of $M^{(\alpha_k,\beta_k)}$ near $0$ (given in (\ref{eqn:M at 0, neq 0}) and (\ref{eqn:M at 0, 0})), imply that $S(z) \tilde{P}_k^{-1}(z)$ is bounded at $z_k$, which shows that that $\tilde{P}_k$ is a parametrix for $S$ in $\tilde{\mathcal{U}}_k$ with the matching condition (\ref{eqn:matching tilde}) with $N(z)$ at $\partial \tilde{\mathcal{U}}_k$.  

\subsubsection{$\omega(n)/n < t < t_0$. Final Transformation}
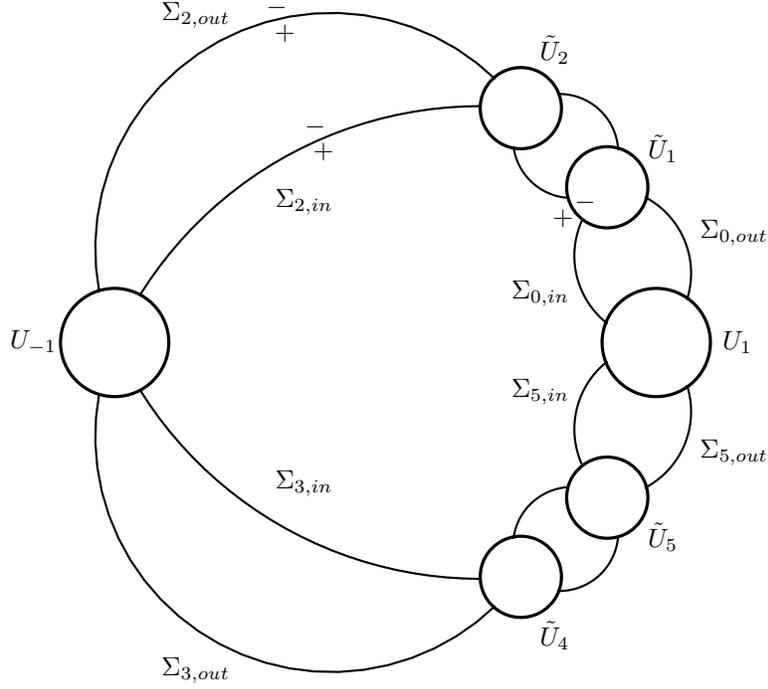
\begin{figure}[H] 
\centering
\begin{tikzpicture}[scale = 1.2]

\def\a{35} \def\b{60} \def\r{3}

\def\rsmallA{( (0.85*\r*sin(0.5*\a))^2 + (\r -0.85*\r*cos(0.5*\a))^2)^0.5}

\def\rsmallB{( (0.95*\r*sin(0.5*(\b-\a)))^2 + (\r - 0.95*\r*cos(0.5*(\b-\a)))^2)^0.5}

\def\rsmallC{( (0.4*\r*sin(0.5*(180-\b)))^2 + (\r - 0.4*\r*cos(0.5*(180-\b)))^2)^0.5}

\def\rbigC{( (0.7*\r*sin(0.5*(180-\b)))^2 + (\r + 0.7*\r*cos(0.5*(180-\b)))^2)^0.5}

\draw[name path=ellipse,black,very thick]
(\r,0) circle[x radius = 0.2*\r cm, y radius = 0.2*\r cm];

\fill ({1.3*\r},0) node[] {$U_1$};

\draw[name path=ellipse,black,very thick]
(-\r,0) circle[x radius = 0.2*\r cm, y radius = 0.2*\r cm];

\fill ({-1.3*\r},0) node[] {$U_{-1}$};

\draw [black,thick,domain=3:92] plot ({0.95*\r*cos(0.5*(\a+\b)) + \rsmallB*cos(\x)},{0.95*\r*sin(0.5*(\a+\b)) + \rsmallB*sin(\x)});

\draw [black,thick,domain=187:265] plot ({1.02*\r*cos(0.5*(\a+\b))+\rsmallB*cos(\x)},{1.02*\r*sin(0.5*(\a+\b)) + \rsmallB*sin(\x)});

\draw[name path=ellipse,black,very thick]
({\r*cos(\a)},{\r*sin(\a)}) circle[x radius = 0.15*\r cm, y radius = 0.15*\r cm];

\fill ({1.25*\r*cos(\a)},{1.25*\r*sin(\a)}) node[] {$\tilde{U}_1$};
\fill ({0.9*\r*cos(\a)},{0.9*\r*sin(\a)}) node[] {$-$}; 
\fill ({0.8*\r*cos(\a)},{0.8*\r*sin(\a)}) node[] {$+$};

\draw[name path=ellipse,black,very thick]
({\r*cos(\b)},{\r*sin(\b)}) circle[x radius = 0.15*\r cm, y radius = 0.15*\r cm];
\fill ({1.25*\r*cos(\b)},{1.25*\r*sin(\b)}) node[] {$\tilde{U}_2$};

\draw [black,thick,domain=3:92] plot ({0.95*\r*cos(0.5*(\a+\b)) + \rsmallB*cos(\x)},{-0.95*\r*sin(0.5*(\a+\b)) + \rsmallB*sin(-\x)});

\draw [black,thick,domain=187:265] plot ({1.02*\r*cos(0.5*(\a+\b))+\rsmallB*cos(\x)},{-1.02*\r*sin(0.5*(\a+\b)) + \rsmallB*sin(-\x)});

\draw[name path=ellipse,black,very thick]
({\r*cos(\a)},{-\r*sin(\a)}) circle[x radius = 0.15*\r cm, y radius = 0.15*\r cm];
\fill ({1.25*\r*cos(\a)},{-1.25*\r*sin(\a)}) node[] {$\tilde{U}_5$};

\draw[name path=ellipse,black,very thick]
({\r*cos(\b)},{-\r*sin(\b)}) circle[x radius = 0.15*\r cm, y radius = 0.15*\r cm];
\fill ({1.25*\r*cos(\b)},{-1.25*\r*sin(\b)}) node[] {$\tilde{U}_4$};

\draw [black,thick,domain=-18:62.5] plot ({0.85*\r*cos(0.5*\a) + \rsmallA*cos(\x)},{0.85*\r*sin(0.5*\a) + \rsmallA*sin(\x)});

\draw [black,thick,domain=156:232] plot ({1.065*\r*cos(0.5*\a) + \rsmallA*cos(\x)},{1.065*\r*sin(0.5*\a) + \rsmallA*sin(\x)});

\fill ({1.35*\r*cos(0.5*\a)},{1.35*\r*sin(0.5*\a)}) node[] {$\Sigma_{0,out}$};
\fill ({0.6*\r*cos(0.5*\a)},{0.6*\r*sin(0.5*\a)}) node[] {$\Sigma_{0,in}$};

\draw [black,thick,domain=-18:62.5] plot ({0.85*\r*cos(0.5*\a) + \rsmallA*cos(\x)},{0.85*\r*sin(-0.5*\a) + \rsmallA*sin(-\x)});

\draw [black,thick,domain=156:232] plot ({1.065*\r*cos(0.5*\a) + \rsmallA*cos(\x)},{1.065*\r*sin(-0.5*\a) + \rsmallA*sin(-\x)});

\fill ({1.35*\r*cos(0.5*\a)},{-1.35*\r*sin(0.5*\a)}) node[] {$\Sigma_{5,out}$};
\fill ({0.6*\r*cos(0.5*\a)},{-0.6*\r*sin(0.5*\a)}) node[] {$\Sigma_{5,in}$};

\draw [black,thick,domain=46:190.5] plot ({0.4*\r*cos(90+0.5*\b) + \rsmallC*cos(\x)},{0.4*\r*sin(90+0.5*\b) + \rsmallC*sin(\x)});

\draw [black,thick,domain=90:148] plot ({-0.7*\r*cos(90+0.5*\b) + \rbigC*cos(\x)},{-0.7*\r*sin(90+0.5*\b) + \rbigC*sin(\x)});

\fill ({1.4*\r*cos(0.5*\b+90)},{1.4*\r*sin(0.5*\b+90)}) node[] {$\Sigma_{2,out}$};
\fill ({0.6*\r*cos(0.5*\b+90)},{0.6*\r*sin(0.5*\b+90)}) node[] {$\Sigma_{2,in}$};

\fill ({0.74*\r*cos(0.6*\b+0.4*180)},{0.74*\r*sin(0.6*\b+0.4*180)}) node[] {$+$};
\fill ({0.84*\r*cos(0.6*\b+0.4*180)},{0.84*\r*sin(0.6*\b+0.4*180)}) node[] {$-$};

\fill ({1.2*\r*cos(0.6*\b+0.4*180)},{1.2*\r*sin(0.6*\b+0.4*180)}) node[] {$+$};
\fill ({1.3*\r*cos(0.6*\b+0.4*180)},{1.3*\r*sin(0.6*\b+0.4*180)}) node[] {$-$};

\draw [black,thick,domain=46:190.5] plot ({0.4*\r*cos(90+0.5*\b) + \rsmallC*cos(\x)},{0.4*\r*sin(-(90+0.5*\b)) + \rsmallC*sin(-\x)});

\draw [black,thick,domain=90:148] plot ({-0.7*\r*cos(90+0.5*\b) + \rbigC*cos(\x)},{-0.7*\r*sin(-(90+0.5*\b)) + \rbigC*sin(-\x)});

\fill ({1.4*\r*cos(0.5*\b+90)},{-1.4*\r*sin(0.5*\b+90)}) node[] {$\Sigma_{3,out}$};
\fill ({0.6*\r*cos(0.5*\b+90)},{-0.6*\r*sin(0.5*\b+90)}) node[] {$\Sigma_{3,in}$};
\end{tikzpicture}
\caption{The jump contour $\Sigma_{\tilde{R}}$ of $\tilde{R}$ in the case $\omega(n)/n < t < t_0$.} \label{figure:R2}
\end{figure}

We transfer Section 7.5.1 in \cite{ClaeysKrasovsky} to our case. Figure \ref{figure:R2} shows the contour chosen for the RHP of $\tilde{R}$, which we define as follows: 
\begin{align} \label{eqn:R2}
\tilde{R}(z) = \begin{cases} S(z) \tilde{P}_k(z)^{-1}, & z \in \tilde{\mathcal{U}}_k, \\
S(z)P_{\pm 1}(z)^{-1}, & z \in U_{\pm 1}, \\
S(z)N(z)^{-1}, & z \in \mathbb{C} \setminus \left( \overline{\tilde{\mathcal{U}}_1} \cup \overline{\tilde{\mathcal{U}}_2} \cup  \overline{\tilde{\mathcal{U}}_4} \cup 
\overline{\tilde{\mathcal{U}}_5} \cup \overline{U_1} \cup \overline{U_{-1}} \right). 
\end{cases}
\end{align}
Then $\tilde{R}$ is analytic, in particular has no jumps inside any of the local parametrices $\tilde{\mathcal{U}}_k$, $k  = 1,2,4,5$, $U_{\pm}$, or on the unit circle. On the rest of the lenses we can see that the jump matrix is $I + O(e^{-\delta nt})$ for some $\delta>0$, uniformly in $p \in (\epsilon, \pi - \epsilon)$ and $\omega(n)/n < t < t_0$. Because of the matching condition (d) of $P_{\pm 1}$ we have as in the case $0 < t < \omega(n)/n$ that
\begin{equation}
\tilde{R}_+(z) = \tilde{R}_-(z) P_{\pm 1}(z) N(z)^{-1} = \tilde{R}_-(z)(I+O(n^{-1})),
\end{equation}
uniformly for $z \in \partial U_{\pm 1}$, $p \in (\epsilon, \pi - \epsilon)$ and $0 < t < t_0$. Using (\ref{eqn:matching tilde}), we get that 
\begin{align} 
\tilde{R}_+(z) = \tilde{R}_-(z) \tilde{P}_k(z) N(z)^{-1} = \tilde{R}_-(z)(I+O((nt)^{-1}),
\end{align}
uniformly for $z \in \partial \tilde{\mathcal{U}}_k$, $p \in (\epsilon, \pi - \epsilon)$ and $\omega(n)/n < t < t_0$. Finally we have that $\lim_{z \rightarrow \infty} \tilde{R}(z) = I$, which by standard theory for RHPs with small jumps and RHPs on contracting contours implies that 
\begin{align} \label{eqn:R2 asymptotics}
\tilde{R}(z) = I + O((nt)^{-1}), \quad \frac{\text{d}\tilde{R}(z)}{\text{d}z} = O((nt)^{-1}),
\end{align}
uniformly for $z$ off the jump contour of $\tilde{R}$, and uniformly in $p \in (\epsilon, \pi - \epsilon)$ and $\omega(n)/n < t < t_0$.

\section{Asymptotics of $D_n(f_{p,t})$} \label{section:Toeplitz}
This section is a transfer of Section 8 in \cite{ClaeysKrasovsky} to our case.

\subsection{Asymptotics of the Differential Identity and Proof of Theorem \ref{thm:T, T+H extended}}

\begin{proposition}
Let $\alpha_1,\alpha_2,\alpha_1 + \alpha_2 > - \frac{1}{2}$, let $\sigma(s)$ be the solution to (\ref{eqn:painleve}) and let $\omega(x)$ be a positive, smooth function for $x$ sufficiently large, s.t.
\begin{equation}
\omega(x) \rightarrow \infty, \quad \omega(x) = o(x), \quad \text{as } x \rightarrow \infty.
\end{equation}
Then the following asymptotic expansion holds:
\begin{align} \label{eqn:diff_id_asymptotics}
\begin{split}
\frac{1}{i} \frac{\text{d}}{\text{d}t} \ln D_n(f_{p,t}) =& 2n(\beta_1 - \beta_2) + d_1(p,t;\alpha_0,\alpha_1,\beta_1,\alpha_2,\beta_2,\alpha_3) \\
&+ d_2(p,t;\alpha_0,\alpha_1,\beta_1,\alpha_2,\beta_2,\alpha_3) + d_3(p,t;\alpha_0,\alpha_1,\beta_1,\alpha_2,\beta_2,\alpha_3) + \epsilon_{n,p,t},
\end{split}
\end{align}
where for the error term $\epsilon_{n,p,t}$
\begin{equation}
\left| \int_0^t \epsilon_{n,p,\tau} \text{d}\tau \right|= O(\omega(n)^{-\delta}) = o(1),
\end{equation}
for some $\delta > 0$, uniformly in $p \in (\epsilon, \pi - \epsilon)$ and $0 < t < t_0$, and where 
\begin{align}
\begin{split}
d_1(p,t;\alpha_0,\alpha_1,\beta_1,\alpha_2,\beta_2,\alpha_3) =& 2\alpha_1 z_1 V'(z_1) - 2\alpha_2 z_2 V'(z_2) + 2(\beta_1 + \beta_2) \sum_{j = 1}^\infty jV_j \left( \cos j(p+t) - \cos j(p-t) \right) \\
& -i(2\alpha_1^2 + \beta_1^2 + \beta_1 \beta_2) \frac{\cos p-t}{\sin p-t} + i(2\alpha_2^2 + \beta_2^2 + \beta_1\beta_2) \frac{ \cos p+t }{ \sin p+t } \\
& + i(\beta_1^2-\beta_2^2) \frac{\cos p}{\sin p}  \\
& - 2i\alpha_1 \alpha_0 \frac{\cos \frac{p-t}{2}}{\sin \frac{p-t}{2}} + 2i\alpha_1 \alpha_3 \frac{\sin \frac{p-t}{2}}{\cos \frac{p-t}{2}} + 2i\alpha_2 \alpha_0 \frac{\cos \frac{p+t}{2}}{\sin \frac{p+t}{2}} - 2i\alpha_2 \alpha_3 \frac{\sin \frac{p+t}{2}}{\cos \frac{p+t}{2}} \\
d_2(p,t;\alpha_0,\alpha_1,\beta_1,\alpha_2,\beta_2,\alpha_3) =& \frac{2}{it} \sigma(s) + i (4\alpha_1 \alpha_2 - (\beta_1+\beta_2)^2) \frac{\cos t}{\sin t} \\
d_3(p,t;\alpha_0,\alpha_1,\beta_1,\alpha_2,\beta_2,\alpha_3) =& 2\sigma_s \Bigg( - 2\sum_{j = 1}^\infty jV_j \left( \cos j(p-t) + \cos j(p+t) \right) - 2\sum_{j=0}^5 \alpha_j \\
&- i\beta_1 \frac{\cos p-t}{\sin p-t} - i\beta_2 \frac{\cos p+t}{\sin p+t} + i(\beta_1 - \beta_2) \left( \frac{\cos t}{\sin t} - \frac{1}{t} \right) - i(\beta_1 + \beta_2) \frac{\cos p}{\sin p} \Bigg).
\end{split}
\end{align}
\end{proposition} 

\noindent \textbf{Proof:} The proof is analogous to the proof of Proposition 8.1 in \cite{ClaeysKrasovsky}. As is done there, we assume below that $\alpha_k>0$, $k = 1,2,4,5$, for simplicity of notation. Once (\ref{eqn:diff_id_asymptotics}) is proven under this assumption, the case where $\alpha_k = 0$ for some $k$ then follows from the uniformity of the error terms in $\alpha_k$, $k = 1,2,4,5$. Extending to the case where $\alpha_k < 0$ for some $k$ is straightforward. We prove the proposition first in the regime $0 < t \leq \omega(n)/n$ and then in the regime $\omega(n)/n < t < t_0$.  \\

Using the transformation $Y \rightarrow T \rightarrow S$ inside the unit circle, outside the lenses, we can rewrite the differential identity (\ref{eqn:diff id}) in the form 
\begin{equation} \label{eqn:transform of diff id}
\frac{1}{i} \frac{\text{d}}{\text{d}t} \ln D_n(f_{p,t}) = \sum_{k = 1,2,4,5} q_k \left( n\beta_k + 2\alpha_k z_k \left( S^{-1} \frac{\text{d}S}{\text{d}z} \right)_{+,22} (z_k) \right), 
\end{equation}
with $q_k = 1$ for $k = 1,4$ and $q_k = -1$ for $k = 2,5$, and where the limit $z \rightarrow z_k$ is taken from the inside of the unit circle and outside the lenses.

\subsubsection{$0 < t \leq \omega(n)/n$} 
By (\ref{eqn:R}) we get 
\begin{equation}
S(z) = R(z) P_\pm(z), \quad z \in U_\pm,
\end{equation}
and thus 
\begin{align} \label{eqn:SS}
\begin{split}
\left( S^{-1} \frac{\text{d}S}{\text{d}z} \right)_{22} (z) =& \left( P_\pm^{-1} \frac{\text{d}P_\pm}{\text{d}z} \right)_{22} (z) + A_{n,p,t}(z), \quad z \in U_{e^{\pm ip}}, \\
A_{n,p,t}(z) =& \left( P_\pm^{-1}(z) R^{-1}(z) \frac{dR}{dz}(z) P_\pm(z) \right)_{22}.
\end{split}
\end{align}

Following exactly the same approach as on pages 60, 61 in \cite{ClaeysKrasovsky}, we can use (\ref{eqn:Psi asymptotics}) and the small and large $|s|$ asymptotics from Sections 5, 6 of \cite{ClaeysKrasovsky} to obtain that for $k = 1,2,4,5$ 
\begin{equation}
\int_0^t |A_{n,p,t}(z_k)| \text{d}t = o(\omega(n)^{-1}), \quad n \rightarrow \infty,
\end{equation}
uniformly in $0 < t \leq \omega(n)/n$ and $p \in (\epsilon, \pi - \epsilon)$, and thus also 
\begin{equation} \label{eqn:epsilon tilde}
\tilde{\epsilon}_{n,p,t} := 2 \sum_{k = 1,2,4,5} q_k \alpha_k z_k A_{n,p,t}(z_k) = O(\omega(n)^{-1}),
\end{equation}
uniformly in $0 < t \leq \omega(n)/n$ and $p \in (\epsilon, \pi - \epsilon)$.\\ 

By (\ref{eqn:P1}) we have
\begin{align} \label{eqn:PP1}
\begin{split}
\left( P_\pm^{-1}\frac{\text{d}P_\pm}{\text{d}z} \right)_{22}(z) =& -\frac{n}{2z} + \frac{1}{2}\frac{f'_{p,t}}{f_{p,t}}(z) + \left( \Phi_\pm^{-1} \frac{\text{d}\Phi_\pm}{\text{d}z} \right)_{22}(z) + \left( \Phi_\pm^{-1} E_\pm^{-1} \frac{\text{d}E_\pm}{\text{d}z} \Phi_\pm \right)_{22}(z),\\
\end{split}
\end{align}
with $z$ inside the unit circle and outside of the lenses of $\Sigma_S$. By (\ref{eqn:E}) we have for $z$ near $e^{ip}$:
\begin{align} \label{eqn:EE}
\begin{split} 
E_\pm(z)^{-1} \frac{\text{d} E_\pm(z)}{\text{d} z} =& h_\pm(z) \sigma_3, 
\end{split}
\end{align}
where 
\begin{align} \label{eqn:h_pm}
\begin{split}
h_\pm(z) =& \pm \frac{\beta_1}{z\ln \frac{z}{e^{\pm ip}} \pm itz} \pm \frac{\beta_2}{z\ln \frac{z}{e^{\pm ip}} \mp itz}  \\
&- \frac{1}{2} \sum_{j = 1}^\infty j V_j z^{j-1} + \frac{1}{2} \sum_{j=-1}^{-\infty} j V_j z^{j-1} - \sum_{j=0}^5 \frac{\beta_j}{z-z_j} - \frac{1}{2z} \sum_{j=0}^5 \alpha_j.
\end{split}
\end{align}

In the following equation we need the fact that
\begin{align} \label{eqn:1/log}
\begin{split}
\frac{1}{\ln (1+z)} =& \frac{1}{z} + \frac{1}{2} + O(z), \quad z \rightarrow 0, \\ 
\frac{1}{\ln z - \ln z_k} =& \frac{1}{\ln \frac{z}{z_k}} = \frac{z_k}{z - z_k} + \frac{1}{2} + O(|z-z_k|), \quad z \rightarrow z_k.
\end{split}
\end{align}
Let $k = 1,2$ and denote $k' = 1$ for $k = 2$, $k' = 2$ for $k = 1$. Putting together (\ref{eqn:PP1}) and (\ref{eqn:EE}), we obtain for $k = 1,2$ (as in (8.33) in \cite{ClaeysKrasovsky})
\begin{align} \label{eqn:PPlimit}
\begin{split}
\left( P_+^{-1} \frac{\text{d}P_+}{\text{d}z} \right)_{22,+}(z_k) =& - \frac{n}{2z_k} + \lim_{z \rightarrow z_k} \left( \frac{1}{2} \frac{f'_{p,t}}{f_{p,t}}(z) - \frac{\alpha_k}{z \ln z - z \ln z_1} \right) \\
&+ \frac{1}{tz_k} \left( F_{+,k'}^{-1} \frac{\text{d}F_{+,k'}}{\text{d}z} \right)_{22} \left( \frac{1}{t} \ln \frac{z_k}{e^{ip}} \right) + h_+(z_k) (F_{+,k'}^{-1}\sigma_3 F_{+,k'})_{22} \left( \frac{1}{t} \ln \frac{z_k}{e^{ip}} \right) \\
=& - \frac{n}{2z_k} + \frac{1}{2} V'(z_k) + \sum_{j = 0}^5 \frac{\beta_j}{2z_k} + \sum_{j,j \neq k} \frac{\alpha_j}{z_k - z_j} - \frac{\alpha_j}{2z_k} \\ 
& + \frac{1}{tz_k} \left( F_{+,k'}^{-1} \frac{\text{d}F_{+,k'}}{\text{d}z} \right)_{22} \left( \frac{1}{t} \ln \frac{z_k}{e^{ip}} \right) + h_+(z_k) (F_{+,k'}^{-1}\sigma_3 F_{+,k'})_{22} \left( \frac{1}{t} \ln \frac{z_k}{e^{ip}} \right), \\
\left( P_-^{-1} \frac{\text{d}P_-}{\text{d}z} \right)_{22,+}(\overline{z_k}) =& - \frac{n}{2\overline{z_k}} + \lim_{z \rightarrow \overline{z_k}} \left( \frac{1}{2} \frac{f'_{p,t}}{f_{p,t}}(z) - \frac{\alpha_k}{z \ln z - z \ln \overline{z_k}} \right) \\
&+ \frac{1}{t\overline{z_k}} \left( F_{-,k}^{-1} \frac{\text{d}F_{-,k}}{\text{d}z} \right)_{22} \left( \frac{1}{t} \ln \frac{\overline{z_k}}{e^{-ip}} \right) + h_-(\overline{z_2}) (F_{-,k}^{-1}\sigma_3 F_{-,k})_{22} \left( \frac{1}{t} \ln \frac{\overline{z_k}}{e^{-ip}} \right) \\
=& - \frac{n}{2\overline{z_k}} + \frac{1}{2} V'(\overline{z_k}) + \sum_{j = 0}^5 \frac{\beta_j}{2\overline{z_k}} + \sum_{j, z_j \neq \overline{z_k}} \frac{\alpha_j}{\overline{z_k} - z_j} - \frac{\alpha_j}{2\overline{z_k}} \\ 
& + \frac{1}{t\overline{z_k}} \left( F_{-,k}^{-1} \frac{\text{d}F_{-,k}}{\text{d}z} \right)_{22} \left( \frac{1}{t} \ln \frac{\overline{z_k}}{e^{-ip}} \right) + h_-(\overline{z_k}) (F_{-,k}^{-1}\sigma_3 F_{-,k})_{22} \left( \frac{1}{t} \ln \frac{\overline{z_k}}{e^{-ip}} \right),
\end{split}
\end{align}
where in the second equalities we used (\ref{eqn:1/log}), and where $F_{\pm,k}$ equal the functions $F_k$ defined in (\ref{eqn:F_1 neq 0}), (\ref{eqn:F_2 neq 0}), (\ref{eqn:F_1 0}) and (\ref{eqn:F_2 0}), with $(\alpha_1, \alpha_2, \beta_1, \beta_2)$ in the appendix replaced by $(\alpha_2,\alpha_1,-\beta_2,-\beta_1)$ in the $-$ case. When replacing $(\alpha_1,\alpha_2,\beta_1,\beta_2)$ in the Painlev\'{e} equation (3.52) in \cite{ClaeysKrasovsky} with $(\alpha_2,\alpha_1,\beta_2,\beta_1)$ or $(\alpha_1,\alpha_2,-\beta_1,-\beta_2)$, we get the same Painlev\'{e} equation as in our Theorem \ref{thm:painleve}. Thus we can see that Propositions 3.1 and 3.2 in \cite{ClaeysKrasovsky} become in our case:
\begin{proposition} \label{prop:F}
We have the identities
\begin{align}
\begin{split}
\alpha_2(F_{+,1}(i;s)^{-1}\sigma_3 F_{+,1}(i;s))_{22} =& -\sigma_s(s) + \frac{\beta_1+\beta_2}{2}, \\
\alpha_1(F_{+,2}(-i;s)^{-1}\sigma_3 F_{+,2}(-i;s))_{22} =& \sigma_s(s) + \frac{\beta_1+\beta_2}{2}, \\
\alpha_1(F_{-,1}(i;s)^{-1}\sigma_3 F_{-,1}(i;s))_{22} =& -\sigma_s(s) - \frac{\beta_1+\beta_2}{2}, \\ 
\alpha_2(F_{-,2}(-i;s)^{-1}\sigma_3 F_{-,2}(-i;s))_{22} =& \sigma_s(s) - \frac{\beta_1+\beta_2}{2}, 
\end{split}
\end{align}
and
\begin{align}
\begin{split}
\alpha_2 \left( F_{+,1}(i;s)^{-1} F_{+,1,\zeta}(i;s) \right)_{22} =& \frac{i}{4} \sigma(s) - \frac{i}{8}(\beta_1 + \beta_2)s + \frac{i}{8}(\beta_1+\beta_2)^2, \\
\alpha_1 \left(F_{+,2}(-i;s)^{-1} F_{+,2,\zeta}(-i;s) \right)_{22} =& -\frac{i}{4} \sigma(s) - \frac{i}{8}(\beta_1 + \beta_2)s - \frac{i}{8}(\beta_1+\beta_2)^2, \\
\alpha_1 \left(F_{-,1}(i;s)^{-1} F_{-,1,\zeta}(i;s) \right)_{22} =& \frac{i}{4} \sigma(s) + \frac{i}{8}(\beta_1 + \beta_2)s + \frac{i}{8}(\beta_1+\beta_2)^2, \\
\alpha_2 \left(F_{-,2}(-i;s)^{-1} F_{-,2,\zeta}(-i;s) \right)_{22} =& -\frac{i}{4} \sigma(s) + \frac{i}{8}(\beta_1 + \beta_2)s - \frac{i}{8}(\beta_1+\beta_2)^2.
\end{split}
\end{align}
\end{proposition}

Putting together (\ref{eqn:transform of diff id}), (\ref{eqn:SS}) and (\ref{eqn:PPlimit}), we obtain:
\begin{align} \label{eqn:diff_id_F}
\begin{split}
&\frac{1}{i} \frac{\text{d}}{\text{d}t} \ln D_n(f_{p,t}) \\
=& \sum_{k = 1,2,4,5} q_k \left( n\beta_k + \alpha_k z_k V'(z_k) + \alpha_k \sum_{j=0}^5 \beta_j + 2 \alpha_k \sum_{j \neq k} \alpha_j \frac{z_k}{z_k - z_j} - \alpha_k \sum_{j\neq k} \alpha_j \right) \\
&+ \sum_{k = 1}^2 (-1)^{k+1} \left( 2 \alpha_k z_k h_+(z_k) \left( F^{-1}_{+,k'} \sigma_3 F_{+,k'} \right)_{22} ((-1)^ki) + \frac{2}{t} \alpha_k \left( F_{+,k'}^{-1} F_{+,k'}' \right)_{22} ((-1)^ki) \right) \\
&+ \sum_{k = 1}^2 (-1)^k \left( 2 \alpha_k \overline{z_k} h_-(\overline{z_k}) \left( F^{-1}_{-,k} \sigma_3 F_{-,k} \right)_{22} ((-1)^{k+1}i) + \frac{2}{t} \alpha_k \left( F_{-,k}^{-1} F_{-,k}' \right)_{22} ((-1)^{k+1}i) \right) \\
&+ \tilde{\epsilon}_{n,p,t},\\
=& 2n(\beta_1 - \beta_2) + 2\alpha_1 z_1 V'(z_1) - 2\alpha_2 z_2 V'(z_2) \\
&+ 2 \sum_{k = 1,2,4,5} q_k \alpha_k \sum_{j \neq k} \alpha_j \frac{z_k}{z_k - z_j} \\
&+ \sum_{k = 1}^2 (-1)^{k+1} \left( 2 \alpha_k z_k h_+(z_k) \left( F^{-1}_{+,k'} \sigma_3 F_{+,k'} \right)_{22} ((-1)^ki) + \frac{2}{t} \alpha_k \left( F_{+,k'}^{-1} F_{+,k'}' \right)_{22} ((-1)^ki) \right) \\
&+ \sum_{k = 1}^2 (-1)^k \left( 2 \alpha_k \overline{z_k} h_-(\overline{z_k}) \left( F^{-1}_{-,k} \sigma_3 F_{-,k} \right)_{22} ((-1)^{k+1}i) + \frac{2}{t} \alpha_k \left( F_{-,k}^{-1} F_{-,k}' \right)_{22} ((-1)^{k+1}i) \right) \\
&+ \tilde{\epsilon}_{n,p,t},
\end{split}
\end{align}
where we used that since $\beta_1 = - \beta_5$, $\beta_2 = - \beta_4$, $\alpha_1 = \alpha_5$, $\alpha_2 = \alpha_4$ and $V(z) = V(\overline{z})$, it holds that: 
\begin{align}
\begin{split}
\sum_{k = 1,2,4,5} q_k \alpha_k \sum_{j \neq k} \alpha_j =& 0, \quad \sum_{j = 0}^5 \beta_j = 0, \quad zV'(z) = - \overline{z}V'(\overline{z}).
\end{split}
\end{align}
By Proposition \ref{prop:F} we get 
\begin{align} \label{eqn:FF}
\begin{split}
&\sum_{k = 1}^2 (-1)^{k+1} \frac{2}{t} \alpha_k \left( F_{+,k'}^{-1} F_{+,k'}' \right)_{22} ((-1)^ki) + \sum_{k = 1}^2 (-1)^k \frac{2}{t} \alpha_k \left( F_{-,k}^{-1} F_{-,k}' \right)_{22} ((-1)^{k+1}i) \\
=& \frac{2}{it} \sigma(s) - \frac{i}{t}(\beta_1 + \beta_2)^2,
\end{split}
\end{align} 
and 
\begin{align} \label{eqn:FsigmaF}
\begin{split}
&\sum_{k = 1}^2 (-1)^{k+1} 2 \alpha_k z_k h_+(z_k) \left( F^{-1}_{+,k'} \sigma_3 F_{+,k'} \right)_{22} ((-1)^ki)\\
&+ \sum_{k = 1}^2 (-1)^k 2 \alpha_k \overline{z_k} h_-(\overline{z_k}) \left( F^{-1}_{-,k} \sigma_3 F_{-,k} \right)_{22} ((-1)^{k+1}i) \\
=& \left( 2\sigma_s(s) + \beta_1 + \beta_2 \right) \left( z_1 h_+(z_1) + \overline{z_1}h_-(\overline{z_1}) \right) + \left( 2\sigma_s(s) - \beta_1 - \beta_2 \right) \left( z_2 h_+(z_2) + \overline{z_2}h_-(\overline{z_2}) \right)  \\
=& 2\sigma_s \left( z_1 h_+(z_1) + \overline{z_1}h_-(\overline{z_1}) + z_2 h_+(z_2) + \overline{z_2}h_-(\overline{z_2}) \right) \\
&+ (\beta_1 + \beta_2) \left( z_1 h_+(z_1) + \overline{z_1}h_-(\overline{z_1}) - z_2 h_+(z_2) - \overline{z_2}h_-(\overline{z_2}) \right).
\end{split}
\end{align}
Then (\ref{eqn:diff_id_F}), (\ref{eqn:FF}) and (\ref{eqn:FsigmaF}) result in 
\begin{align} \label{eqn:diff_id_sum_hh}
\begin{split}
\frac{1}{i} \frac{\text{d}}{\text{d}t} \ln D_n(f_{p,t}) =& 2n(\beta_1 - \beta_2) + 2\alpha_1 z_1 V'(z_1) - 2\alpha_2 z_2 V'(z_2) \\
&+ 2 \sum_{k = 1,2,4,5} q_k \alpha_k \sum_{j \neq k} \alpha_j \frac{z_k}{z_k - z_j} \\
&+ \frac{2}{it} \sigma(s) - \frac{i}{t}(\beta_1 + \beta_2)^2 \\
&+ 2\sigma_s \left( z_1 h_+(z_1) + \overline{z_1}h_-(\overline{z_1}) + z_2 h_+(z_2) + \overline{z_2}h_-(\overline{z_2}) \right) \\
&+ (\beta_1 + \beta_2) \left( z_1 h_+(z_1) + \overline{z_1}h_-(\overline{z_1}) - z_2 h_+(z_2) - \overline{z_2}h_-(\overline{z_2}) \right) \\
&+ \tilde{\epsilon}_{n,p,t}.
\end{split}
\end{align}

\noindent With $h_+(z)$ and $h_-(z)$ given in (\ref{eqn:h_pm}), and using (\ref{eqn:1/log}), we obtain: 
\begin{align} \label{eqn:hz}
\begin{split}
h_+(z_1) z_1 =& - \frac{1}{2} \sum_{j = 1}^\infty j V_j (z_1^{j} + \overline{z_1}^j) - \sum_{j\neq 1} \frac{\beta_jz_1}{z_1-z_j} - \frac{1}{2} \sum_{j=0}^5 \alpha_j + \frac{\beta_1}{2} - \frac{\beta_{2}}{2it}, \\
h_+(z_2) z_2 =& - \frac{1}{2} \sum_{j = 1}^\infty j V_j (z_2^{j} + \overline{z_2}^j) - \sum_{j\neq 2} \frac{\beta_jz_2}{z_2-z_j} - \frac{1}{2} \sum_{j=0}^5 \alpha_j + \frac{\beta_2}{2} + \frac{\beta_{1}}{2it}, \\
h_-(\overline{z_1}) \overline{z_1} =& - \frac{1}{2} \sum_{j = 1}^\infty j V_j (z_1^{j} + \overline{z_1}^j) - \sum_{j\neq 5} \frac{\beta_j z_5}{z_5-z_j} - \frac{1}{2} \sum_{j=0}^5 \alpha_j - \frac{\beta_1}{2} - \frac{\beta_{2}}{2it}, \\
h_-(\overline{z_2}) \overline{z_2} =& - \frac{1}{2} \sum_{j = 1}^\infty j V_j (z_2^{j} + \overline{z_2}^j) - \sum_{j\neq 4} \frac{\beta_j z_4}{z_4-z_j} - \frac{1}{2} \sum_{j=0}^5 \alpha_j - \frac{\beta_2}{2} + \frac{\beta_{1}}{2it}. 
\end{split}
\end{align}
We thus have 
\begin{align} \label{eqn:hh}
\begin{split}
h_+(z_1)z_1 + \overline{z_1} h_-(\overline{z_1}) =& - 2\sum_{j = 1}^\infty jV_j \cos j(p-t) - \sum_{j=0}^5 \alpha_j - \frac{\beta_2}{it} + \beta_1 \left( \frac{z_1 + \overline{z_1}}{z_1 - \overline{z_1}} \right) \\
&+ \beta_2 \left( - \frac{z_1}{z_1 - z_2} + \frac{z_1}{z_1 - \overline{z_2}} - \frac{\overline{z_1}}{\overline{z_1} - z_2} + \frac{\overline{z_1}}{\overline{z_1} - \overline{z_2}} \right) \\
=& - 2\sum_{j = 1}^\infty jV_j \cos j(p-t) - \sum_{j=0}^5 \alpha_j - \frac{\beta_2}{it} - i\beta_1 \frac{\cos p-t}{\sin p-t} - i\beta_2\frac{\cos t}{\sin t} \\
&- i\beta_2 \frac{\cos p}{\sin p}, \\
h_+(z_2)z_2 + \overline{z_2}h_-(\overline{z_2}) =& - 2\sum_{j = 1}^\infty jV_j \cos j(p+t) - \sum_{j=0}^5 \alpha_j  + \frac{\beta_1}{it} + \beta_2 \left( \frac{z_2 + \overline{z_2}}{z_2 - \overline{z_2}} \right) \\
&+ \beta_1 \left( - \frac{z_2}{z_2 - z_1} + \frac{z_2}{z_2 - \overline{z_1}} - \frac{\overline{z_2}}{\overline{z_2} - z_1} + \frac{\overline{z_2}}{\overline{z_2} - \overline{z_1}} \right) \\
=& - 2\sum_{j = 1}^\infty jV_j \cos (p+t) - \sum_{j=0}^5 \alpha_j  + \frac{\beta_1}{it} - i\beta_2 \frac{\cos p+t}{\sin p+t} + i\beta_1 \frac{\cos t}{\sin t} \\
& - i\beta_1 \frac{\cos p}{\sin p}. 
\end{split}
\end{align}
Further we calculate
\begin{align} \label{eqn:sum_sin}
\begin{split}
\sum_{k = 1,5} q_k \alpha_k \sum_{j \neq k} \alpha_j \frac{z_k}{z_k - z_j} =& \alpha_1 \alpha_0 \left( \frac{z_1}{z_1 - 1} - \frac{\overline{z_1}}{\overline{z_1} - 1} \right) + \alpha_1^2 \left( \frac{z_1}{z_1 - \overline{z_1}} - \frac{\overline{z_1}}{\overline{z_1} - z_1} \right) \\
&+ \alpha_1 \alpha_2 \left( \frac{z_1}{z_1 - z_2} + \frac{z_1}{z_1 - \overline{z_2}} - \frac{\overline{z_1}}{\overline{z_1} - z_2} - \frac{\overline{z_1}}{\overline{z_1} - \overline{z_2}} \right) \\
&+ \alpha_1 \alpha_3 \left( \frac{z_1}{z_1 + 1} - \frac{\overline{z_1}}{\overline{z_1} + 1} \right), \\
- \sum_{k = 2,4} q_k \alpha_k \sum_{j \neq k} \alpha_j \frac{z_k}{z_k - z_j} =& \alpha_2 \alpha_0 \left( \frac{z_2}{z_2 - 1} - \frac{\overline{z_2}}{\overline{z_2} - 1} \right) \\
&+ \alpha_2 \alpha_1 \left( \frac{z_2}{z_2 - z_1} + \frac{z_2}{z_2 - \overline{z_1}} - \frac{\overline{z_2}}{\overline{z_2} - z_1} - \frac{\overline{z_2}}{\overline{z_2} - \overline{z_1}} \right) \\
&+ \alpha_2^2 \left( \frac{z_2}{z_2 - \overline{z_2}} - \frac{\overline{z_2}}{\overline{z_2} - z_2} \right) + \alpha_2 \alpha_3 \left( \frac{z_2}{z_2 + 1} - \frac{\overline{z_2}}{\overline{z_2} + 1} \right), \\
\sum_{k = 1,2,4,5} q_k \alpha_k \sum_{j \neq k} \alpha_j \frac{z_k}{z_k - z_j} =& -i\alpha_1^2 \frac{\cos p-t}{\sin p-t} + i\alpha_2^2 \frac{ \cos p+t }{ \sin p+t } + 2i \alpha_1 \alpha_2 \frac{\cos t}{\sin t} \\
&+ \alpha_1 \alpha_0 \frac{z_1 + 1}{z_1 - 1} + \alpha_1 \alpha_3 \frac{z_1 - 1}{z_1 + 1} - \alpha_2 \alpha_0 \frac{z_2 + 1}{z_2 - 1} - \alpha_2 \alpha_3 \frac{z_2 - 1}{z_2 + 1} \\
=& -i\alpha_1^2 \frac{\cos p-t}{\sin p-t} + i\alpha_2^2 \frac{ \cos p+t }{ \sin p+t } + 2i \alpha_1 \alpha_2 \frac{\cos t}{\sin t} \\
&- i\alpha_1 \alpha_0 \frac{\cos \frac{p-t}{2}}{\sin \frac{p-t}{2}} + i\alpha_1 \alpha_3 \frac{\sin \frac{p-t}{2}}{\cos \frac{p-t}{2}} + i\alpha_2 \alpha_0 \frac{\cos \frac{p+t}{2}}{\sin \frac{p+t}{2}} - i\alpha_2 \alpha_3 \frac{\sin \frac{p+t}{2}}{\cos \frac{p+t}{2}}. 
\end{split}
\end{align}
Putting together (\ref{eqn:diff_id_sum_hh}), (\ref{eqn:hh}) and (\ref{eqn:sum_sin}) we get that uniformly in $p \in (\epsilon, \pi - \epsilon)$ and $0 < t \leq \omega(n)/n$:
\begin{align} \label{eqn:diff_id_final}
\begin{split}
\frac{1}{i} \frac{\text{d}}{\text{d}t} \ln D_n(f_{p,t}) =& 2n(\beta_1 - \beta_2) + 2\alpha_1 z_1 V'(z_1) - 2\alpha_2 z_2 V'(z_2) \\
& -2i\alpha_1^2 \frac{\cos p-t}{\sin p-t} + 2i\alpha_2^2 \frac{ \cos p+t }{ \sin p+t } + 4i \alpha_1 \alpha_2 \frac{\cos t}{\sin t} \\
&- 2i\alpha_1 \alpha_0 \frac{\cos \frac{p-t}{2}}{\sin \frac{p-t}{2}} + 2i\alpha_1 \alpha_3 \frac{\sin \frac{p-t}{2}}{\cos \frac{p-t}{2}} + 2i\alpha_2 \alpha_0 \frac{\cos \frac{p+t}{2}}{\sin \frac{p+t}{2}} - 2i\alpha_2 \alpha_3 \frac{\sin \frac{p+t}{2}}{\cos \frac{p+t}{2}} \\
&+ \frac{2}{it} \sigma(s) - \frac{i}{t}(\beta_1 + \beta_2)^2 \\
&+ 2\sigma_s(s) \Bigg( - 2\sum_{j = 1}^\infty jV_j \left( \cos j(p-t) + \cos j(p+t) \right) - 2\sum_{j=0}^5 \alpha_j + \frac{\beta_1 - \beta_2}{it} \\
&- i\beta_1 \frac{\cos p-t}{\sin p-t} - i\beta_2 \frac{\cos p+t}{\sin p+t} + i(\beta_1 - \beta_2) \frac{\cos t}{\sin t} - i(\beta_1 + \beta_2) \frac{\cos p}{\sin p} \Bigg) \\
&+ (\beta_1 + \beta_2) \Bigg( 2\sum_{j = 1}^\infty jV_j \left( \cos j(p+t) - \cos j(p-t) \right) - \frac{\beta_1 + \beta_2}{it} \\
&- i\beta_1 \frac{\cos p-t}{\sin p-t} + i\beta_2 \frac{\cos p+t}{\sin p+t} - i(\beta_1 + \beta_2) \frac{\cos t}{\sin t} + i(\beta_1 - \beta_2) \frac{\cos p}{\sin p} \Bigg) \\
&+ \tilde{\epsilon}_{n,p,t}. 
\end{split}
\end{align}
Simplifying further and setting $\epsilon_{n,p,t} = \tilde{\epsilon}_{n,p,t}$ for $0 < t \leq \omega(n)/n$ we obtain (\ref{eqn:diff_id_asymptotics}) for $0 < t \leq \omega(n)/n$.

\subsubsection{$\omega(n)/n < t < t_0$}
For $\omega(n)/n < t < t_0$, $z \in \tilde{\mathcal{U}}_k$, we obtain instead of (\ref{eqn:SS}):
\begin{align} \label{eqn:SS2}
\begin{split}
\left( S^{-1}(z) \frac{\text{d}S(z)}{\text{d}z} \right)_{22} =& \left( \tilde{P}_k(z)^{-1} \frac{\text{d}\tilde{P}_k(z)}{\text{d}z} \right)_{22} + A_{n,p,t}, \\
A_{n,p,t}(z) =& \left( \tilde{P}_k(z)^{-1} \tilde{R}^{-1}(z) \frac{\text{d}\tilde{R}(z)}{\text{d}z} \tilde{P}_k(z) \right)_{22},
\end{split}
\end{align}
with $\tilde{R}(z)$ given in (\ref{eqn:R2}). From (\ref{eqn:P tilde}) and (\ref{eqn:W}) it follows that for $|z| < 1$
\begin{align} \label{eqn:PP2}
\begin{split}
\left( \tilde{P}_k(z)^{-1} \frac{\text{d}\tilde{P}_k(z)}{\text{d}z} \right)_{22} =& -\frac{n}{2z} + \frac{1}{2} \frac{f_{p,t}'(z)}{f_{p,t}(z)} + \left( M_k(z)^{-1} \frac{\text{d}M_k(z)}{\text{d}z} \right)_{22} \\
&+ \tilde{h}_1(z) \left( M_k(z)^{-1} \sigma_3 M_k(z) \right)_{22}, \\
A_{n,p,t}(z_k) =& \lim_{z \rightarrow z_k} A_{n,p,t}(z) \\
=& \lim_{z \rightarrow z_k} \left( M_k^{-1} \tilde{E}_k^{-1} \tilde{R}^{-1} \frac{\text{d}\tilde{R}}{\text{d}z} \tilde{E}_k M_k \right)_{22}(z),
\end{split}
\end{align}
where 
\begin{align}
\begin{split}
M_1(z) =& M^{(\alpha_1,\beta_1)}(nt(\frac{1}{t} \ln \frac{z}{e^{ip}} + i)), \\
M_2(z) =& M^{(\alpha_2,\beta_2)}(nt(\frac{1}{t} \ln \frac{z}{e^{ip}} - i)), \\
M_4(z) =& M^{(\alpha_4,\beta_4)}(nt(\frac{1}{t} \ln \frac{z}{e^{-ip}} + i)), \\
M_5(z) =& M^{(\alpha_5,\beta_5)}(nt(\frac{1}{t} \ln \frac{z}{e^{-ip}} - i)), 
\end{split}
\end{align}
and where
\begin{align} \label{eqn:hhtilde}
\begin{split}
\tilde{h}_1(z)\sigma_3 =& \tilde{E}_1(z)^{-1} \frac{\text{d} \tilde{E}_1(z)}{\text{d}z}, \quad \tilde{h}_1(z) = h_+(z) - \frac{\beta_2}{z \ln \frac{z}{e^{ip}} - itz}, \\ 
\tilde{h}_2(z)\sigma_3 =& \tilde{E}_2(z)^{-1} \frac{\text{d} \tilde{E}_2(z)}{\text{d}z}, \quad \tilde{h}_2(z) = h_+(z) - \frac{\beta_1}{z \ln \frac{z}{e^{ip}} + itz}, \\ 
\tilde{h}_4(z)\sigma_3 =& \tilde{E}_4(z)^{-1} \frac{\text{d} \tilde{E}_4(z)}{\text{d}z}, \quad \tilde{h}_4(z) = h_-(z) + \frac{\beta_1}{z \ln \frac{z}{e^{-ip}} - itz}, \\ 
\tilde{h}_5(z)\sigma_3 =& \tilde{E}_5(z)^{-1} \frac{\text{d} \tilde{E}_5(z)}{\text{d}z}, \quad \tilde{h}_5(z) = h_-(z) + \frac{\beta_2}{z \ln \frac{z}{e^{ip}} + itz}, 
\end{split}
\end{align}
with $h_+(z),h_-(z)$ given in (\ref{eqn:h_pm}). By (\ref{eqn:R2 asymptotics}) and the fact that $\tilde{E}_k$ is uniformly bounded for $p \in (\epsilon, \pi - \epsilon)$, $\omega(n)/n < t < t_0$, and $z \in \tilde{U}_k$, we see that 
\begin{align}
A_{n,p,t}(z_k) = O((nt)^{-1}) = O(\omega(n)^{-1})
\end{align}
uniformly in $p \in (\epsilon, \pi - \epsilon)$, $\omega(n)/n < t < t_0$, and thus also 
\begin{equation} \label{eqn:epsilon tilde 2}
\tilde{\epsilon}_{n,p,t} := 2 \sum_{k = 1,2,4,5} q_k \alpha_k z_k A_{n,p,t}(z_k) = O(\omega(n)^{-1}),
\end{equation}
uniformly in $0 < t < \omega(n)/n$ and $p \in (\epsilon, \pi - \epsilon)$. This implies that as $n \rightarrow \infty$
\begin{align} \label{eqn:error2}
\int_{\omega(n)/n}^{t_0} |\tilde{\epsilon}_{n,p,t}|\text{d}t = o(\omega(n)^{-1}), 
\end{align}
uniformly in $p \in (\epsilon, \pi - \epsilon)$. 

From (8.41) and (8.42) in \cite{ClaeysKrasovsky} we can see that for $z \rightarrow z_k$ inside the unit circle and outside of the lenses of $\Sigma_S$ we have 
\begin{align}
\begin{split}
\left( M_1^{-1} \frac{\text{d}M_1}{\text{d}z} \right) (z) =& - \left( \frac{\beta_1}{2\alpha_1} + \frac{\alpha_1}{n \ln \frac{z}{e^{ip}} + int} \right) \frac{n}{z} + o(1),\\
\left( M_2^{-1} \frac{\text{d}M_2}{\text{d}z} \right) (z) =& - \left( \frac{\beta_2}{2\alpha_2} + \frac{\alpha_2}{n \ln \frac{z}{e^{ip}} - int} \right) \frac{n}{z} + o(1),\\
\left( M_4^{-1} \frac{\text{d}M_4}{\text{d}z} \right) (z) =& - \left( \frac{-\beta_2}{2\alpha_2} + \frac{\alpha_2}{n \ln \frac{z}{e^{-ip}} + int} \right) \frac{n}{z} + o(1),\\
\left( M_5^{-1} \frac{\text{d}M_5}{\text{d}z} \right) (z) =& - \left( \frac{-\beta_1}{2\alpha_1} + \frac{\alpha_1}{n \ln \frac{z}{e^{-ip}} - int} \right) \frac{n}{z} + o(1),
\end{split}
\end{align}
and in the same limit,
\begin{align}
\begin{split}
\left( M_k \sigma_3 M \right)_{22}(z_k) =& \frac{\beta_k}{\alpha_k}.
\end{split}
\end{align}
Together with (\ref{eqn:hz}) and (\ref{eqn:hhtilde}) we obtain (again in the same limit)
\begin{align}
\begin{split}
&z_k \left( M^{-1}_k \frac{\text{d}M_k}{\text{d}z} \right)_{22}(z) + z_k\tilde{h}_k(z_k) \left( M_k^{-1} \sigma_3 M_k \right)_{22} (z_k) \\
=& \left( \frac{\beta_k}{2\alpha_k} + \frac{\alpha_k}{n \ln \frac{z}{z_k}} \right)n + \frac{\beta_k}{\alpha_k} \left(  -\frac{1}{2} \sum_{j = 1}^\infty j V_j (z_k^{j} + \overline{z_k}^j) - \sum_{j,j\neq k} \frac{\beta_jz_k}{z_k-z_j} - \frac{1}{2} \sum_{j=0}^5 \alpha_j + \frac{\beta_k}{2} \right) +o(1). \\
\end{split}
\end{align}
Combining this with (\ref{eqn:1/log}) and (\ref{eqn:PP2}) we get
\begin{align}
\begin{split}
&2\alpha_kz_k \left( \tilde{P}_k^{-1} \frac{\text{d}\tilde{P}_k}{\text{d}z} \right)_{+,22}(z_k) \\
=& - n(\alpha_k + \beta_k) + \alpha_k z_k V'(z_k) + 2 \alpha_k \sum_{j,j \neq k} \frac{\alpha_jz_k}{z_k - z_j} - \alpha_k \sum_{j,j\neq k} \alpha_j \\
& +2\beta_k \left(  -\frac{1}{2} \sum_{j = 1}^\infty j V_j (z_k^{j} + \overline{z_k}^j) - \sum_{j,j\neq k} \frac{\beta_jz_k}{z_k-z_j} - \frac{1}{2} \sum_{j=0}^5 \alpha_j + \frac{\beta_k}{2} \right).
\end{split}
\end{align}
Together with (\ref{eqn:transform of diff id}), (\ref{eqn:SS2}) and (\ref{eqn:epsilon tilde 2}) we obtain
\begin{align} \label{eqn:diff id final2}
\begin{split}
\frac{1}{i} \frac{\text{d}}{\text{d}t} \ln D_n(f_{p,t}) =& S(p,t;\alpha_0,\alpha_1,\beta_1,\alpha_2, \beta_2,\alpha_3) + \tilde{\epsilon}_{n,p,t},
\end{split}
\end{align}
where 
\begin{align}
\begin{split}
&S(p,t;\alpha_0,\alpha_1,\beta_1,\alpha_2, \beta_2,\alpha_3) \\
=& 2\sum_{k=1}^2 (-1)^{k+1} \left( (\alpha_k- \beta_k) \left(\sum_{j = 1}^\infty jV_j z_k^j\right) - (\alpha_k + \beta_k) \left( \sum_{j = 1}^\infty jV_{j} \overline{z_k}^j \right) \right) \\  
&- 2 (\beta_1 - \beta_2) \sum_{k = 0}^5 \alpha_k \\
&-2i(\alpha_1^2 + \beta_1^2) \frac{\cos (p-t)}{\sin (p-t)} + 2i(\alpha_2^2 + \beta_2^2) \frac{\cos(p+t)}{\sin(p+t)} + 4i(\alpha_1\alpha_2 - \beta_1 \beta_2) \frac{\cos t}{\sin t} \\
&- 2i\alpha_1 \alpha_0 \frac{\cos \frac{p-t}{2}}{\sin \frac{p-t}{2}} + 2i\alpha_1 \alpha_3 \frac{\sin \frac{p-t}{2}}{\cos \frac{p-t}{2}} + 2i\alpha_2 \alpha_0 \frac{\cos \frac{p+t}{2}}{\sin \frac{p+t}{2}} - 2i\alpha_2 \alpha_3 \frac{\sin \frac{p+t}{2}}{\cos \frac{p+t}{2}}.
\end{split}
\end{align}

Now we compare this expression to (\ref{eqn:diff_id_final}), obtained for $0 < t \leq \omega(n)/n$. Consider (\ref{eqn:diff_id_sum_hh}) for large $ s = -2int$ and without the error term. Substituting there the asymptotics of $\sigma(s)$ from Theorem \ref{thm:painleve} and using (\ref{eqn:sum_sin}) we see that
\begin{align}
\begin{split}
\frac{1}{i} \frac{\text{d}}{\text{d}t} \ln D_n(f_{p,t}) =& 2\alpha_1 z_1 V'(z_1) - 2\alpha_2 z_2 V'(z_2) \\
& -2i\alpha_1^2 \frac{\cos p-t}{\sin p-t} + 2i\alpha_2^2 \frac{ \cos p+t }{ \sin p+t } + 4i \alpha_1 \alpha_2 \frac{\cos t}{\sin t} \\
&- 2i\alpha_1 \alpha_0 \frac{\cos \frac{p-t}{2}}{\sin \frac{p-t}{2}} + 2i\alpha_1 \alpha_3 \frac{\sin \frac{p-t}{2}}{\cos \frac{p-t}{2}} + 2i\alpha_2 \alpha_0 \frac{\cos \frac{p+t}{2}}{\sin \frac{p+t}{2}} - 2i\alpha_2 \alpha_3 \frac{\sin \frac{p+t}{2}}{\cos \frac{p+t}{2}}  \\
& + \frac{1}{it}4 \beta_1 \beta_2 + 2\beta_1\left( h_+(z_1)z_1 + h_-(\overline{z_1}) \overline{z_1} \right) - 2\beta_2 \left( h_+(z_2)z_2 + h_-(\overline{z_2}) \overline{z_2} \right) \\
&+ \Theta_{n,p,t},
\end{split}
\end{align}
where $\Theta_{n,p,t}$ arises from the error term in the asymptotics of $\sigma(s)$, and becomes of order $\omega(n)^{-\delta}$ after integration w.r.t. $t$, i.e.
\begin{equation}
\left| \int_{\omega(n)/n}^t \Theta_{n,p,\tau} \text{d}\tau \right| = O(\omega(n)^{-\delta}),
\end{equation}
uniformly in $p \in (\epsilon, \pi - \epsilon)$ and $\omega(n)/n < t < t_0$. Using (\ref{eqn:hh}) now we see that  
\begin{align}
\begin{split}
n(\beta_1 - \beta_2) + d_1 + d_2 + d_3 = S + \Theta_{n,p,t}, \quad \omega(n)/n < t < t_0.
\end{split}
\end{align}
Thus when setting 
\begin{align}
\begin{split}
\epsilon_{n,p,t} =& \tilde{\epsilon}_{n,p,t} + \Theta_{n,p,t}, \quad \omega(n)/n < t < t_0,
\end{split}
\end{align}
we see that (\ref{eqn:diff_id_asymptotics}) remains valid also in the region $\omega(n)/n < t < t_0$, where the smallness of the error terms follows from (\ref{eqn:epsilon tilde 2}). \qed

\begin{remark} Integrating (\ref{eqn:diff id final2}) from $t$ to $t_0$ with $\omega(n)/n < t < t_0$ and using Theorem \ref{thm:T non-uniform} for the expansion of $D_n(f_{p,t_0})$, we get the same expansion for $D_{n}(f_{p,t})$ that Theorem \ref{thm:T non-uniform} gives. The error term is then $O(\omega(n)^{-1})$ and uniform for $p \in (\epsilon, \pi - \epsilon)$ and $\omega(n)/n < t < t_0$. Thus we have proven the statement on Toeplitz determinants in Theorem \ref{thm:T, T+H extended}.   
\end{remark}

\subsection{Integration of the Differential Identity}

We now integrate (\ref{eqn:diff_id_asymptotics}), where we use exactly the same approach as in Section 8.2 of \cite{ClaeysKrasovsky}. We obtain
\begin{align}
\begin{split}
&\int_0^t d_1(p,\tau;\alpha_0,\alpha_1,\beta_1,\alpha_2,\beta_2,\alpha_3) \text{d}\tau \\
=& +2i\alpha_1 \left( V(e^{i(p-t)}) - V(e^{ip}) \right) + 2i\alpha_2 \left( V(e^{i(p+t)}) - V(e^{ip}) \right) \\
&+ 2(\beta_1+\beta_2) \sum_{j = 1}^\infty V_j \left( \sin j(p+t) + \sin j(p-t) - 2 \sin jp\right) \\
& + i\left( 2\alpha_1^2 + \beta_1^2 + \beta_1\beta_2 \right) \ln \frac{\sin p-t}{\sin p} +i \left( 2\alpha_2^2 + \beta_2^2 + \beta_1\beta_2 \right) \ln \frac{\sin p+t}{\sin p} + it(\beta_1^2 - \beta_2^2) \frac{\cos p}{\sin p}  \\
& + 4i\alpha_1 \alpha_0 \ln \frac{\sin \frac{p-t}{2}}{\sin \frac{p}{2}} - 4i\alpha_1 \alpha_3 \ln \frac{\cos \frac{p-t}{2}}{\cos \frac{p}{2}} + 4i\alpha_2 \alpha_0 \ln \frac{ \sin \frac{p+t}{2} }{ \sin \frac{p}{2}} - 4i\alpha_2 \alpha_3 \ln \frac{ \cos \frac{p+t}{2} }{ \cos \frac{p}{2} },
\end{split}
\end{align}
and 
\begin{align}
\begin{split}
&\int_0^t d_2(p,\tau;\alpha_0,\alpha_1,\beta_1,\alpha_2,\beta_2,\alpha_3) \text{d}\tau \\
=&- 2i\int_0^{-2int} \frac{1}{s} \left( \sigma(s) - 2\alpha_1\alpha_2 + \frac{1}{2}(\beta_1+\beta_2)^2 \right) \text{d}s + i\left( (\beta_1+\beta_2)^2 - 4\alpha_1\alpha_2 \right) \ln \frac{\sin t}{t},
\end{split}
\end{align}
and 
\begin{align}
\begin{split}
&\int_0^t d_3(p,\tau;\alpha_0,\alpha_1,\beta_1,\alpha_2,\beta_2,\alpha_3) \text{d}\tau \\
=& (\beta_1 - \beta_2) \Bigg[ 2\sum_{j = 1}^\infty V_j \left( \sin j(p-t) - \sin j(p+t) \right) - 2t \sum_{j=0}^5 \alpha_j \\
& + i\beta_1 \ln \frac{\sin p-t}{\sin p} - i\beta_2 \ln \frac{\sin p+t}{\sin p} + i(\beta_1 - \beta_2) \ln  \frac{\sin t}{t} - it(\beta_1 + \beta_2) \frac{\cos p}{\sin p} \Bigg]. 
\end{split}
\end{align}
Putting things together we see that 
\begin{align}
\begin{split} 
\ln D_n(f_{p,t}) =& \ln D_n(f_{p,0}) + 2int (\beta_1 - \beta_2) \\
&- 2\alpha_1 \left( V(e^{i(p-t)}) - V(e^{ip}) \right) - 2\alpha_2 \left( V(e^{i(p+t)}) - V(e^{ip}) \right) \\
&+ 4i \sum_{j = 1}^\infty V_j \left( \beta_ 1\sin j(p-t) + \beta_2 \sin j(p+t) - (\beta_1+\beta_2) \sin jp\right) \\
&- 2\left( \alpha_1^2 + \beta_1^2 \right) \ln \frac{\sin p-t}{\sin p} - 2\left( \alpha_2^2 + \beta_2^2 \right) \ln \frac{\sin p+t}{\sin p} \\
& - 4\alpha_1 \alpha_0 \ln \frac{\sin \frac{p-t}{2}}{\sin \frac{p}{2}} - 4\alpha_1 \alpha_3 \ln \frac{\cos \frac{p-t}{2}}{\cos \frac{p}{2}} - 4\alpha_2 \alpha_0 \ln \frac{ \sin \frac{p+t}{2} }{ \sin \frac{p}{2}} - 4i\alpha_2 \alpha_3 \ln \frac{ \cos \frac{p+t}{2} }{ \cos \frac{p}{2} }\\
&- 2i\int_0^{-2int} \frac{1}{s} \left( \sigma(s) - 2\alpha_1\alpha_2 + \frac{1}{2}(\beta_1+\beta_2)^2 \right) \text{d}s + 4\left( \beta_1 \beta_2 - \alpha_1\alpha_2 \right) \ln \frac{\sin t}{t}\\ 
&- 2it(\beta_1 - \beta_2) \sum_{j = 0}^5 \alpha_j \\
&+ o(1),
\end{split}
\end{align}
uniformly in $p \in (\epsilon, \pi - \epsilon)$ and $0 < t < t_0$. To calculate the asymptotics of $D_n(f_{p,0})$ we use Theorem \ref{thm:T non-uniform} and get
\begin{align}
\begin{split}
\ln D_n(f_{p,0}) =& n V_0 + \sum_{k = 1}^\infty kV_k^2 + \ln n \left( \alpha_0^2 + \alpha_3^2 + 2(\alpha_1 + \alpha_2)^2 - 2(\beta_1 + \beta_2)^2 \right) \\
&- \alpha_0 \ln \left( b_+(1)b_-(1) \right) - \alpha_3 \ln \left( b_+(-1)b_-(-1) \right) \\
&- 2(\alpha_1 + \alpha_2 - \beta_1 - \beta_2) \ln b_+(e^{ip}) - 2(\alpha_1 + \alpha_2 + \beta_1 + \beta_2) \ln b_+(e^{-ip}) \\
&- 2\alpha_0\alpha_3 \ln 2 - 4\alpha_0(\alpha_1+\alpha_2) \ln 2 \sin \frac{p}{2} - 4\alpha_3(\alpha_1+\alpha_2) \ln 2 \cos \frac{p}{2} \\
&- 2((\beta_1+\beta_2)^2 + (\alpha_1+\alpha_2)^2) \ln 2\sin p  \\
&+ 2i\alpha_0(\beta_1+\beta_2)(p - \pi) + 2i\alpha_3(\beta_1+\beta_2) p  \\
&+ 2i(\alpha_1+\alpha_2)(\beta_1+\beta_2)(2p - \pi) \\
&+ \ln \frac{G(1+\alpha_0)^2}{G(1+2\alpha_0)} + \ln \frac{G(1+\alpha_3)^2}{G(1+2\alpha_3)} \\
&+ \ln \frac{G(1+\alpha_1+\alpha_2+\beta_1+\beta_2)^2 G(1+\alpha_1+\alpha_2-\beta_1-\beta_2)^2}{G(1+2\alpha_1+2\alpha_2)^2}\\
&+ o(1),
\end{split}
\end{align}
uniformly in $p \in (\epsilon, \pi - \epsilon)$. Combining the 2 last equations we obtain
\begin{align} \label{eqn:ln_D_sin}
\begin{split}
\ln D_n(f_{p,t}) =& 2int (\beta_1 - \beta_2) + n V_0 + \sum_{k = 1}^\infty kV_k^2 + \ln n \sum_{j = 0}^5 (\alpha_j^2 - \beta_j^2)  \\
&- \sum_{j=0}^5 (\alpha_j - \beta_j) \left(\sum_{k = 1}^\infty V_k z_j^k\right) + (\alpha_j + \beta_j) \left( \sum_{k = 1}^\infty V_k \overline{z_j}^k \right) \\
&- 2\left( \alpha_1^2 + \beta_1^2 \right) \ln 2\sin (p-t) - 2\left( \alpha_2^2 + \beta_2^2 \right) \ln 2\sin (p+t) \\
&- 4(\alpha_1\alpha_2 + \beta_1\beta_2) \ln 2\sin p - 2\alpha_0\alpha_3 \ln 2\\
& - 4\alpha_1 \alpha_0 \ln 2\sin \frac{p-t}{2} - 4\alpha_1 \alpha_3 \ln 2\cos \frac{p-t}{2} - 4\alpha_2 \alpha_0 \ln 2\sin \frac{p+t}{2} - 4\alpha_2 \alpha_3 \ln 2\cos \frac{p+t}{2} \\
&+ 2\int_0^{-2int} \frac{1}{s} \left( \sigma(s) - 2\alpha_1\alpha_2 + \frac{1}{2}(\beta_1+\beta_2)^2 \right) \text{d}s + 4\left( \beta_1 \beta_2 - \alpha_1\alpha_2 \right) \ln \frac{\sin t}{nt}\\ 
&+ 2i\alpha_0\beta_1(p-t-\pi) + 2i\alpha_0\beta_2(p+t-\pi) + 2i\alpha_3\beta_1 (p-t) + 2i\alpha_3\beta_2(p+t) \\
&+ 2i\alpha_1\beta_1(2p-2t-\pi) + 2i\alpha_2\beta_1(2p-2t-\pi) + 2i\alpha_1\beta_2 (2p+2t-\pi) + 2i\alpha_2\beta_2(2p+2t-\pi) \\
&+ \ln \frac{G(1+\alpha_0)^2}{G(1+2\alpha_0)} + \ln \frac{G(1+\alpha_3)^2}{G(1+2\alpha_3)} \\
&+ \ln \frac{G(1+\alpha_1+\alpha_2+\beta_1+\beta_2)^2 G(1+\alpha_1+\alpha_2-\beta_1-\beta_2)^2}{G(1+2\alpha_1+2\alpha_2)^2}\\
&+ o(1),
\end{split}
\end{align}
uniformly in $p \in (\epsilon, \pi - \epsilon)$ and $0 < t < t_0$. We note that
\begin{align}
\begin{split}
&\sum_{0 \leq j < k \leq 5, (j,k) \neq (1,2),(4,5)} (\alpha_j\beta_k - \alpha_k \beta_j) \ln \frac{z_k}{z_j e^{i\pi}}  \\
=& 2i\alpha_0\beta_1 (p-t-\pi) + 2i\alpha_0\beta_2(p+t-\pi) + 2i\alpha_3\beta_1(p-t) + 2i\alpha_3\beta_2(p+t) \\
&+ 2i(\alpha_2\beta_1 + \alpha_1\beta_2)(2p - \pi) + 2i\alpha_1\beta_1(2p - 2t -\pi) + 2i\alpha_2\beta_2(2p + 2t -\pi).
\end{split}
\end{align}
Thus we finally get that
\begin{align} \label{eqn:ln_D_uniform}
\begin{split}
\ln D_n(f_{p,t}) =& 2int(\beta_1 - \beta_2) + n V_0 + \sum_{k = 1}^\infty kV_k^2 + \ln n \sum_{j = 0}^5 (\alpha_j^2 - \beta_j^2) \\
&- \sum_{j=0}^5 (\alpha_j - \beta_j) \left(\sum_{k = 1}^\infty V_k z_j^k\right) + (\alpha_j + \beta_j) \left( \sum_{k = 1}^\infty V_k \overline{z_j}^k \right) \\
&+ \sum_{0 \leq j < k \leq 5, (j,k) \neq (1,2),(4,5)} 2(\beta_j\beta_k - \alpha_j\alpha_k) \ln |z_j - z_k| + (\alpha_j\beta_k - \alpha_k \beta_j) \ln \frac{z_k}{z_j e^{i\pi}} \\
&+4it(\alpha_1\beta_2 - \alpha_2\beta_1) \\
&+ 2\int_0^{-2int} \frac{1}{s} \left( \sigma(s) - 2\alpha_1\alpha_2 + \frac{1}{2}(\beta_1+\beta_2)^2 \right) \text{d}s + 4\left( \beta_1 \beta_2 - \alpha_1\alpha_2 \right) \ln \frac{\sin t}{nt}\\ 
&+ \ln \frac{G(1+\alpha_0)^2}{G(1+2\alpha_0)} + \ln \frac{G(1+\alpha_3)^2}{G(1+2\alpha_3)} \\
&+ \ln \frac{G(1+\alpha_1+\alpha_2+\beta_1+\beta_2)^2 G(1+\alpha_1+\alpha_2-\beta_1-\beta_2)^2}{G(1+2\alpha_1+2\alpha_2)^2}\\
&+ o(1),
\end{split}
\end{align}
uniformly in $p \in (\epsilon, \pi - \epsilon)$ and $0 < t < t_0$, which proves Theorem \ref{thm:T uniform}.

\section{Asymptotics of $\Phi_n(0)$, $\Phi_n(\pm 1)$ and $D_n^{T+H,\kappa}(f_{p,t})$} \label{section:T+H}

Tracing back the transformations of Section \ref{section:asymptotics of polynomials} we get  
\begin{align}
\begin{split}
\Phi_n(0) &= Y^{(n)}_{11}(0) = T(0)_{11} = S(0)_{11} = \left( R(0)N(0) \right)_{11} = - R_{12}(0) D_{\text{in},p,t}(0)^{-1} \\
&= - R_{12}(0) e^{-V_0} = o(1),
\end{split}
\end{align}
uniformly in $p \in (\epsilon, \pi - \epsilon)$ and $0 < t < t_0$. \\

The asymptotics of $\Phi_n(1)$ and $\Phi_n(-1)$ can be calculated exactly as in Chapter 7 of \cite{DeiftItsKrasovsky}. As is apparent from (7.13) there, these asymptotics are uniform for all other singularities bounded away from $\pm 1$. Thus when using Lemma \ref{lem:connection between T and T+H} to calculate from the asymptotics of $D_n(f_{p,t})$ given in Theorem \ref{thm:T, T+H extended}, $\Phi_n(\pm 1)$ and $\Phi_n(0)$, the asymptotics of the corresponding Toeplitz+Hankel determinants, then uniformity of the error terms in $p,t$ is preserved. This proves the statement on Toeplitz+Hankel determinants in Theorem \ref{thm:T, T+H extended}. 

Further, since we know the relation (\ref{eqn:relation}) between the asymptotic expansion of $D_n(f_{p,t})$ which is uniform for $0 < t < t_0$, and the asymptotic expansion for $t > t_0$, Lemma \ref{lem:connection between T and T+H} immediately gives the relationship between the expansions of $D_n^{T+H,\kappa}(f_{p,t})$ in the two regimes $0 < t < t_0$ and $t > t_0$. In view of the way the uniform asymptotics of $D_n^{T+H,\kappa}(f_{p,t})$ are derived from the uniform asymptotics of $D_{2n}(f_{p,t})^{1/2}$, $\Phi_n(\pm 1)^{1/2}$ and $\Phi_n(0)$, the relationship between the $0 < t < t_0$ asymptotics of $D_n^{T+H,\kappa}(f_{p,t})$ and the $t > t_0$ asymptotics one gets from Theorem \ref{thm:T+H uniform} is given by (\ref{eqn:relation}), with both sides divided by $2$, and $n$ replaced by $2n$. This proves Theorem \ref{thm:T+H uniform}.

\section*{Acknowledgements}
Our work was supported by ERC Advanced Grant 740900 (LogCorRM).  We are most grateful to Theo Assiotis for his kind permission to state Theorem \ref{thm:Gaussian field} here and to set out its proof in Appendix \ref{appendix:Gaussian fields}, as well as for many extremely helpful discussions.  We are also most grateful to Tom Claeys, Gabriel Glesner, Alexander Minakov and Meng Yang for having kindly shared with us their work in progress and for having communicated to us one of their results from \cite{Claeys et al}, prior to posting it on the arXiv, which we quote in Theorem \ref{thm:T+H Claeys}. We make use of this result to prove our Theorem \ref{thm:main2}. Further we thank Mo Dick Wong for helpful comments and suggestions. Finally we thank two anonymous referees for their careful reading, helpful remarks and suggestions.

\begin{appendix}
\section{Convergence to the Gaussian Fields} \label{appendix:Gaussian fields}
In this appendix we set out the proof of Theorem \ref{thm:Gaussian field}. As stated in the introduction, this theorem was established by Assiotis \& Keating and we are most grateful to Dr Assiotis for permitting us to use it and to give its proof here. \\

We first need the following definition:
\begin{definition}[The Sobolev spaces $H^{-\epsilon}_0$] For $s \in \mathbb{R}$, consider the space of formal Fourier series 
\begin{equation}
H^{s} = \left\{ f \sim \sum_{k \in \mathbb{Z}} f_k e^{ik\theta} : \sum_{k \in \mathbb{Z}} (1+k^2)^s |f_k|^2 < \infty \right\}
\end{equation}
with inner product
\begin{equation}
\langle f, g \rangle_s = \sum_{k \in \mathbb{Z}} (1+k^2)^s f_kg_k^*.
\end{equation}
The closed subspace $\left\{ f \in H^s: f_0 = 0 \right\}$ will be denoted by $H^s_0$.
\end{definition}
\begin{remark}
$\left(H^s, \langle \cdot , \cdot \rangle_s \right)$ is a Hilbert space for all $s \in \mathbb{R}$. For $s \geq 0$ $H^s$ is a subspace of $H^0$, i.e. the space of square-integrable functions on the unit circle. For $s < 0$, $H^s$ can be interpreted as the dual space of $H^{-s}$, and as a space of generalized functions. 
\end{remark}
\noindent \textbf{Proof of Theorem \ref{thm:Gaussian field}:} The proof strategy is as follows: we treat $\left( \Re \ln p_n, \Im \ln p_n \right)_{n \in \mathbb{N}}$ as a sequence in $H^{-\epsilon}_0 \times H^{-\epsilon}_0$ and show that if any of its subsequences has a limit then that limit has to be $( X \pm x, \hat{X} \pm \hat{x} )$, with $+$ when the underlying matrices are symplectic and $-$ when they are orthogonal. We do this by showing that the finite-dimensional distributions of $\left( \Re \ln p_n, \Im \ln p_n \right)_{n \in \mathbb{N}}$, i.e. the distributions of finite sets of pairs of Fourier coefficients, converge to those of $(X \pm x, \hat{X} \pm \hat{x})$. We then show that the set $\left\{ \Re \ln p_n, \Im \ln p_n \right\}_{n \in \mathbb{N}}$ is tight in $H^{-\epsilon}_0 \times H^{-\epsilon}_0$. Since $H^{-\epsilon}_0 \times H^{-\epsilon}_0$ is complete and seperable, Prokhorov's theorem implies that the closure of $\left\{ \Re \ln p_n, \Im \ln p_n \right\}_{n \in \mathbb{N}}$ is sequentially compact w.r.t. the topology of weak convergence. In particular this means that every subsequence of $\left( \Re \ln p_n, \Im \ln p_n \right)_{n \in \mathbb{N}}$ has a weak limit in $H^{-\epsilon}_0 \times H^{-\epsilon}_0$. Since any such limit has to be $( X \pm x, \hat{X} \pm \hat{x} )$ it follows that the whole sequence $\left( \Re \ln p_n, \Im \ln p_n \right)_{n \in \mathbb{N}}$ must converge weakly to $( X \pm x, \hat{X} \pm \hat{x} )$.\\

We recall that 
\begin{equation}
\ln (1 - z) = - \sum_{k = 1}^\infty \frac{z^k}{k}
\end{equation}
for $|z| \leq 1$, where for $z = 1$ both sides equal $-\infty$. By using the identity $\ln \det = \text{Tr} \ln$ we see that the Fourier expansions of $\ln p_n$, $\Re \ln p_n$, $\Im \ln p_n$ in $H_0^{-\epsilon}$ are given as follows: 
\begin{align}
\begin{split}
\ln p_n(\theta) =& - \sum_{k = 1}^{\infty} \frac{\text{Tr} (U_n^k)}{k} e^{-ik\theta},\\
\Re \ln p_n(\theta) =& - \frac{1}{2} \sum_{k = 1}^\infty \frac{\text{Tr} (U_n^k)}{k} (e^{ik\theta} + e^{-ik\theta}),\\
\Im \ln p_n(\theta) =& \frac{1}{2i} \sum_{k = 1}^\infty \frac{\text{Tr} (U_n^k)}{k} (e^{ik\theta} - e^{-ik\theta}).
\end{split}
\end{align}
The convergence of the finite-dimensional distributions now follows immediately from Theorem \ref{thm:traces}: we have for any $l \in \mathbb{N}$, as $n \rightarrow \infty$:
\begin{align}
\begin{split}
&\Big( (\Re \ln p_n)_{-l}, (\Im \ln p_n)_{-l}, (\Re \ln p_n)_{-(l-1)}, (\Im \ln p_n)_{-(l-1)}, ..., (\Re \ln p_n)_{l}, (\Im \ln p_n)_{l} \Big) \\
\xrightarrow{d}  &\Big( (X \pm x)_{-l}, (\hat{X} \pm \hat{x})_{-l}, (X \pm x)_{-(l-1)}, (\hat{X} + \hat{x})_{-(l-1)}, ..., (X \pm x)_{l}, (\hat{X} + \hat{x})_{l} \Big).   
\end{split}
\end{align}  
We proceed to show tightness of $\left( \Re \ln p_n, \Im \ln p_n \right)_{n \in \mathbb{N}}$ in $H^{-\epsilon}_0 \times H^{-\epsilon}_0$, i.e. for every $\delta > 0$ we construct a compact $K_\delta \subset H^{-\epsilon}_0 \times H^{-\epsilon}_0$ for which
\begin{equation} \label{eqn:tightness}
\sup_{n \in \mathbb{N}} \mathbb{P} \left( \left( \Re \ln p_n, \Im \ln p_n \right) \notin K_\delta \right) < \delta.
\end{equation}
By the Rellich-Kondrachov theorem we have that for any $s \in (-\epsilon, 0 )$ the closed ball $\overline{B(0,R)}_{H^s}$ of radius $R>0$ in $H^{s}$ is compact in $H^{-\epsilon}$, which, since $\overline{B(0,R)}_{H^s_0} = \overline{B(0,R)}_{H^s} \cap H_0^{-\epsilon}$, also implies that $\overline{B(0,R)}_{H^s_0}$ is compact in $H_0^{-\epsilon}$. Thus when fixing $s$, choosing $K_\delta = \overline{B(0,R(\delta))}_{H^s_0} \times \overline{B(0,R(\delta))}_{H^s_0}$, and using Chebyshev's inequality, we get 
\begin{align}
\begin{split}
\sup_{n \in \mathbb{N}} \mathbb{P} \left( \left( \Re \ln p_n, \Im \ln p_n \right) \notin K_\delta \right) =& \sup_{n \in \mathbb{N}} \mathbb{P} \left( \max \left\{ ||\Re \ln p_n||_{H^s_0}, ||\Im \ln p_n ||_{H^s_0} \right\} > R(\delta) \right) \\
\leq & \frac{1}{R(\delta)^2} \sup_{n \in \mathbb{N}} \mathbb{E} \left( \max \left\{ ||\Re \ln p_n||^2_{H^s_0}, ||\Im \ln p_n ||^2_{H^s_0} \right\} \right).
\end{split}
\end{align}  
We have 
\begin{align}
\begin{split}
\sup_{n \in \mathbb{N}} \mathbb{E} \left( || \Re \ln p_n||_{H^s_0}^2 \right) =& \sup_{n \in \mathbb{N}} \mathbb{E} \left( || \Im \ln p_n||_{H^s_0}^2 \right)\\
=& \sup_{n \in \mathbb{N}} \mathbb{E} \left( \frac{1}{2} \sum_{k=1}^\infty (1+k^2)^s \frac{ \left| \text{Tr} (U_n^k) \right|^2}{k^2} \right) \\
=& \sup_{n \in \mathbb{N}} \frac{1}{2} \sum_{k = 1}^\infty (1+k^2)^s \frac{\mathbb{E} \left( \left| \text{Tr} (U_n^k) \right|^2 \right)}{k^2} \\
\leq & \text{const} \sup_{n \in \mathbb{N}} \frac{1}{2} \sum_{k = 1}^\infty (1+k^2)^s \frac{\min \{k,n\}}{k^2} \\
\leq & \text{const} \frac{1}{2} \sum_{k = 1}^\infty (1 + k^2)^s \frac{1}{k} < \infty,
\end{split}
\end{align}
where in the first inequality we used (\ref{eqn:trace bound}). Thus choosing $R(\delta)$ big enough we get (\ref{eqn:tightness}). This finishes the proof. \qed

\section{Construction of the Gaussian Multiplicative Chaos Measure} \label{appendix:GMC}
In this section we "exponentiate" the field $Y_{\alpha,\beta}$ in (\ref{eqn:Y}) to obtain the non-trivial Gaussian multiplicative chaos measure $\mu_{\alpha,\beta}$. We follow the approach of Kahane in \cite{Kahane}, with the only difference being that in our case the Gaussian field's covariance function (\ref{eqn:Cov Y})  has singularities not just on the diagonal, but also on the antidiagonal. \\

For $k \in \mathbb{N}$ define the truncated Gaussian fields 
\begin{align}
\begin{split}
Y_{\alpha,\beta}^{(k)} (\theta) &= \sum_{j=1}^k \frac{\mathcal{N}_j}{\sqrt{j}} \left( 2\alpha \cos(j\theta) - 2i\beta \sin(j\theta) \right), 
\end{split}
\end{align}
whose covariance functions are given by 
\begin{align}\label{eqn:Cov finite}
\begin{split}
&\text{Cov}(Y_{\alpha,\beta}^{(k)} (\theta), Y_{\alpha,\beta}^{(k)}(\theta') ) \\
=& \sum_{j=1}^k \frac{1}{j} \left( 2(\alpha^2 - \beta^2) \cos(j(\theta-\theta')) + 2(\alpha^2 + \beta^2) \cos(j(\theta+\theta')) - 4i \alpha \beta \sin(j(\theta+\theta')) \right).
\end{split}
\end{align}
Define the measures $\mu_{\alpha,\beta}^{(k)}$ on $S^1 \sim [0,2\pi)$ by 
\begin{align}
\begin{split}
\mu_{\alpha,\beta}^{(k)}(d\theta) &= \frac{e^{Y_{\alpha,\beta}^{(k)}(\theta)}}{\mathbb{E}\left(e^{Y_{\alpha,\beta}^{(k)}(\theta)}\right)} \text{d}\theta = e^{Y_{\alpha,\beta}^{(k)}(\theta) - \frac{1}{2}\text{Var}(Y^{(k)}_{\alpha,\beta}(\theta))} \text{d}\theta. 
\end{split}
\end{align}
For any measurable $A \subset [0,2\pi)$ and $k \in \mathbb{N}$ it holds by Fubini that 
\begin{align} \label{eqn:expectation of mu k}
\mathbb{E}\left( \mu_{\alpha,\beta}^{(k)}(A) \right) = \int_0^{2\pi} 1_A \text{d}\theta,
\end{align}
and 
\begin{align}
\mathbb{E}\left( \mu_{\alpha,\beta}^{(k)}(A) | \sigma(\mathcal{N}_1,...,\mathcal{N}_{k-1}) \right) = \mu_{\alpha,\beta}^{(k-1)}(A),
\end{align} 
where $\sigma \left( \mathcal{N}_1,...,\mathcal{N}_{k-1} \right)$ is the $\sigma$-algebra generated by $\mathcal{N}_1,...,\mathcal{N}_{k-1}$. Thus, since also for any $k \in \mathbb{N}$ the random variable $\mu_{\alpha,\beta}^{(k)}(A)$ is measurable w.r.t. $\sigma \left( \mathcal{N}_1,...,\mathcal{N}_{k} \right)$, the sequence $\left(\mu_{\alpha,\beta}^{(k)}(A) \right)_{k \in \mathbb{N}}$ is a martingale. Since the sequence is non-negative it converges a.s. to a random variable which will be denoted by $\mu_{\alpha,\beta}(A)$. One can show that a.s. the map $A \mapsto \mu_{\alpha,\beta}(A)$ is a measure and we have that a.s. $\mu^{(k)}_{\alpha,\beta} \xrightarrow{d} \mu_{\alpha,\beta}$ in the space of Radon measures on $S^1 \sim [0,2\pi)$, equipped with the topology of weak convergence.  \\

Recall that $I$ denotes either $I_\epsilon = (\epsilon, \pi - \epsilon) \cup (\pi + \epsilon, 2\pi - \epsilon)$ or $[0,2\pi)$. For any measurable $A \subset I$ the martingale $\left( \mu_{\alpha,\beta}^{(k)}(A)\right)_{k \in \mathbb{N}}$ is bounded in $L^2$ (and thus uniformly integrable), since by (\ref{eqn:Cov finite}), (\ref{eqn:bounded in L2 - 1}) and (\ref{eqn:bounded in L2 - 2}) we see that 
\begin{align}
\begin{split}
&\limsup_{k \rightarrow \infty} \mathbb{E}\left( \mu_{\alpha,\beta}^{(k)}(A)^2 \right) = \limsup_{k \rightarrow \infty} \int_{A} \int_{A} e^{\text{Cov}(Y_{\alpha,\beta}^{(k)}(\theta), Y_{\alpha,\beta}^{(k)}(\theta'))} \text{d}\theta \text{d}\theta' \\
\leq & \int_{A} \int_{A} |e^{i\theta} - e^{i\theta'}|^{-2(\alpha^2-\beta^2)} |e^{i\theta} - e^{-i\theta'}|^{-2(\alpha^2+\beta^2)} e^{4i\alpha \beta \Im \ln (1-e^{i(\theta+\theta')})} \text{d}\theta \text{d}\theta' < \infty, 
\end{split}
\end{align}
for $\alpha > - 1/2$ and $\alpha^2 - \beta^2 < 1/2$ in the case $I = I_\epsilon$, and $\alpha \geq 0$, $\alpha^2 - \beta^2 < 1/2$ and $4\alpha^2 < 1$ in the case $I = [0,2\pi)$. Thus, by (\ref{eqn:expectation of mu k}), we get that
\begin{align}
\mathbb{E}\left( \mu_{\alpha,\beta}(A) \right) = \lim_{k \rightarrow \infty}\mathbb{E}\left( \mu_{\alpha,\beta}^{(k)}(A) \right) = \int_0^{2\pi} 1_A \text{d}\theta,
\end{align}
which implies that the event $\{\mu_{\alpha,\beta} \text{ is the } 0 \text{ measure} \}$ does not have probability $1$. Since that event is independent of any finite number of the $\mathcal{N}_{k}$, $k \in \mathbb{N}$, Kolmogorov's zero-one law implies that this event then has probability $0$. Thus $\mu_{\alpha,\beta}$ is a.s. non-trivial for $\alpha \geq 0$, $\alpha^2 - \beta^2 < 1/2$ and $4\alpha^2 < 1$, while when restricted to $I_\epsilon$ the measure $\mu_{\alpha,\beta}$ is almost surely non-trivial for $\alpha > - 1/4$ and $\alpha^2 - \beta^2 < 1/2$. In both cases though, one can expect $\mu_{\alpha,\beta}$ to be a.s. non-trivial for a larger set of values $\alpha,\beta$.

\section{Riemann-Hilbert Problem for $\Psi$} \label{appendix:Psi}
This appendix is a mostly verbatim transfer from the beginning of Section 3 of \cite{ClaeysKrasovsky}. We include it here to make our account self-contained.  We use $\Psi$ to construct local parametrices for the RHP for the orthogonal polynomials in Section \ref{section:local 1}. We always assume that $\alpha_1,\alpha_2 > -\frac{1}{2}$ and $\beta_1,\beta_2 \in i \mathbb{R}$ (in \cite{ClaeysKrasovsky} also the more general case of $\alpha_1,\alpha_2,\beta_1,\beta_2 \in \mathbb{C}$ was considered). \\

\noindent \textbf{RH Problem for $\Psi$}
\begin{enumerate}[label=(\alph*)]
\item $\Psi:\mathbb{C} \setminus \Gamma \rightarrow \mathbb{C}^{2\times 2}$ is analytic, where 
\begin{align*}
\Gamma =& \cup_{k = 1}^7 \Gamma_k,  &&\Gamma_1 = i + e^{\frac{i\pi}{4}} \mathbb{R}_+, &&\Gamma_2 = i + e^{\frac{3i\pi}{4}\mathbb{R}_+} \\
\Gamma_3 =& -i + e^{\frac{5i\pi}{4}} \mathbb{R}_+, &&\Gamma_4 = -i + e^{\frac{7i\pi}4} \mathbb{R}_+, &&\Gamma_5 = -i + \mathbb{R}_+, \\
\Gamma_6 =& i + \mathbb{R}_+, &&\Gamma_7 = [-i,i], 
\end{align*}
with the orientation chosen as in Figure \ref{figure:Psi} ("-" is always on the RHS of the contour).
\item $\Psi$ satisfies the jump conditions
\begin{equation}
\Psi(\zeta)_+ = \Psi(\zeta)_- J_k, \quad \zeta \in \Gamma_k,
\end{equation}
where 
\begin{align}
J_1 =& \left( \begin{array}{cc} 1 & e^{2\pi i(\alpha_2 - \beta_2)} \\ 0 & 1 \end{array} \right), && J_2 = \left( \begin{array}{cc} 1 & 0 \\ - e^{-2\pi i(\alpha_2 - \beta_2)} & 1 \end{array} \right), \\
J_3 =& \left( \begin{array}{cc} 1 & 0 \\ -e^{2\pi i (\alpha_1 - \beta_1)} & 1 \end{array} \right), && J_4 = \left( \begin{array}{cc} 1 & e^{-2\pi i(\alpha_1 - \beta_1)} \\ 0 & 1 \end{array} \right) \\
J_5 =& e^{2\pi i \beta_1 \sigma_3}, && J_6 = e^{2\pi i \beta_2 \sigma_3}, \\
J_7 =& \left( \begin{array}{cc} 0 & 1 \\ -1 & 1 \end{array} \right).
\end{align}

\item We have in all regions:
\begin{equation} \label{eqn:Psi asymptotics}
\Psi(\zeta) = \left( I + \frac{\Psi_1}{\zeta} + \frac{\Psi_2}{\zeta^2} + O(\zeta^{-3}) \right) P^{(\infty)} (\zeta) e^{-\frac{is}{4} \zeta \sigma_3}, \quad \text{as } \zeta \rightarrow \infty,
\end{equation}
where 
\begin{equation}
P^{(\infty)}(\zeta) = \left( \frac{is}{2} \right)^{-(\beta_1 + \beta_2)\sigma_3} (\zeta - i)^{-\beta_2 \sigma_3} (\zeta + i)^{-\beta_1 \sigma_3},
\end{equation}
with the branches corresponding to the arguments between $0$ and $2\pi$, and where $s \in -i\mathbb{R}_+$.

\item The functions $F_1$ and $F_2$ defined in (\ref{eqn:F_1 neq 0}), (\ref{eqn:F_2 neq 0}), (\ref{eqn:F_1 0}) and (\ref{eqn:F_2 0}) below are analytic functions of $\zeta$ at $i$ and $-i$ respectively.
\end{enumerate}

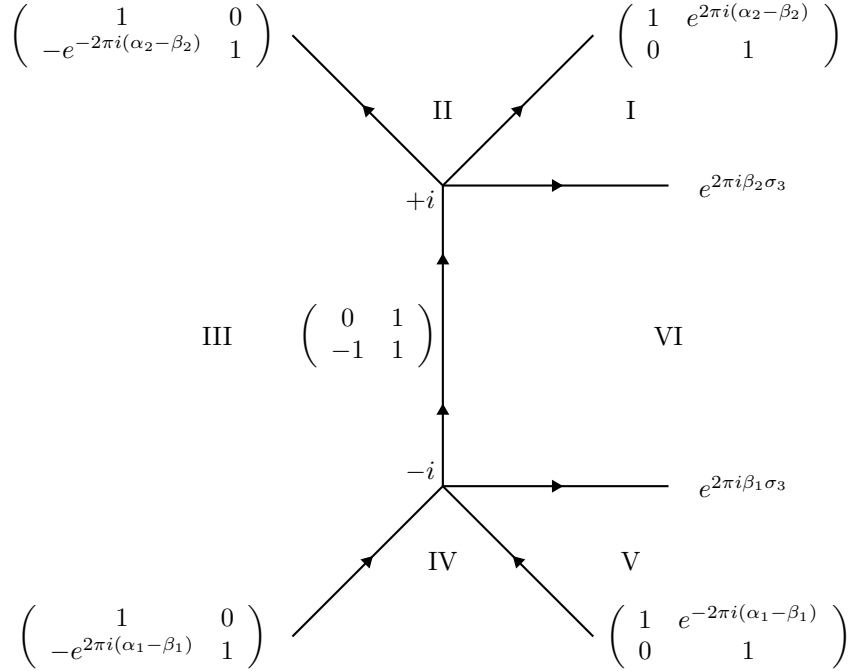
\begin{figure} [H]
\centering
\begin{tikzpicture}
\draw [thick] (0,-2) -- (0,2);
\draw [thick] (0,-2) -- (3,-2);
\draw [thick] (0,2) -- (3,2);
\draw [thick] (0,-2) -- (2,-4);
\draw [thick] (0,-2) -- (-2,-4);
\draw [thick] (0,2) -- (2,4);
\draw [thick] (0,2) -- (-2,4);

\fill (-0.3,1.8) node[] {$+i$};
\fill (-0.3,-1.8) node[] {$-i$};
\fill (-1,0) node[] {$\left( \begin{array}{cc} 0 & 1 \\ -1 & 1 \end{array} \right)$};
\fill (3.8,4) node[] {$\left( \begin{array}{cc} 1 & e^{2\pi i(\alpha_2 - \beta_2)} \\ 0 & 1 \end{array} \right)$};
\fill (-4,4) node[] {$\left( \begin{array}{cc} 1 & 0 \\ -e^{-2\pi i(\alpha_2 - \beta_2)} & 1 \end{array} \right)$};
\fill (4,2) node[] {$e^{2\pi i\beta_2 \sigma_3}$};
\fill (4,-2) node[] {$e^{2\pi i\beta_1 \sigma_3}$};
\fill (3.8,-4) node[] {$\left( \begin{array}{cc} 1 & e^{-2\pi i(\alpha_1 - \beta_1)} \\ 0 & 1 \end{array} \right)$};
\fill (-4,-4) node[] {$\left( \begin{array}{cc} 1 & 0 \\ -e^{2\pi i(\alpha_1 - \beta_1)} & 1 \end{array} \right)$};

\node[fill=black,regular polygon, regular polygon sides=3,inner sep=1.pt, shape border rotate = -45] at (1,3) {};
\node[fill=black,regular polygon, regular polygon sides=3,inner sep=1.pt, shape border rotate = 45] at (-1,3) {};
\node[fill=black,regular polygon, regular polygon sides=3,inner sep=1.pt, shape border rotate = 0] at (0,1) {};
\node[fill=black,regular polygon, regular polygon sides=3,inner sep=1.pt, shape border rotate = 0] at (0,-1) {};
\node[fill=black,regular polygon, regular polygon sides=3,inner sep=1.pt, shape border rotate = 45] at (1,-3) {};
\node[fill=black,regular polygon, regular polygon sides=3,inner sep=1.pt, shape border rotate = -45] at (-1,-3) {};
\node[fill=black,regular polygon, regular polygon sides=3,inner sep=1.pt, shape border rotate = -90] at (1.5,2) {};
\node[fill=black,regular polygon, regular polygon sides=3,inner sep=1.pt, shape border rotate = -90] at (1.5,-2) {};

\fill (0,3) node[] {II};
\fill (0,-3) node[] {IV};
\fill (3,0) node[] {VI};
\fill (-3,0) node[] {III};
\fill (2.5,3) node[] {I};
\fill (2.5,-3) node[] {V};
\end{tikzpicture}
\caption{The jump contour and jump matrices of $\Psi$.} \label{figure:Psi}
\end{figure}

The solution $\Psi = \Psi(\zeta;s)$ to this RHP not only depends on the complex variable $\zeta$, but also on the complex parameter $s \in -i\mathbb{R}_+$. Without the additional condition (d) on the behaviour of $\Psi$ near the points $\pm i$, the RHP wouldn't have a unique solution. If $2 \alpha_2 \notin \mathbb{N} \cup \{0\}$, define $F_1(\zeta,s)$ by the equations
\begin{equation} \label{eqn:F_1 neq 0}
\Psi(\zeta;s) = F_1(\zeta,s)(\zeta - i)^{\alpha_2\sigma_3} G_j, \quad \zeta \in \text{region } j,
\end{equation}
where $j \in \{I,II,III,VI\}$, and where $(\zeta - i)^{\alpha_2 \sigma_3}$ is taken with the branch cut on $i + e^{\frac{3\pi i}{4}} \mathbb{R}_+$, with the argument of $\zeta - i$ between $-5\pi/4$ and $3\pi/4$. The matrices $G_j$ are piecewise constant matrices consistent with the jump relations; they are given by 
\begin{align}
\begin{split}
G_{III} =& \left( \begin{array}{cc} 1 & g \\ 0 & 1 \end{array} \right), \quad g = - \frac{1}{2 i \sin (2\alpha_2)}(e^{2\pi i\alpha_2} - e^{-2\pi i \beta_2}),\\
G_{VI} =& G_{III} J_7^{-1}, \quad G_I = G_{VI}J_6, \quad G_{II} = G_I J_1.
\end{split}
\end{align}
It is straighforward to verify that $F_1$ has no jumps near $i$, and it is thus meromorphic in a neighborhood of $i$, with possibly an isolated singularity at $i$. 

Similarly, for $\zeta$ near $-i$, if $2\alpha_1 \notin \mathbb{N} \cup \{0\}$, we define $F_2$ by the equations
\begin{equation} \label{eqn:F_2 neq 0}
\Psi(\zeta;s) = F_2(\zeta,s)(\zeta + i)^{\alpha_1\sigma_3} H_j, \quad \zeta \in \text{region } j,
\end{equation}
where $j \in \{III,,IV,V,VI\}$, where $(\zeta + i)^{\alpha_1 \sigma_3}$ is taken with the branch cut on $-i + e^{\frac{5\pi i}{4}} \mathbb{R}_+$, with the argument of $\zeta + i$ between $-3\pi/4$ and $5\pi/4$, and where the matrices $H_j$ are piecewise constant matrices consistent with the jump relations; they are given by 
\begin{align} \label{eqn:H_j}
\begin{split}
H_{III} =& \left( \begin{array}{cc} 1 & h \\ 0 & 1 \end{array} \right), \quad h = - \frac{1}{2 i \sin (2\alpha_1)}(e^{2\pi i\beta_1} - e^{-2\pi i \alpha_1}),\\
H_{IV} =& G_{III} J_3^{-1}, \quad H_V = G_{IV}J_4^{-1}, \quad H_{VI} = H_V J_5.
\end{split}
\end{align}
Similarly as at $i$, one shows using the jump conditions for $\Psi$ that $F_2$ is  meromorphic near $-i$ with a possible singularity at $-i$. \\

If $2\alpha_2 \in \mathbb{N} \cup \{0\}$, the constant $g$ and the matrices $G_j$ are ill-defined, and we need a different definition of $F_1$:
\begin{align} \label{eqn:F_1 0}
\Psi(\zeta;s) = F_1(\zeta;s) (\zeta - i)^{\alpha_2 \sigma_3} \left( \begin{array}{cc} 1 & g_{int} \ln (\zeta - i) \\ 0 & 1 \end{array} \right) G_j, \quad \zeta \in \text{region } j,
\end{align}
where 
\begin{align}
g_{int} = \frac{e^{-2\pi i \beta_2} - e^{2\pi i \alpha_2}}{2\pi i e^{2\pi i \alpha_2}},
\end{align}
and $G_{III} = I$, and the other $G_j$'s are defined as above by applying the appropriate jump conditions. Thus defined, $F_1$ has no jumps in a neighborhood of $i$. Similarly, if $2 \alpha_1 \in \mathbb{N} \cup \{0\}$, we define $F_2$ by the expression:
\begin{align} \label{eqn:F_2 0}
\Psi(\zeta;s) = F_2(\zeta;s) (\zeta + i)^{\alpha_1\sigma_3} \left( \begin{array}{cc} 1 & \frac{e^{-2\pi i \alpha_1} - e^{2\pi i \beta_1}}{2\pi i e^{2\pi i \alpha_1}} \\ 0 & 1 \end{array} \right) H_j, \quad \zeta \in \text{region } j, 
\end{align}
with $H_{III} = I$, and the other $H_j$'s expressed via $H_{III}$ as in (\ref{eqn:H_j}). Then $F_2$ has no jumps near $-i$. \\

Given parameters $s,\alpha_1,\alpha_2,\beta_1,\beta_2$, the uniqueness of the function $\Psi$ which satisfies RH conditions (a) - (d) can be proven by standard arguments. \\

In Section 3 of \cite{ClaeysKrasovsky} it was shown that for $\alpha_1,\alpha_2, \alpha_1 + \alpha_2 > -\frac{1}{2}$ and $\beta_1,\beta_2 \in i\mathbb{R}$, the RHP is solvable for any $s \in -i\mathbb{R}_+$. Furthermore they analyzed the RHP asymptotically as $s \rightarrow -i\infty$ and $s \rightarrow -i0_+$. 

\begin{remark}
We have to be careful what their $(\alpha_1,\alpha_2,\beta_1,\beta_2)$ correspond to in our case, when using $\Psi$ from \cite{ClaeysKrasovsky}. In their paper $\alpha_1,\beta_1$ correspond to the singularity left of the merging point and $\alpha_2,\beta_2$ correspond to the singularity to the right of the merging point, while for us in the case + for example, $\alpha_1, \beta_1$ are right and $\alpha_2,\beta_2$ are left.
\end{remark}

\section{Riemann-Hilbert Problem for $M$} \label{appendix:M}
This appendix is a mostly verbatim transfer of Section 4 of \cite{ClaeysKrasovsky}. We include it here to make our account self-contained.  Let $\alpha > -\frac{1}{2}$ and $\beta \in i\mathbb{R}$. In Section 4.2.1 of \cite{ClaeysItsKrasovsky}, see also \cite{DeiftItsKrasovsky, ItsKrasovsky, MorenoMartinez-Finkelshtein}, a function $M = M^{(\alpha,\beta)}$ was constructed explicitly in terms of the confluent hypergeometric function, which solves the following RH problem:\\

\noindent \textbf{RH Problem for $M$}
\begin{enumerate}[label=(\alph*)]
\item $M: \mathbb{C}\setminus \left( e^{\pm \frac{\pi i }{4}} \mathbb{R} \cup \mathbb{R}_+ \right) \rightarrow \mathbb{C}^{2 \times 2}$ is analytic, 
\item $M$ has continuous boundary values on $e^{\pm \frac{\pi i }{4}} \mathbb{R} \cup \mathbb{R}_+ \setminus \{0\}$ related by the conditions:
\begin{align}
M(\lambda)_+ =& M(\lambda)_- \left( \begin{array}{cc} 1 & e^{\pi i (\alpha - \beta)} \\ 0 & 1 \end{array} \right), && \lambda \in e^{\frac{i\pi}{4}} \mathbb{R}_+,\\
M(\lambda)_+ =& M(\lambda)_- \left( \begin{array}{cc} 1 & 0 \\ -e^{-\pi i (\alpha - \beta)} & 0 \end{array} \right), && \lambda \in e^{\frac{3i\pi}{4}} \mathbb{R}_+,\\
M(\lambda)_+ =& M(\lambda)_- \left( \begin{array}{cc} 1 & 0 \\ e^{\pi i (\alpha - \beta)} & 0 \end{array} \right), && \lambda \in e^{\frac{5i\pi}{4}} \mathbb{R}_+,\\
M(\lambda)_+ =& M(\lambda)_- \left( \begin{array}{cc} 1 & -e^{-\pi i (\alpha - \beta)} \\ 0 & 1 \end{array} \right), && \lambda \in e^{\frac{7i\pi}{4}} \mathbb{R}_+,\\
M(\lambda)_+ =& M(\lambda)_-(\lambda) e^{2\pi i \beta \sigma_3}, && \lambda \in \mathbb{R}_+,
\end{align}
where all the rays of the jump contour are oriented away from the origin. 

\item Furthermore, in all sectors,
\begin{equation} \label{eqn:M asymptotics}
M(\lambda) = (I + M_1 \lambda^{-1} + O(\lambda^{-2})) \lambda^{-\beta \sigma_3} e^{-\frac{1}{2} \lambda \sigma_3}, \quad \text{as } \lambda \rightarrow \infty,
\end{equation}
where $0 < \arg \lambda < 2\pi$, and 
\begin{equation}
M_1 = M_1^{(\alpha,\beta)} = \left( \begin{array}{cc} \alpha^2 - \beta^2 & - e^{-2\pi i \beta} \frac{\Gamma(1 + \alpha - \beta)}{\Gamma(\alpha + \beta)} \\ e^{2\pi i \beta} \frac{\Gamma(1+\alpha+\beta)}{\Gamma(\alpha - \beta)} & - \alpha^2 + \beta^2 \end{array} \right).
\end{equation}
\end{enumerate}

\begin{figure} [H]
\centering
\begin{tikzpicture}
\draw [thick] (0,0) -- (3,0);
\draw [thick] (0,0) -- (2,-2);
\draw [thick] (0,0) -- (-2,-2);
\draw [thick] (0,0) -- (2,2);
\draw [thick] (0,0) -- (-2,2);

\fill (-0.3,0) node[] {$0$};
\fill (3.5,2) node[] {$\left( \begin{array}{cc} 1 & e^{\pi i(\alpha - \beta)} \\ 0 & 1 \end{array} \right)$};
\fill (-3.8,2) node[] {$\left( \begin{array}{cc} 1 & 0 \\ -e^{-\pi i(\alpha - \beta)} & 1 \end{array} \right)$};
\fill (3.7,0) node[] {$e^{2\pi i\beta \sigma_3}$};
\fill (3.8,-2) node[] {$\left( \begin{array}{cc} 1 & -e^{-\pi i(\alpha - \beta)} \\ 0 & 1 \end{array} \right)$};
\fill (-3.8,-2) node[] {$\left( \begin{array}{cc} 1 & 0 \\ e^{\pi i(\alpha - \beta)} & 1 \end{array} \right)$};

\node[fill=black,regular polygon, regular polygon sides=3,inner sep=1.pt, shape border rotate = -45] at (1,1) {};
\node[fill=black,regular polygon, regular polygon sides=3,inner sep=1.pt, shape border rotate = 45] at (-1,1) {};
\node[fill=black,regular polygon, regular polygon sides=3,inner sep=1.pt, shape border rotate = 45] at (1,-1) {};
\node[fill=black,regular polygon, regular polygon sides=3,inner sep=1.pt, shape border rotate = -45] at (-1,-1) {};
\node[fill=black,regular polygon, regular polygon sides=3,inner sep=1.pt, shape border rotate = -90] at (1.5,0) {};

\fill (0,2) node[] {2};
\fill (0,-2) node[] {4};
\fill (-2,0) node[] {3};
\fill (2.5,1) node[] {1};
\fill (2.5,-1) node[] {5};
\end{tikzpicture}
\caption{The jump contour and jump matrices of $M$.} \label{figure:M}
\end{figure}
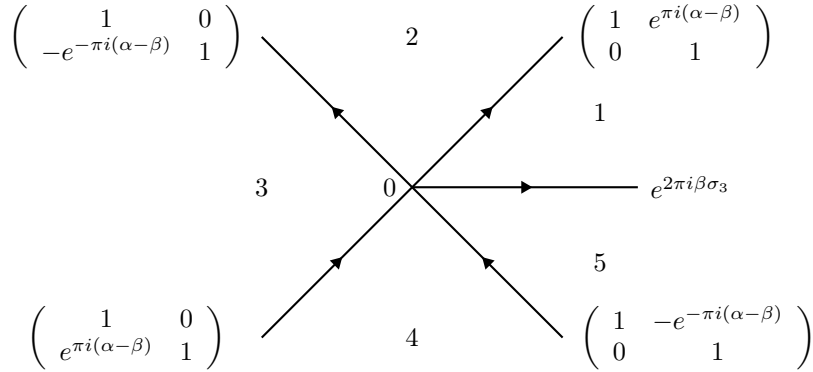

We use $M$ to construct local parametrices for the RHP of the orthogonal polynomials in Section \ref{section:local 2}. For that we also need the local behaviour of $M$ at zero in region 3, i.e. the region between the lines $e^{\frac{3\pi i}{4}} \mathbb{R}_+$ and $e^{\frac{5\pi i}{4}} \mathbb{R}_+$. Write $M \equiv M^{(3)}$ in this region. It is known (see Section 4.2.1 of \cite{ClaeysItsKrasovsky}) that $M^{(3)}$ can be written in the form 
\begin{align} \label{eqn:M at 0, neq 0}
M^{(3)}(\lambda) = L(\lambda)\lambda^{\alpha\sigma_3} \tilde{G}_3, \quad 2\alpha \notin \mathbb{N} \cup \{0\},
\end{align}
with the branch of $\lambda^{\pm \alpha}$ chosen with $0 < \arg \lambda < 2\pi$. Here 
\begin{align}
L(\lambda) = e^{-\lambda/2} \left( \begin{array}{cc} e^{-i\pi(\alpha+\beta)} \frac{\Gamma(1+\alpha - \beta)}{\Gamma(1+2\alpha)} \varphi(\alpha+\beta,1+2\alpha,\lambda) & e^{i\pi(\alpha-\beta)} \frac{\Gamma(2\alpha)}{\Gamma(\alpha+\beta)} \varphi(-\alpha+\beta,1-2\alpha,\lambda) \\
-e^{-i\pi(\alpha-\beta)} \frac{\Gamma(1+\alpha + \beta)}{\Gamma(1+2\alpha)} \varphi(1+\alpha+\beta,1+2\alpha,\lambda) & e^{i\pi(\alpha+\beta)} \frac{\Gamma(2\alpha)}{\Gamma(\alpha-\beta)} \varphi(1-\alpha+\beta,1-2\alpha,\lambda) \\
& \end{array} \right)
\end{align}
is an entire function, with
\begin{align}
\varphi(a,b;z) = 1 + \sum_{n = 1}^\infty \frac{a(a+1) \cdots (a + n - 1)}{c(c+1) \cdots (c + n - 1)} \frac{z^n}{n!}, \quad c \neq 0,-1,-2,...,
\end{align}
and $\tilde{G}_3$ is the constant matrix 
\begin{align}
\tilde{G}_3 = \left( \begin{array}{cc} 1 & \tilde{g} \\ 0 & 1 \end{array} \right), \quad \tilde{g} = - \frac{\sin \pi(\alpha+\beta)}{\sin 2\pi \alpha}.
\end{align}
If $2 \alpha$ is an integer, we have 
\begin{align} \label{eqn:M at 0, 0}
M^{(3)}(\lambda) =& \tilde{L}(\lambda) \lambda^{\alpha \sigma_3} \left( \begin{array}{cc} 1 & m(\lambda) \\ 0 & 1 \end{array} \right), \\
m(\lambda) =& \frac{(-1)^{2\alpha + 1}}{\pi} \sin \pi(\alpha + \beta) \ln (\lambda e^{-i\pi}), 
\end{align}
where $\tilde{L}(\lambda)$ is analytic at $0$, and the branch of the logarithm corresponds to the argument of $\lambda$ between $0$ and $2\pi$. 
\end{appendix}

\end{document}